\documentclass[a4paper,11pt,pointlessnumbers]{scrartcl}

\usepackage[utf8]{inputenc}
\usepackage{amsthm}
\usepackage{amssymb}
\usepackage{enumerate}
\usepackage{mathtools}

\usepackage{fullpage}

\usepackage{lmodern}
\usepackage{MnSymbol}
\usepackage{booktabs}
\usepackage{tabularx}

\usepackage{hyperref}
\usepackage{cleveref}

\usepackage{tikz}
\usetikzlibrary{shapes.geometric}
\usetikzlibrary{arrows.meta}

\newcommand{\problem}[6][ ]{
\begin{prob}
Compute a function $#2$ with the following properties:\\
\begin{tabular}{@{}p{0.125\linewidth}@{}p{0.875\linewidth}@{}}
Input & #3.\\
Output & A labeling coset $#2#4=\Lambda\leq\lab(V)$ such that:\\
(CL1) & $#2#4=\phi#2#5$ for all $\phi\in\iso(V;V')$.\\
(CL2) & $#2#4=#6\pi$#1for some (and thus for all) $\pi\in\Lambda$.
\end{tabular}
\end{prob}
}

\theoremstyle{plain}
\newtheorem{theo}{Theorem}
\newtheorem*{theo*}{Theorem}
\crefname{theo}{Theorem}{Theorems}
\newtheorem{lem}[theo]{Lemma}
\crefname{lem}{Lemma}{Lemmata}
\newtheorem{prop}[theo]{Proposition}
\crefname{prop}{Proposition}{Propositions}
\newtheorem{cor}[theo]{Corollary}
\newtheorem*{cor*}{Corollary}
\crefname{cor}{Corollary}{Corollarys}

\theoremstyle{definition}
\newtheorem{defn}[theo]{Definition}
\newtheorem{exa}[theo]{Example}
\crefname{exa}{Example}{Examples}
\newtheorem{prob}[theo]{Problem}
\crefname{prob}{Problem}{Problems}

\theoremstyle{remark}

\crefname{section}{Section}{Sections}
\crefname{appendix}{Appendix}{Appendices}
\crefname{figure}{Figure}{Figures}
\crefname{equation}{}{}

\newenvironment{cl1}{\vspace{0.1cm}\noindent(\textit{CL1.})}{}
\newenvironment{cl2}{\vspace{0.1cm}\noindent(\textit{CL2.})}{}

\newenvironment{runtime}{\vspace{0.1cm}\noindent(\textit{Running time.})}{}
\newcommand{\case}[1]{\item[\itshape\mdseries If {#1}:]}
\newenvironment{cs}{\begin{description}}{\end{description}}
\newenvironment{for}[1]{\begin{description}\item[\itshape\mdseries For each {#1} do:]~\\}{\end{description}}
\newcommand{\algorithmDAN}[1]{\vspace{0.1cm}\noindent
\underline{An algorithm for $#1$:}}

\newcommand{\comment}[1]{\textcolor{gray}{$\vartriangleright$}~\begin{minipage}[t]{\linewidth-0.38cm}\emph{#1.}\end{minipage}~}

\definecolor{myBlue}{rgb}{0.0, 0.5, 1.0}
\definecolor{myRed}{rgb}{1.0, 0.5, 0}
\definecolor{myGreen}{rgb}{0.0, 0.5, 0.0}
\definecolor{myYellow}{rgb}{1.0, 1.0, 0.0}

\newcommand{\RN}[1]{
  \textup{\uppercase\expandafter{\romannumeral#1}}%
}

\newcommand{\xx}{\mathfrak{x}}

\renewcommand{\phi}{\varphi}
\renewcommand{\epsilon}{\varepsilon}

\newcommand{\NN}{{\mathbb N}}

\newcommand{\QQ}{{\mathbb Q}}

\newcommand{\PP}{\mathbb{P}}
\newcommand{\C}{\mathbb{C}}
\newcommand{\JJ}{\mathbb{J}}

\newcommand{\LA}{\textsf{\upshape A}}

\newcommand{\CB}{{\mathcal B}}
\newcommand{\CC}{{\mathcal C}}

\newcommand{\CH}{{\mathcal H}}

\newcommand{\CJ}{{\mathcal J}}

\newcommand{\CL}{{\mathcal L}}

\newcommand{\CO}{{\mathcal O}}
\newcommand{\CP}{{\mathcal P}}
\newcommand{\CQ}{{\mathcal Q}}

\newcommand{\CS}{{\mathcal S}}

\newcommand{\CX}{{\mathcal X}}

\let\downarrowOld\downarrow
\renewcommand{\downarrow}{\mathord{\downarrowOld}}
\let\uparrowOld\uparrow
\renewcommand{\uparrow}{\mathord{\uparrowOld}}

\renewcommand{\hat}{\widehat}
\renewcommand{\tilde}{\widetilde}
\renewcommand{\hat}{\widehat}
\renewcommand{\bar}{\overline}

\newcommand{\sym}{\operatorname{Sym}}
\newcommand{\alt}{\operatorname{Alt}}
\newcommand{\aut}{\operatorname{Aut}}
\newcommand{\iso}{\operatorname{Iso}}
\newcommand{\stab}{\operatorname{Stab}}
\newcommand{\lab}{{\operatorname{Label}}}
\newcommand{\tranCl}{\operatorname{TClosure}}
\renewcommand{\ker}{\operatorname{ker}}

\newcommand{\polylog}{\operatorname{polylog}}
\newcommand{\tw}{\operatorname{tw}}
\newcommand{\cw}{\operatorname{cw}}
\newcommand{\sep}{\operatorname{sep}}
\newcommand{\mindeg}{\operatorname{minDeg}}

\newcommand{\pid}{{\operatorname{rb}}}
\newcommand{\pidb}{{\operatorname{b}}}
\newcommand{\ind}{{\operatorname{ind}}}
\newcommand{\orb}{{\operatorname{orb}}}

\newcommand{\set}{{\operatorname{Set}}}
\newcommand{\obj}{\operatorname{Objects}}
\newcommand{\hfs}{\operatorname{HFS}}
\newcommand{\ord}{{\operatorname{Can}}}

\newcommand{\recPart}{\textsc{recurseOnPartition}}
\newcommand{\toGroup}{\textsc{reduceToSubgroup}}
\newcommand{\toJohnson}{\textsc{reduceToJohnson}}
\newcommand{\processJohnson}{\textsc{produceCertificates}}
\newcommand{\processAut}{\textsc{aggregateCertificates}}

\newcommand{\isoBasic}{\operatorname{Iso}_{\operatorname{Basic}}}
\newcommand{\isoTree}{\operatorname{Iso}_{\operatorname{Tree}}}
\newcommand{\can}{\operatorname{CL}}
\newcommand{\canInt}{\operatorname{CL}_{\operatorname{Int}}}
\newcommand{\canSet}{\operatorname{CL}_{\operatorname{Set}}}
\newcommand{\canSetHyper}{\operatorname{CL}_{\operatorname{SetHyper}}}
\newcommand{\canHyper}{\operatorname{CL}_{\operatorname{Hyper}}}
\newcommand{\canSetSet}{\operatorname{CL}_{\operatorname{SetSet}}}
\newcommand{\canObj}{\operatorname{CL}_{\operatorname{Object}}}
\newcommand{\canGraph}{\operatorname{CL}_{\operatorname{Graph}}}
\newcommand{\canRel}{\operatorname{CL}_{\operatorname{Rel}}}

\begin{document}

\title{Graph isomorphism in quasipolynomial time parameterized by treewidth}

\author{
Daniel Wiebking\\
RWTH Aachen University\\
\texttt{wiebking@informatik.rwth-aachen.de}}

\date{}

\KOMAoptions{abstract=true}

\maketitle

\begin{abstract}
We extend Babai's quasipolynomial-time graph isomorphism test (STOC 2016)
and develop a quasipolynomial-time algorithm for the multiple-coset isomorphism problem.
The algorithm for the multiple-coset isomorphism problem allows to
exploit graph decompositions of
the given input graphs within Babai's group-theoretic framework.

We use it to develop a graph isomorphism test that runs
in time $n^{\polylog(k)}$
where
$n$ is the number of vertices
and $k$ is the minimum treewidth of the given graphs
and $\polylog(k)$ is some polynomial in $\log(k)$.
Our result generalizes Babai's quasipolynomial-time graph isomorphism test.
\end{abstract}

\section{Introduction}

The graph isomorphism problem asks for a structure preserving bijection between two given graphs $G$ and $H$, i.e.,
a bijection $\phi:V(G)\to V(H)$ such that $vw\in E(G)$ if and only if $\phi(v)\phi(w)\in E(H)$.
One central open problem in theoretical computer science is the question whether
the graph isomorphism problem can be solved in polynomial time.
There are a few evidences that the problem might not be NP-hard.
For example, NP-hardness of the problem implies a collapse of the polynomial hierarchy
\cite{DBLP:journals/jcss/Schoning88}.
Moreover, NP-hardness of the graph isomorphism problem
would refute the exponential time hypothesis since
the problem can be decided in quasipolynomial time \cite{DBLP:conf/stoc/Babai16}.

The research of the graph isomorphism problem started with two fundamental graph classes,
i.e., the class of trees and the class of planar graphs.
In 1970, Zemlyachenko gave a polynomial-time isomorphism algorithm for trees \cite{zemlyachenko1970canonical}.
One year later, Hopcroft and Tarjan extended a result of Weinberg and designed a polynomial-time isomorphism algorithm for planar graphs 
\cite{DBLP:journals/ipl/HopcroftT71},\cite{weinberg1966simple}.
In 1980, Filotti, Mayer and Miller extended the polynomial-time algorithm to graphs of bounded genus
\cite{DBLP:conf/stoc/Miller80},\cite{DBLP:conf/stoc/FilottiM80}\footnote{Myrvold and Kocay pointed out an error in Filotti's techniques \cite{DBLP:journals/jcss/MyrvoldK11}.
However, different algorithms have been given which show that
the graph isomorphism problem for graphs of bounded genus is indeed decidable in polynomial time
\cite{DBLP:journals/iandc/Miller83a,DBLP:conf/stoc/Grohe00,DBLP:journals/corr/Kawarabayashi15a}.}.
The genus is a graph parameter that measures how far away the graph is from being planar.

In Luks's pioneering work in 1982, he gave a polynomial-time isomorphism algorithm
for graphs of bounded degree \cite{DBLP:journals/jcss/Luks82}.
His group-theoretic approach laid the foundation of many other algorithms that were developed ever since.
It turns out that the research in the graph isomorphism problem for restricted graph classes
was a promising approach in tackling the graph isomorphism problem in general.
Shortly after Luks's result,
a combinatorial partitioning lemma by Zemlyachenko was combined with Luks's framework.
This resulted in an isomorphism algorithm for graphs with $n$ vertices in general that runs in time $2^{\CO(\sqrt{n\log n})}$
\cite{zemlyachenko1985graph},\cite{DBLP:conf/stoc/BabaiL83}.
This algorithm was the fastest for decades.

In 1983, the seminal work of Robertson and Seymour in graph minors started a new era of graph theory \cite{DBLP:journals/jct/RobertsonS83}.
At the same time, Miller extended Luks's group-theoretic framework to hypergraphs
\cite{DBLP:journals/iandc/Miller83a}.
It turned out that the study of general structures such as hypergraphs
was also a promising approach in tackling the graph isomorphism problem.
In 1991, Ponomarenko could in fact use Miller's hypergraph algorithm
to design a polynomial-time isomorphism algorithm for graphs
excluding a minor \cite{Ponomarenko1991}.

The work of Robertson and Seymour also rediscovered the notion of treewidth \cite{DBLP:books/daglib/0030488}, a graph parameter
that measures how far away the graph is from being a tree.
The treewidth parameter was reborn and has been studied ever since.
So, researchers went back to the roots and studied the isomorphism problem for graphs of bounded treewidth.
In 1990, Bodlaender gave a simple isomorphism-algorithm for graphs of treewidth $k$ with $n$ vertices that runs in time $n^{\CO(k)}$ \cite{DBLP:journals/jal/Bodlaender90}.
However, no FPT-algorithm was known, i.e., an isomorphism algorithm with a running time of the form $f(k)\cdot n^{\CO(1)}$.
The search of a FPT-algorithm occupied researchers over years
and this open problem was explicitly stated by several authors
\cite{DBLP:journals/algorithmica/YamazakiBFT99,bodlaender2006open,DBLP:conf/stoc/KawarabayashiM08,
DBLP:conf/swat/KratschS10,DBLP:conf/isaac/Otachi12,DBLP:conf/iwpec/BoulandDK12,
DBLP:series/txcs/DowneyF13,DBLP:journals/siamcomp/GroheM15}.
In 2017, Lokshtanov, Pilipczuk, Pilipczuk
and Saurabh finally solved this open problem and designed a FPT-algorithm for the graph isomorphism problem \cite{DBLP:journals/siamcomp/LokshtanovPPS17}.
Their algorithm runs in time $2^{\CO(k^5\log k)}n^{\CO(1)}$ where
$n$ is the number of vertices and
$k$ is the minimum treewidth of the given graphs.

At the same time, Babai made a breakthrough and designed a quasipolynomial-time
algorithm for the graph isomorphism problem in general \cite{DBLP:conf/stoc/Babai16}.
His algorithm runs in time $n^{\polylog(n)}$ where $n$ is the number of vertices
and $\polylog(n)$ is some polynomial in $\log(n)$ (according to Helfgott's analysis
the function $\polylog(n)$ can chosen to be quadratic in $\log(n)$ \cite{helfgott2017isomorphismes}).
To achieve this result, Babai built on Luks's group-theoretic framework,
which actually solves the more general string isomorphism problem.
One of the main questions is how to combine Babai's group-theoretic algorithm
with the graph-theoretic techniques that have been developed.
For example, it is unclear how to exploit a decomposition of the given graphs
within Babai's framework since his algorithm actually processes strings rather than graphs.

Recently, Grohe, Neuen and Schweitzer were able to extend Babai's algorithm
to graphs of maximum degree $d$ and an isomorphism algorithm was developed that runs in time $n^{\polylog(d)}$ \cite{DBLP:conf/focs/GroheNS18}.
They suggest that their techniques might be useful
also for graphs parameterized by treewidth and
conjectured that the isomorphism problem for graphs of treewidth $k$ can be decided in time $n^{\polylog(k)}$.

In \cite{DBLP:conf/icalp/GroheNSW18},
the graph-theoretic FPT-algorithm of Lokshtanov et al.~was improved by
using Babai's group-theoretic algorithm and the extension
given by Grohe et al.~as a black box.
They decomposed a graph of bounded treewidth into subgraphs
with particular properties.
They were able to design a faster algorithm that computes the isomorphisms between these subgraphs.
However, they pointed out a central problem that arises when dealing with graph decompositions:
When the isomorphisms between these subgraphs are already computed,
how can they be efficiently merged in order to compute the isomorphisms
between the entire graphs?
This problem was named as \emph{multiple-coset isomorphism problem} and is formally defined as follows.
Given two sets $J=\{\rho_1\Delta_1^\ord,\ldots,\rho_t\Delta^\ord_t\}$ and $J'=\{\rho_1'{\Delta'}_1^\ord,\ldots,\rho_t'{\Delta'}_t^\ord\}$
where $\rho_i:V\to n,\rho_i':V'\to n$ are bijections and $\Delta_i^\ord,{\Delta'}_i^\ord\leq\sym([n])$
are permutation groups for all $i\in[t]$,
the problem is to decide whether there are bijections $\phi:V\to V',\pi:[t]\to [t]$
such that $\Delta_i^\ord={\Delta'}_{\pi(i)}^\ord$ and $\phi\in\rho_i\Delta_i^\ord\rho_{\pi(i)}'$ for all $i\in[t]$.
By applying the group-theoretic black box algorithms, they
achieved an improved isomorphism test
for graphs of treewidth $k$
that runs in time
$2^{k\cdot\polylog(k)}n^{\CO(1)}$.
However, for further improvements, it did not
seem to be enough to use the group-theoretic algorithms as a black box only.
The question of an isomorphism algorithm that runs in time $n^{\polylog(k)}$ remained open.

In \cite{DBLP:conf/stoc/SchweitzerW19},
the study of the multiple-coset isomorphism problem continued.
Rather than using group-theoretic algorithms as a black box,
they were able to extend Luks's group-theoretic framework to
the multiple-coset isomorphism problem.
In order to facilitate their recursion, they
introduced the class of combinatorial objects.
Their class of combinatorial objects contains hypergraphs, colored graphs, relational structures,
explicitly given codes and more.
However, the key idea in order to handle the involved structures recursively,
was to add so-called labeling cosets to their structures.
By doing so, they could combine combinatorial decomposition techniques
with Luks's group-theoretic framework.
This led to a simply-exponential time algorithm for the
multiple-coset isomorphism problem.
Although the achieved running time was far away from being
quasipolynomial, their result led to improvements of several algorithms.
For example, it led to the currently best algorithm for the normalizer problem
(a central problem in computational group theory) \cite{DBLP:conf/soda/Wiebking20}.
However, they were not able to extend also Babai's techniques to
their framework
and the question of a graph isomorphism algorithm
running in time $n^{\polylog(k)}$ remained open.

\paragraph{Our Contribution}

In this paper, we give a quasipolynomial-time algorithm for the
multiple-coset isomorphism problem.
This leads to an answer
of the conjecture in \cite{DBLP:conf/focs/GroheNS18} mentioned above.

\begin{theo*}[\cref{theo:isoGraph}]
The graph isomorphism problem can be decided in time $n^{\polylog(k)}$
where $n$ is the number of vertices and $k$ is the minimum treewidth of the input graphs.
\end{theo*}

When $k=\polylog(n)$, our algorithm runs in time $n^{\CO(\log(\log n)^c)}$ (for some constant $c$)
and is significantly faster
than Babai's algorithm and existing FPT-algorithms for graphs parameterized by treewidth.

For the present work,
we exploit the fact that Babai's algorithm was recently extended to canonization
\cite{DBLP:conf/stoc/Babai19}.
A canonical labeling of a graph is a function that labels the vertices $V$ of the graph with integers $1,\ldots,|V|$
in such a way that the labeled versions of two isomorphic graphs are equal (rather than isomorphic).
The computation of canonical forms and labelings, rather than isomorphism testing, is an important task in the
area of graph isomorphism and is especially useful for practical applications.
Also the framework given in \cite{DBLP:conf/stoc/SchweitzerW19}
is actually designed for the canonization problem.
The present paper is based on these works
and our algorithms provide canonical labelings as well.
Only the algorithm given in the last section depends on the bounded-degree isomorphism algorithm
of Grohe et al.~for which no adequate canonization version is known.

The first necessary algorithm that we provide in our work
is a simple canonization algorithm for hypergraphs.

\begin{theo*}[\cref{theo:canHyper}]
Canonical labelings for hypergraphs $(V,H)$ can be computed in time $(|V|+|H|)^{\polylog|V|}$.
\end{theo*}

There is a simple argument why this algorithm is indeed necessary for our main result.
It is well-known that a hypergraph $X=(V,H)$ can be encoded as a bipartite graph $G_X=(V\cupdot H,E)$
(the bipartite graph $G_X$ has an edge $(v,S)\in E$, if and only if $v\in S$).
It is not hard to show that the treewidth $k$ of this bipartite graph $G_X$ is at most $|V|$.
The bipartite graph $G_X$ uniquely encodes the hypergraph $X$,
in particular, two hypergraphs are isomorphic if and only if their corresponding bipartite graphs
are isomorphic.
This means that an isomorphism algorithm for graphs of treewidth $k$ running in time
$n^{\polylog(k)}$
would imply an isomorphism algorithm for hypergraphs running in time $(|V|+|H|)^{\polylog|V|}$.
However, applying Babai's algorithm to the bipartite graph would lead to a running time of
$(|V|+|H|)^{\polylog(|V|+|H|)}$.
Instead of applying Babai's algorithm to the bipartite graph
directly, we decompose the hypergraph and canonize the substructures recursively.
To merge
the canonical labelings of all subhypergraphs, we use a canonical
version of the multiple-coset isomorphism problem.
However, for the hypergraph algorithm, it suffices to use Babai's
algorithm as a black box only.

Our decomposition technique for hypergraphs can also be used to
design a simple
canonization algorithm for $k$-ary relations.

\begin{theo*}[\cref{theo:canRel}]
Canonical labelings for $k$-ary relations $R\subseteq V^k$ can be computed in time $2^{\polylog|V|}|R|^{\CO(1)}$.
\end{theo*}

The algorithm
improves the currently best algorithm from \cite{DBLP:conf/focs/GroheNS18}.
As graphs can be seen as binary relations, our algorithm
generalizes the quasipolynomial-time bound
for graphs.
The achieved running time is the best one can hope for as long as the
graph isomorphism problem has no solution better than quasipolynomial time.

Our main algorithm finally solves the multiple-coset isomorphism problem.
In fact, the algorithm computes canonical labelings as well.

\begin{theo*}[\cref{theo:canSet}]
Canonical labelings for a set $J$ consisting of labeling cosets can be computed in time $(|V|+|J|)^{\polylog|V|}$.
\end{theo*}

This result is actually of independent interest as it also implies
a faster canonization algorithm for the entire class of combinatorial objects (\cref{cor:canObj}).

To solve this problem, the simple hypergraph canonization algorithm can be used as a subroutine in some places.
However, we do not longer use Babai's and Luks's techniques as a black box only.
To extend their methods, we follow the route of \cite{DBLP:conf/stoc/SchweitzerW19}
and consider combinatorial objects
that allows to combine combinatorial structures with permutation group theory.
In particular, we can extend Luks's subgroup reduction and Babai's method and aggregation of local certificates to our framework.
All these methods were designed for the string isomorphism problem and need non-trivial extensions when dealing with a set of labeling cosets
rather than a string.

\paragraph{Related Work}
Another extension of Babai's quasipolynomial time algorithm has been
independently proposed by Daniel Neuen \cite{NeuenHyper}
who provided another algorithm for the isomorphism problem
of hypergraphs.
However, Neuen can exploit groups with restricted composition factors
that are given as additional input in order to speed up his algorithm.
This can be exploited in the setting of graphs of bounded Euler genus.
He provides a graph isomorphism algorithm that runs in time $n^{\polylog(g)}$
where $n$ is the number of vertices
and $g$ is the minimum genus of the given graphs.

On the other hand, his algorithm is not able to handle labeling cosets occurring in the combinatorial structures.
In particular, his algorithm is not able to solve the
multiple-coset isomorphism problem in the desired time bound,
which we require for
our isomorphism algorithm for graphs parameterized by treewidth.
Moreover, his techniques do not provide canonical labelings.

We hope that both algorithms can be combined
to give a faster isomorphism test for the large class of graphs
excluding a topological subgraph.
This large class of graphs
includes the graphs of bounded treewidth, graphs
of bounded genus, graphs of bounded degree and
graphs excluding a minor.
In fact,
Grohe and Marx provide a structure theorem
which shows that
the graph classes mentioned above
also characterize graphs excluding a topological subgraph.
Informally, they showed that
graphs
excluding a topological subgraph
can be decomposed into almost bounded-degree parts
and minor-free parts which in turn can be decomposed into almost-embeddable parts \cite{DBLP:journals/siamcomp/GroheM15}.
Therefore, we hope that the improved algorithms for
the isomorphism problem for bounded-degree graphs
and bounded-genus graphs can be combined with our algorithm
to exploit the occurring graph decomposition.

\paragraph{Organization of the Paper}

In \cref{sec:smallObj}, we show how the multiple-coset isomorphism problem
and its canonical version can be
reduced to a string canonization problem which in turn can be processed with Babai's algorithm.
However, this reduction does not lead to the desired time bound and only works efficiently
when the instance is small enough.
\Cref{sec:rel} deals with $k$-ary relations $R\subseteq V^k$ over a vertex set $V$.
We demonstrate how a partitioning technique can be used to
reduce the canonization problem of a $k$-ary relation to instances of small size in each decomposition level.
Since we only need to handle small instances at each decomposition level, we can make use of
the subroutines given in the previous section.
As a result, we obtain a canonization algorithm for $k$-ary relations that runs in time $2^{\polylog|V|}|R|^{\CO(1)}$.
In \cref{sec:hyper}, we extend our technique to hypergraphs and so-called coset-labeled hypergraphs.
The algorithm for coset-labeled hypergraphs is used as a subroutine
in our main algorithm given in the next section.
In \cref{sec:sets}, we finally present our main algorithm which canonizes a set of labeling cosets
and solves the multiple-coset isomorphism problem.
Our main algorithm is divided into five subroutines.
In the first subroutine, we extend the partitioning technique to families of partitions.
The second subroutine extends Luks's subgroup reduction to our framework.
The third subroutine reduces to the barrier configuration
which can be characterized by a giant representation.
The fourth and fifth subroutine extend Babai's method and aggregation of local certificates to our framework.
In \cref{sec:treewidth}, a straightforward application of the multiple-coset isomorphism problem
leads to an isomorphism algorithm that runs in time $n^{\polylog(k)}$
where $n$ is the number of vertices and $k$ is the treewidth of the given graphs.

\section{Preliminaries}
We recall the framework given in \cite{DBLP:conf/stoc/SchweitzerW19}.

\paragraph{Set Theory}
For an integer $t$, we write $[t]$ for $\{1,\ldots,t\}$.
For a set $S$ and an integer $k$, we write $\binom{S}{k}$ for the $k$-element
subsets of $S$ and $2^S$ for the power set of $S$.

\paragraph{Group Theory}
The composition of two functions $f:V\to U$ and $g:U\to W$ is denoted by
$fg$ and is defined as the function that first applies $f$ and then applies $g$.
The symmetric group on a set $V$ is denoted by $\sym(V)$
and the symmetric group of degree $t\in\NN$ is denoted by $\sym(t)$.
In the following, let $G\leq\sym(V)$ be a group.
The index of a subgroup $H\leq G$ is denoted by $(G:H)$.
The setwise
stabilizer of $A\subseteq V$ in $G$ is denoted by
$\stab_G(A):=\{g\in G\mid g(a)\in A
\text{ for all }a\in A\}$.
The pointwise stabilizer of $A\subseteq V$ in $G$
is denoted by $G_{(A)}:=\{g\in G\mid g(a)=a\text{ for all }a\in A\}$.
Analogously, the stabilizer of a vertex $v\in V$ in $G$
is denoted by $G_{(v)}:=G_{(\{v\})}$.
A set $A\subseteq V$ is called $G$-invariant
if $\stab_G(A)=G$.
A set $v^G:=\{g(v)\mid g\in G\}$ is called $G$-orbit of $v\in V$.
The $G$-orbit partition of $V$ is the partition of
$V$ in which each part is a $G$-orbit (for some $v\in V$).
partition.
A group $G\leq\sym(V)$ is called transitive if $V$ is one single $G$-orbit.
A partition of $V=V_1\cupdot\ldots\cupdot V_t$ is called $G$-invariant
if each part $V_i$ and each $g\in G$ it holds $V_i^g:=\{g(v)\mid v\in V_i\}\in G$.
For transitive groups $G$, the $G$-invariant partitions of $V$ are also
called block systems for $G$.
A group $G\leq\sym(V)$ is called primitive if there are no non-trivial block systems for $G$.

\paragraph{Labeling Cosets}
A \emph{labeling} of a set $V$ is a bijection $\rho:V\to\{1,\ldots,|V|\}$.
A \emph{labeling coset} of a set $V$ 
is a set of bijections $\Lambda$
such that
$\Lambda=\Delta\rho=\{\delta\rho\mid \delta\in\Delta\}$
for some subgroup $\Delta\leq\sym(V)$
and some labeling $\rho:V\to\{1,\ldots,|V|\}$.
We write $\lab(V)$ to denote
the labeling coset $\sym(V)\rho=
\{\sigma\rho\mid \sigma\in\sym(V)\}$
where $\rho$ is an arbitrary labeling of $V$.
Analogous to subgroups, a set $\Theta\tau$ is called a \emph{labeling subcoset}
of $\Delta\rho$, written $\Theta\tau\leq\Delta\rho$,
if the labeling coset $\Theta\tau$ is a subset of $\Delta\rho$.

\paragraph{Hereditarily Finite Sets and Combinatorial Objects}
Inductively, we define \emph{hereditarily finite sets}, denoted by $\hfs(V)$,
over a ground set $V$.
\begin{itemize}
  \item A vertex $v\in V$ is an atom and a hereditarily finite set $v\in\hfs(V)$,
  \item a labeling coset $\Delta\rho\leq\lab(V)$ is an atom and a hereditarily finite set $\Delta\rho\in\hfs(V)$,
  \item if $X_1,\ldots,X_t\in\hfs(V)$, then also $\CX=\{X_1,\ldots,X_t\}\in\hfs(V)$ where $t\in\NN\cup\{0\}$, and
  \item if $X_1,\ldots,X_t\in\hfs(V)$, then also $\CX=(X_1,\ldots,X_t)\in\hfs(V)$ where $t\in\NN\cup\{0\}$.
\end{itemize}
A \emph{(combinatorial) object} is a pair $(V,\CX)$ consisting of a ground set $V$ and a
hereditarily finite set $\CX\in\hfs(V)$.
The ground set $V$ is usually apparent
from context and
the combinatorial object $(V,\CX)$ is identified with
the hereditarily finite set $\CX$.
The set $\obj(V)$ denotes the set of all (combinatorial)
objects over $V$. 
The \emph{transitive closure} of an object $\CX$, denoted by $\tranCl(\CX)$, is
defined as all objects that recursively occur in $\CX$.
All labeling cosets that occur in $\CX$ are succinctly represented via generating sets.
The encoding size of an object $\CX$ can be chosen polynomial in $|\tranCl(\CX)|+|V|+t_{\max}$ where $t_{\max}$ is the maximal length of
a tuple in $\tranCl(\CX)$.

\paragraph{Ordered Objects}
An object is called \emph{ordered} if the ground set $V$ is linearly
ordered.
The linearly ordered ground sets that we consider are always
subsets of natural numbers with their standard ordering ``$<$''.
An object is \emph{unordered}
if $V$ is a usual set (without a given order).
Partially ordered objects in which some, but not all, atoms are
comparable are not considered.
\begin{lem}[\cite{DBLP:conf/stoc/SchweitzerW19}]\label{lem:prec}
There is an ordering ``$\prec$'' on pairs of ordered objects that can be computed in polynomial time.
\end{lem}

\paragraph{Applying Functions to Unordered Objects}
Let $V$ be an unordered ground set and let $V'$
be a ground set that is either ordered or unordered.
The image of an unordered object $\CX\in\obj(V)$
under a bijection $\mu:V\to V'$
is an object $\CX^\mu\in\obj(V')$ that is defined as follows.
\begin{itemize}
  \item $v^\mu:=\mu(v)$,
  \item $(\Delta\rho)^\mu:=\mu^{-1}\Delta\rho$,
  \item $\{X_1,\ldots,X_t\}^\mu:=\{X_1^\mu,\ldots,X_t^\mu\}$ and
  \item $(X_1,\ldots,X_t)^\mu:=(X_1^\mu,\ldots,X_t^\mu)$.
\end{itemize}

\paragraph{Isomorphisms and Automorphisms of Unordered Objects}
The set of all isomorphisms
from an object $\CX\in\obj(V)$ and to an object $\CX'\in\obj(V')$ is denoted by
$\iso(\CX;\CX'):=\{\phi:V\to V'\mid \CX^\phi=\CX'\}$.
The set of all automorphisms of an object $\CX$ is denoted by
$\aut(\CX):=\iso(\CX;\CX)$. Both isomorphisms and
automorphisms are defined for objects that are unordered only.

For two unordered sets $V$ and $V'$, the set $\iso(V;V')$ is also used to denote the set of
all bijections from $V$ to $V'$.
This notation indicates and stresses
that both $V$ and $V'$ have to be unordered.
Additionally, it is used in a context where
$\phi\in\iso(V;V')$ is seen as an isomorphism $\phi\in\iso(\CX;\CX^\phi)$.

\paragraph{Induced Groups and Labeling Cosets}
In the following, let $\CX\in\obj(V)$ be a set
and $\Delta\leq\aut(\CX)\leq\sym(V)$ be a group consisting of automorphisms of $\CX$.
For a permutation $\delta\in\Delta$, we define the permutation \emph{induced} on $\CX$, denoted by $\delta[\CX]$,
as the permutation that maps $X\in\CX$ to $\delta[\CX](X):=X^\delta\in\CX$.
We define
the group $\Delta$ \emph{induced} on $\CX$, denoted by $\Delta[\CX]\leq\sym(\CX)$,
as the group consisting of the elements $\delta[\CX]\in\sym(\CX)$ for $\delta\in\Delta$.
Similarly, for a labeling $\rho$ of $V$, we define the labeling $\rho$ \emph{induced} on $\CX$, denoted by $\rho[\CX]:\CX\to\{1,\ldots,|\CX|\}$,
as the labeling that
orders the elements in $\CX$ according to the ordering ``$\prec$'' from \cref{lem:prec}, i.e.,
$\rho(X_i)<\rho(X_j)$ if and only if $X_i^\rho\prec X_j^\rho$.
Furthermore, for a given labeling cosets $\Delta\rho\leq\lab(V)$, we define the \emph{induced labeling coset} on $\CX$, denoted by
$(\Delta\rho)[\CX]\leq\lab(\CX)$,
as $\Delta[\CX]\rho[\CX]$.

\paragraph{Generating Sets and Polynomial-Time Library}
For the basic theory of handling
permutation groups given by generating sets, we refer to \cite{seress}.
Indeed, most algorithms are based on strong generating sets.
However, given an arbitrary generating set, the Schreier-Sims algorithm is used to compute a
strong generating set (of size quadratic in the degree) in polynomial time.
In particular, we will use that the following tasks can be performed efficiently
when a group is given by a generating set.

\begin{enumerate}
  \item Given a vertex $v\in V$ and a
  group $G\leq\sym(V)$, the Schreier-Sims algorithm can be used to compute
  the pointwise stabilizer $G_{(v)}$ in polynomial time.
  \item Given a group $G\leq\sym(V)$, a subgroup that has a polynomial time membership problem
  can be computed in time polynomial in the index and the degree of the subgroup.
  \item Let $\CS=\Delta_1\rho_1,\ldots,\Delta_t\rho_t\leq\lab(V)$ be
  a sequence of labeling cosets of $V$. We write
  $\langle \CS\rangle$ for the smallest labeling coset $\Lambda$
  such that $\Delta_i\rho_i\subseteq \Lambda$ for all $i\in[t]$.
  Given a representation for
  $\CS$, the coset $\langle \CS\rangle$ can
  be computed in polynomial time.
  Furthermore, the computation of $\langle \CS\rangle$ is
  isomorphism invariant w.r.t.~$\CS$, i.e.,
  $\phi^{-1}\langle \CS\rangle=
\langle \phi^{-1}\CS \rangle$ for
all bijections $\phi:V\to V'$.
\end{enumerate}

\begin{defn}[\cite{DBLP:conf/stoc/SchweitzerW19}]
Let $\CC$ be an isomorphisms-closed class of unordered objects, i.e.,
for all $\CX\in\CC$ over a ground set $V$ and all bijections $\phi:V\to V'$
it holds that $\CX^\phi\in\CC$.
A \emph{canonical labeling function} $\can$ is
a function that assigns each unordered object
$\CX\in\CC$ a labeling coset $\can(\CX)=\Lambda\leq\lab(V)$
such that:

\begin{enumerate}[(\textnormal{CL}1)]
  \item\label{ax:cl1} $\can(\CX)=\phi\can(\CX^\phi)$
  for all $\phi\in\iso(V;V')$, and
  \item\label{ax:cl2} $\can(\CX)=\aut(\CX)\pi$
  for some (and thus for all) $\pi\in\can(\CX)$.
\end{enumerate}
In this case, the labeling coset $\Lambda$ is also called a \emph{canonical
labeling} for $\CX$.
\end{defn}

\begin{lem}[\cite{DBLP:conf/stoc/SchweitzerW19}, Object
Replacement Lemma]\label{lem:rep} Let $\CX=\{X_1,\ldots,X_t\}$ be an object and
let $\can$ and $\canSet$ be canonical labeling functions.
Define $\CX^\set:=\{\Delta_1\rho_1,\ldots,\Delta_t\rho_t\}$
where $\Delta_i\rho_i:=\can(X_i)$ is a canonical labeling for $X_i\in\CX$.
Assume that
$X_i^{\rho_i}=X_j^{\rho_j}$ for all $i,j\in [t]$.
Then, $\canObj(\CX):=\canSet(\CX^\set)$ defines a canonical labeling
for $\CX$.
\end{lem}

\section{Handling Small Objects via String Canonization}\label{sec:smallObj}

We consider the canonical labeling problem for a pair $(E,\Delta\rho)$ consisting of an edge relation $E\subseteq V^2$ and a labeling coset $\Delta\rho\leq\lab(V)$.

\problem{\canGraph\label{prob:CL:Graph}}
{$(E,\Delta\rho)\in\obj(V)$ where $E\subseteq V^2$,
$\Delta\rho\leq\lab(V)$ and $V$ is an unordered set}
{(E,\Delta\rho)}
{(E^\phi,\phi^{-1}\Delta\rho)}
{\aut((E,\Delta\rho))}

The automorphism group of $(E,\Delta\rho)$ is precisely $\aut((E,\Delta\rho))=\{\delta\in\Delta\mid (v,w)\in E\iff
(\delta(v),\delta(w))\in E\}$.
For $\Delta\rho=\lab(V)$, this is exactly the canonical labeling problem for directed graphs.
However, for labeling cosets $\Delta\rho\leq\lab(V)$ in general, the
problem is equivalent to the string canonization problem (this can be shown by defining a string $\xx:V^2\to\{0,1\}$
with positions $V^2$ such that $\xx((v,w))=1$ if and only if $(v,w)\in E$).

\begin{theo}[\cite{DBLP:conf/stoc/Babai19}]\label{theo:canGraph}
A function $\canGraph$
for \cref{prob:CL:Graph} can be computed in time
$2^{\polylog|V|}$.
\end{theo}

The next problem can be seen as a canonical intersection-problem for labeling cosets.

\problem{\canInt\label{prob:CL:Int}}
{$(\Theta\tau,\Delta\rho)\in\obj(V)$ where $\Theta\tau,\Delta\rho\leq\lab(V)$
and $V$ is an unordered set}
{(\Theta\tau,\Delta\rho)}
{(\phi^{-1}\Theta\tau,\phi^{-1}\Delta\rho)}
{(\Theta\cap\Delta)}

\begin{lem}\label{lem:canInt}
A function $\canInt$
solving \cref{prob:CL:Int} can be computed in time
$2^{\polylog|V|}$.
\end{lem}

\begin{proof}
It is know that this problem reduces to graph canonization in polynomial time \cite{DBLP:conf/stoc/SchweitzerW19}.
\end{proof}

Next, we define the central problem of this paper which is
introduced in \cite{DBLP:conf/icalp/GroheNSW18},\cite{DBLP:conf/stoc/SchweitzerW19}.
This problem is a canonical version of the multiple-coset isomorphism problem.

\problem{\canSet\label{prob:CL:Set0}}
{$J\in\obj(V)$ where $J=\{\Delta_1\rho_1,\ldots,\Delta_t\rho_t\}$,
$\Delta_i\rho_i\leq\lab(V)$ for all $i\in[t]$ and $V$ is an unordered set}
{(J)}
{(J^\phi)}
{\aut(J)}

The automorphism group of $J$ is precisely $\aut(J)=\{\sigma\in\sym(V)\mid\exists\psi\in\sym(t)
\forall
i\in[t]:\sigma^{-1}\Delta_i\rho_i
=\Delta_{\psi(i)}\rho_{\psi(i)}\}$.
We explain why this problem is the central problem when dealing with graph decompositions.

\paragraph{The Intuition Behind this Central Problem}
We want to keep this subsection as simple as possible and do not want to introduce tree decompositions yet.
For our purpose, we consider a simplified
formulation of a graph decomposition.
In this subsection, a graph decomposition of a graph $G=(V,E)$ is a family of subgraphs $\{H_i\}_{i\in[t]}$
that covers the edges of the entire graph, i.e., $E(G)=E(H_1)\cup\ldots\cup E(H_t)$.
We say that a graph decomposition is defined in an isomorphism-invariant way
if for two isomorphic graphs $G,G'$ the decompositions $\{H_i\}_{i\in[t]},\{H_i'\}_{i\in[t]}$ are defined in such a way
that
each isomorphism $\phi\in\iso(G;G')$
also maps each subgraph $H_i$ of the decomposition of $G$ to a subgraph $H_j'$ of the decomposition of $G'$.
In particular, such a decomposition has to be invariant under automorphisms of the graph.

Assume we have given a graph $G$ for which we can construct a graph decomposition $\{H_i\}_{i\in[t]}$ in an isomorphism-invariant way
and our task is the computation of a canonical labeling for $G$.
A priori, it is unclear how to exploit our graph decomposition.
In a first step, we could compute canonical labelings $\Delta_i\rho_i:=\can(H_i)$ for each subgraph $H_i$ recursively.
The central question is how to merge these labeling cosets $\Delta_i\rho_i$ for $H_i$ in order to obtain a canonical labeling $\Delta\rho$
for the entire graph $G$.

The easy case occurs when all subgraphs $H_i,H_j$ are pairwise non-isomorphic.
In this case, the subgraphs cannot be mapped to each other and indeed $\aut(G)=\aut(H_1)\cap\ldots\cap\aut(H_t)$.
Therefore, the computation of $\Delta\rho$ reduces to a canonical intersection-problem.
In fact, the algorithm from \cref{lem:canInt} can be used to compute canonical labelings $\Delta_{ij}\rho_{ij}:=\canInt(\Delta_i\rho_i,\Delta_j\rho_j)$
with $\Delta_{ij}=\Delta_i\cap\Delta_j=\aut(H_i)\cap\aut(H_j)$.
By an iterated use of that canonical intersection-algorithm, we can finally compute $\Delta\rho$ with
$\Delta=\Delta_1\cap\ldots\cap\Delta_t=\aut(H_1)\cap\ldots\cap\aut(H_t)=\aut(G)$.
Actually, the order in which we ``intersect'' the canonical labelings $\Delta_i\rho_i$ does matter
and we need to be careful in order to ensure isomorphism invariance (CL1).
(For example, there might be a canonical labeling function with
$\canInt((\lab(V),\Delta\rho))=\Delta\rho$ and $\canInt((\Delta\rho,\lab(V)))=\Delta\rho\pi$
for $\Delta\rho\neq\lab(V)$
where $\pi$ is a permutation of $\{1,\ldots,|V|\}$ that swaps 1 and 2 and fixes all other elements.
Clearly, $\canInt$ can be extended to a canonical labeling function satisfying (CL1) and (CL2).
However, $\canInt((\lab(V),\Delta\rho))\neq\canInt((\Delta\rho,\lab(V)))$).

Let us consider the second extreme case in which all subgraphs $H_i,H_j$ are pairwise isomorphic.
In such a case, we have that $\aut(G)=\{\sigma\in\sym(V)\mid\exists\psi(t)\forall i\in[t]:\sigma\in\iso(H_i;H_{\psi(i)})\}$.
Equivalently, we have that $\aut(G)=\aut(\{\Delta_1\rho_1,\ldots,\Delta_t\rho_t\})$.
Therefore,
by the definition of \cref{prob:CL:Set0}, the canonical labeling $\Delta\rho:=\canSet(\{\Delta_1\rho_1,\ldots,\Delta_t\rho_t\})$
defines a canonical labeling for the entire graph $G$.
Alternatively, one can use object replacement (\cref{lem:rep}) which intuitively
says that for the purpose of canonization the subgraphs $H_i$ can be replaced with
their labeling cosets $\Delta_i\rho_i$.
This also shows that $\Delta\rho:=\canSet(\{\Delta_1\rho_1,\ldots,\Delta_t\rho_t\})$
define a canonical labeling for the entire graph $G$.
Roughly speaking, \cref{prob:CL:Set0} can be seen as the task of merging
the given labeling cosets.

The mixed case in which some (but not all) subgraphs $H_i,H_j$ are isomorphic can be handled by a mixture of
the above cases.

The main algorithm (\cref{theo:canSet}) solves \cref{prob:CL:Set0} in a running time of
$(|V|+|J|)^{\polylog|V|}$.
In \cref{sec:treewidth}, we apply this problem
to graphs $G$ with $n$ vertices of treewidth $k$. In fact, we are able to bound $|V|\leq k$
and $|J|\leq n$ in this application which leads to the desired running time of $n^{\polylog(k)}$.

But first of all, we give a simple algorithm that has a weaker running time which is quasipolynomial in $|V|+|J|$.

\begin{lem}\label{lem:canSet0}
A function $\canSet$
solving \cref{prob:CL:Set0} can be computed in time
$(|V|+|J|)^{\polylog(|V|+|J|)}$.
\end{lem}

The proof is similar to the proof of Lemma 23 in the arXiv version of \cite{DBLP:conf/icalp/GroheNSW18}.
By increasing the permutation domain $V$ by a factor $|J|$, \cref{prob:CL:Set0} can actually be reduced to a graph canonization problem.
For the sake of completeness, we give the detailed proof in \cref{app}.

We consider the canonization problem for combinatorial objects.

\problem{\canObj\label{prob:CL:Obj}}
{$\CX\in\obj(V)$ where $V$ is an unordered set}
{(\CX)}
{(\CX^\phi)}
{\aut(\CX)}

For an object $\CX\in\obj(V)$, let $t_{\max}(\CX)$ be the size of the largest set involved in $\CX$, i.e.,
$t_{\max}(v)=0$ and $t_{\max}(\Delta\rho)=0$ for vertices $v\in V$ and labeling cosets $\Delta\rho\leq\lab(V)$
and inductively $t_{\max}((X_1,\ldots,X_s))=\max_{i\in[s]}t_{\max}(\CX_i)$
and $t_{\max}(\{X_1,\ldots,X_s\})=\max\{\max_{i\in[s]}t_{\max}(\CX_i),s\}$.
It is known that canonical labeling for combinatorial objects (on a ground set $V$) reduces to canonical labeling for
instances of \cref{prob:CL:Int} (on the same ground set $V$) and instances of
\cref{prob:CL:Set0} (on the same ground set $V$ and of size $t_{\max}$) in polynomial time \cite{DBLP:conf/stoc/SchweitzerW19}.
Therefore, \cref{prob:CL:Set0} is a central problem when canonizing
combinatorial objects in general.

\begin{cor}\label{cor:canObj0}
A function $\canObj$
solving \cref{prob:CL:Obj} can be computed in time
$2^{\polylog(|V|+t_{\max})}n^{\CO(1)}$
where $n$ is the input size (as defined in the preliminaries) and $t_{\max}\leq n$ is the size of the largest set involved in $\CX$.
\end{cor}

A later algorithm (\cref{cor:canObj}) shows that canonical labelings for combinatorial objects can actually
be computed in
time $n^{\polylog|V|}$ (or more precise $(|V|+t_{\max})^{\polylog|V|}n^{\CO(1)}$).

\section{Canonization of $k$-ary Relations}\label{sec:rel}

In this section, we consider the canonization problem for $k$-ary relations.
As graphs can be seen as binary relations, this problem clearly generalizes the graph canonization problem.

\problem{\canRel\label{prob:CL:Rel}}
{$R\in\obj(V)$ where $R\subseteq V^k$ for some $k\in\NN$ and $V$ is an unordered set}
{(R)}
{(R^\phi)}
{\{\sigma\in\sym(V)\mid(x_1,\ldots,x_k)\in R\iff (\sigma(x_1),\ldots,\sigma(x_k))\in R\}}

One way to canonize $k$-ary relations is by using a well-known reduction to the graph canonization problem \cite{DBLP:journals/jcss/Miller79}.
Alternatively, the algorithm from \cref{cor:canObj0} for combinatorial objects in general could also be applied to
$k$-ary relations.
However, both approaches lead to a running time that is quasipolynomial in $|V|+|R|$, i.e., $2^{\polylog(|V|+|R|)}$.
In this section, we will give a polynomial-time reduction to the canonization problem
for combinatorial objects which are of input size polynomial in $|V|$ (which does not depend on $|R|$).
With this reduction, we obtain an improved algorithm that runs in time $2^{\polylog|V|}|R|^{\CO(1)}$.
Our bound improves the currently best algorithm from \cite{DBLP:conf/focs/GroheNS18}.
Moreover, our time bound is also optimal (when measured in $|V|$ and $|R|$) as long as the graph isomorphism problem can not be solved faster
than quasipolynomial time.

\paragraph{Partitions}
An (unordered) \emph{partition} of a set $\CX\in\obj(V)$ is a set $\CP=\{P_1,\ldots,P_p\}$ such that $\CX=P_1\cupdot\ldots\cupdot P_p$
where $\emptyset\neq P_i\subseteq\CX$ for all $P_i\in\CP$.
In the algorithms that follow, our constructions can lead to ``partitions'' with a non-empty part.
In such a case, we implicitly forget about the empty set in the partition.
We say that $\CP$ is the \emph{singleton partition} if $|\CP|=1$
and we say that $\CP$ is the \emph{partition into singletons} if $|P_i|=1$ for all $P_i\in\CP$.
A partition $\CP$ is called \emph{trivial} if $\CP$ is the singleton partition or the partition into singletons.

\paragraph{The Partitioning Technique}
We suggest a general technique for exploiting partitions.
In this setting, we assume that we are given some object $\CX\in\obj(V)$ for which we can construct a partition $\CP=\{P_1,\ldots,P_p\}$
in an isomorphism-invariant way such that
$2\leq|\CP|\leq 2^{\polylog|V|}$.
The goal is the computation of a canonical labeling for $\CX$ by using an efficient recursion.

Using recursion, we compute a canonical labeling $\Delta_i\rho_i$ for each part $P_i\subseteq\CX$ recursively
(assumed that we can define a partition for each part again).
So far, we computed canonical labelings for each part $P_i\subseteq\CX$ independently.
The main idea is to use our central problem (\cref{prob:CL:Set0})
to merge all these labeling cosets.
Let us restrict our attention to the case in which the parts $P_i,P_j\in\CP$ are pairwise isomorphic.
In this case, we define the set $\CP^\set:=\{\Delta_i\rho_i\mid P_i\in\CP\}$ consisting of the canonical labelings $\Delta_i\rho_i$
for each part.
Moreover, by object replacement (\cref{lem:rep}), a canonical labeling for $\CP^\set$ defines a canonical labeling
for $\CP$ as well.
A canonical labeling for $\CP$ in turn defines a canonical labeling
for $\CX$ since we assume the partition to be defined in an isomorphism-invariant way.
Therefore, it is indeed true that a canonical labeling for $\CP^\set$ would define a
canonical labeling for $\CX$.
For this reason, we can use the algorithm from \cref{lem:canSet0}
to compute a canonical labeling for $\CP^\set$.
The algorithm runs in the desired time bound since
$|\CP^\set|=|\CP|\leq 2^{\polylog|V|}$
is bounded by some quasipolynomial.

Let us consider the number of recursive calls $R(\CX)$ of this approach for a given object $\CX$.
Since we recurse on each part $P_i\in\CP$,
we have a recurrence of $R(\CX)=1+\sum_{P_i\in\CP} R(P_i)$ leading to at most $|\CX|^{\CO(1)}$ recursive calls.
The running time for one single recursive call is bounded by $2^{\polylog|V|}$.
For this reason, the total running time is bounded by $2^{\polylog|V|}|\CX|^{\CO(1)}$.

\begin{theo}\label{theo:canRel}
A function $\canRel$
solving \cref{prob:CL:Rel} can be computed in time
$2^{\polylog|V|}|R|^{\CO(1)}$.
\end{theo}

\begin{proof}
\algorithmDAN{\canRel(R)}
\begin{cs}

\case{$|R|\leq 1$}~\\
Compute and return $\Lambda:=\canObj(R)$ using \cref{cor:canObj0}.\\
\comment{Since the size of the largest set involved in $R$ is the set $R$ itself,
the algorithm from \cref{cor:canObj0} runs in time $2^{\polylog|V|}$}

\case{$|R|\geq 2$}~\\
\comment{In this case, it is possible to
define a partition $\CP$ of $R$ in an isomorphism-invariant way.
We will use this partition for a recursion as described in the partitioning technique}\\
Let $r$ be the first position in which
$R$ differs, i.e., the smallest $r\in[k]$
such that there are $(x_1,\ldots,x_k),(y_1,\ldots,y_k)\in R$
with $x_r\neq y_r$.\\
Define an (unordered) partition $\CP:=\{P_v\mid v\in V\}$ of $R=\bigcupdot_{v\in V}P_v$
where
$P_v:=
\{(x_1,\ldots,x_k)\in R\mid x_r=v\}$.\\
\comment{By the choice of $r\in[k]$, this is not the singleton partition. On the other side,
the size $|\CP|\leq|V|$ is obviously bounded by a quasipolynomial in $|V|$}\\
Compute $\Delta_v\rho_v:=\canRel(P_v)$
for each subrelation $P_v\in\CP$ recursively.\\
Define $\CP^\set:=\{(\Delta_v\rho_v,v)\mid P_v\in\CP\}$.\\
\comment{We define an ordering according to the isomorphism type of the subrelations $P_v,v\in V$}\\
Define an \emph{ordered} partition $\PP:=(\CP_1,\ldots,\CP_p)$ of $\CP^\set=\CP_1\cupdot\ldots\cupdot\CP_p$ such that:\\
$P_v^{\rho_v}\prec P_w^{\rho_w}$, if and only if $(\Delta_v\rho_v,v)\in\CP_i$ and $(\Delta_w\rho_w,w)\in\CP_j$ for some $i,j\in[p]$ with $i<j$.\\
Compute $\Lambda_i:=\canSet(\CP_i)$ for each $\CP_i,i\in[p]$ using \cref{lem:canSet0}.\\
\comment{
Since $|\CP_i|\leq|\CP^\set|=|V|$,
the algorithm from \cref{lem:canSet0} runs in the desired time bound, i.e., $2^{\polylog|V|}$}\\
Compute and return $\Lambda
:=\canObj((\Lambda_1,\ldots,\Lambda_p))$ using \cref{cor:canObj0}.\\
\comment{
Since $(\Lambda_1,\ldots,\Lambda_p)$ is a tuple consisting of atoms, no set is involved in this object.
Therefore, the algorithm from \cref{cor:canObj0} also runs in the desired time bound, i.e., $2^{\polylog|V|}$}
\end{cs}

\begin{cl1}
We consider the Case $|R|\geq 2$.
Assume we have $R^\phi$ instead of
$R$ as an input.
We obtain the partition $\CP^\phi$
instead of $\CP$.
By induction, we compute $\phi^{-1}\Delta_v\rho_v$ instead
of $\Delta_v\rho_v$ and
obtain $(\CP^\set)^\phi$ instead of $\CP^\set$.
By (CL1) of $\canSet$, we obtain
$\phi^{-1}\Lambda_i$ instead of
$\Lambda_i$.
By (CL1) of $\canObj$, we obtain
$\phi^{-1}\Lambda$ instead of
$\Lambda$, which was what we wanted to show.
\end{cl1}

\begin{cl2}
We consider the Case $|R|\geq 2$.
We return $\Lambda=\canObj((\Lambda_1,\ldots,\Lambda_p))$.
By object replacement (\cref{lem:rep}), the labeling coset
$\Lambda$ defines a canonical labeling for $(\CP_1,\ldots,\CP_p)$ as well.
The \emph{ordered} partition $\PP=(\CP_1,\ldots,\CP_p)$ of $\CP^\set$ is defined in an isomorphism-invariant way
and therefore $\Lambda$ defines a canonical
labeling for $\CP^\set$.
Since $\PP$ orders $\CP^\set$ according to the isomorphism type of the subrelations
the object replacement lemma (\cref{lem:rep}) implies that
$\Lambda$ defines a canonical labeling for $\CP$ as well.
The (unordered) partition $\CP=\{P_v\mid v\in V\}$ of $R$ is defined in an isomorphism-invariant way
and therefore $\Lambda$ defines a canonical
labeling for $R$, which was what we wanted to show.
\end{cl2}

\begin{runtime}
We claim that the number of recursive calls $N(R)$ is at most $T:=|R|^2$. By induction, it can be seen that
\begin{equation*}
N(R)= 1+\sum_{v\in V} N(P_v)
\overset{\text{induction}}{\leq}
1 + \sum_{v\in V} |P_v|^2\leq T.
\end{equation*}
We consider the running time of one single recursive call.
The algorithm from \cref{cor:canObj0} runs in time $2^{\polylog|V|}$.
Therefore, the total running time is bounded by $2^{\polylog|V|}|R|^{\CO(1)}$.
\end{runtime}
\end{proof}

\section{Canonization of Hypergraphs}\label{sec:hyper}

In this section, we consider hypergraphs and later so-called coset-labeled hypergraphs.

\problem{\canHyper\label{prob:CL:Hyper}}
{$H\in\obj(V)$ where
$H=\{S_1,\ldots,S_t\}$, $S_i\subseteq V$ for all $i\in[t]$ and $V$ is an unordered set}
{(H)}
{(H^\phi)}
{\{\sigma\in\sym(V)\mid S\in H\iff S^\sigma\in H\}}

We want to extend the previous partitioning technique to hypergraphs.
However, for hypergraphs a non-trivial isomorphism-invariant partition $H=H_1\cupdot\ldots\cupdot H_s$
of the edge set does not always exist, e.g.,
the hypergraph $(V,\{S\subseteq V\mid |S|=2\})$ does not have a non-trivial partition of the edge set that is preserved
under automorphisms.
Therefore, we can not apply the partitioning technique to this setting.
For this reason, we introduce a generalized technique in order to solve this problem.
This generalized technique results in a slightly weaker time bound of $(|V|+|H|)^{\polylog|V|}$ (where the dependency on $|H|$ is not polynomial).
Indeed, it is an open problem whether the running time for the hypergraph isomorphism problem can be improved to $2^{\polylog|V|}\cdot |H|^{\CO(1)}$ \cite{Babai2018GroupsG}.

\paragraph{Covers}
A \emph{cover} of a set $\CX\in\obj(V)$ is a set $\CC=\{C_1,\ldots,C_c\}$ such that $\CX=C_1\cup\ldots\cup C_c$
where $\emptyset\neq C_i\subseteq\CX$ for all $C_i\in\CC$.
In contrast to a partition, the sets $C_i,C_j$ are not necessarily disjoint for $i\neq j$.
We say that $\CC$ is the \emph{singleton cover} if $|\CC|=1$
and we say that $\CC$ is the \emph{cover into singletons} if $|C_i|=1$ for all $C_i\in\CC$.
A cover $\CC$ of $\CX$ is called \emph{sparse} if $|C_i|\leq\frac{1}{2}|\CX|$ for all $C_i\in\CC$.

\paragraph{The Covering Technique}
Extending the partitioning technique, we suggest a technique to handle covers.
In this setting, we assume that we have given some object $\CX\in\obj(V)$ for which we can define a cover $\CC=\{C_1,\ldots,C_c\}$
in an isomorphism-invariant way.
Also here, we assume that $2\leq|\CC|\leq 2^{\polylog|V|}$.
The goal is the computation of a canonical labeling of $\CX$ using an efficient recursion.

First, we reduce to the setting in which $\CC$ is a sparse cover of $\CX$. This can be done as follows.
We define $C_i^*:=C_i$ if $|C_i|\leq\frac{1}{2}|\CX|$ and we define $C_i^*:=\CX\setminus C_i$ if
$|C_i|>\frac{1}{2}|\CX|$.
By definition, we ensured that $|C_i^*|\leq\frac{1}{2}|\CX|$ for all $i\in[c]$.
Let $\CX^*:=\bigcup_{i\in[c]}C_i^*$.
Next, we consider two cases.

If $\CX^*\subsetneq \CX$, then we have found a non-trivial partition $\CX=\CX^*\cupdot \CX^\circ$ where $\CX^\circ:=\CX\setminus\CX^*$.
We proceed analogously as in the partitioning technique explained in \cref{sec:rel}.

Otherwise, if $\CX^*=\CX$, then $\CC^*:=\{C_1^*,\ldots,C_c^*\}$ is also a cover of $\CX$.
But more importantly, the cover $\CC^*$ is indeed sparse.
In the case of a sparse cover, we also proceed analogously as in the partition technique explained in \cref{sec:rel}.
However, the key difference of the covering technique
compared to the partitioning technique
lies in the recurrence for the number of recursive calls since the sets $C_i^*,C_j^*\in\CC^*$ are not necessarily pairwise disjoint.
The recurrence we have is $R(\CX)= 1+ \sum_{C_i^*\in\CC^*} R(C_i^*)$.
By using that $|\CC^*|=|\CC|\leq 2^{\polylog|V|}$
and that $|C_i^*|\leq\frac{1}{2}|\CX|$, we obtain at most $|\CX|^{\polylog|V|}$ recursive calls.
This is exactly the reason why the algorithm for relations is faster than the algorithm
for hypergraphs.

\begin{theo}\label{theo:canHyper}
A function $\canHyper$
for \cref{prob:CL:Hyper} can be computed in time
$(|V|+|H|)^{\polylog|V|}$.
\end{theo}

\begin{proof}
\algorithmDAN{\canHyper(H)}
\begin{cs}
\case{$|H|\leq 1$}~\\
Compute and return $\canObj(H)$ using \cref{cor:canObj0}.\\
\comment{Since $H$ consists of at most one hyperedge, the largest set involved in $H$ is bounded by $|V|$.
Therefore, the algorithm from \cref{cor:canObj0} runs in time $2^{\polylog|V|}$}

\case{$|H|\geq 2$}~\\
\comment{In this case, it is possible to define a cover of the hypergraph $H$
in an isomorphism-invariant way}\\
Define a cover $\CC:=\{C_v\mid v\in V\}$ of $H=\bigcup_{v\in V}C_v$ where
$C_v:=\{S_i\in H\mid v\in S_i\}$.\\
\comment{Since $|H|\geq 2$, this is not the singleton cover. On the other side,
the size $|\CC|\leq|V|$ is obviously bounded by a quasipolynomial in $|V|$.
However, the cover might not be sparse.
Next, we want to find a sparse cover}\\
Define $C_v^*:=\begin{cases}
C_v,&\text{if }|C_v|\leq\frac{1}{2}|H|\\
H\setminus C_v,&\text{otherwise if }|C_v|>\frac{1}{2}|H|.\\
\end{cases}$.\\
Define $H^*=\bigcup_{v\in V} C_v^*$.
\begin{cs}
\case{$H^*\subsetneq H$}~\\
\comment{In this case, we found an ordered partition of $H$ and proceed with the partitioning technique}\\
Define an \emph{ordered} partition $\CH=(H^*,H^\circ)$ of $H$ where $H^\circ:=H\setminus H^*$.\\
Compute $\Lambda_1:=\canHyper(H^*)$ recursively.\\
Compute $\Lambda_2:=\canHyper(H^\circ)$ recursively.\\
\comment{Next, we combine the two labeling cosets by using a canonical intersection-problem}\\
Compute and return $\Lambda:=\canObj((\Lambda_1,\Lambda_2))$ using \cref{lem:canInt} or \cref{cor:canObj0}.

\case{$H^*= H$}~\\
\comment{In this case, we found a sparse cover of $H$ and proceed with the covering technique}\\
Define a sparse cover $\CC^*:=\{C_v^*\mid v\in V\}$ of $H=\bigcup_{v\in V}C_v^*$.\\
Compute $\Delta_v\rho_v:=\canHyper(C_v^*)$ for each subhypergraph $C_v^*\in\CC^*$ recursively.\\
Define ${\CC^*}^\set:=\{(\Delta_v\rho_v,v)\mid C_v^*\in\CC^*\}$.\\
\comment{We define an ordering according to the isomorphism type of the subhypergraphs $C_v^*\in\CC^*$}\\
Define an \emph{ordered} partition $\PP:=(\CP_1,\ldots,\CP_p)$ of ${\CC^*}^\set=\CP_1\cupdot\ldots\cupdot\CP_p$ such that:\\
$(C_v^*)^{\rho_v}\prec (C_w^*)^{\rho_w}$, if and only if $(\Delta_v\rho_v,v)\in\CP_i$ and $(\Delta_w\rho_w,w)\in\CP_j$ for some $i,j\in[p]$ with $i<j$.\\
Compute $\Lambda_i:=\canSet(\CP_i)$ for each $\CP_i,i\in[p]$ using \cref{lem:canSet0}.\\
\comment{
Since $|\CP_i|\leq|\CP^\set|=|V|$,
the algorithm from \cref{lem:canSet0} runs in the desired time bound, i.e., $2^{\polylog|V|}$}\\
Compute and return $\Lambda
:=\canObj((\Lambda_1,\ldots,\Lambda_p))$ using \cref{cor:canObj0}.\\
\comment{
Since $(\Lambda_1,\ldots,\Lambda_p)$ is a tuple consisting of atoms, no set is involved in this object.
Therefore, the algorithm from \cref{cor:canObj0} also runs in the desired time bound, i.e., $2^{\polylog|V|}$}
\end{cs}
\end{cs}

The proof for conditions (CL1) and (CL2) is similar to the proof of \cref{theo:canRel}.

\begin{runtime}
We claim that the number of recursive calls $R(H)$ is at most $T:=|H|^{2\log_2|V|}$.
For the case $H^*\subsetneq H$, we have that
\begin{align*}
R(H)&=1+R(H^*)+R(H^\circ)\\
&\overset{\mathclap{\text{induction}}}{\leq}
1+|H^*|^{2\log_2|V|}+|H^\circ|^{2\log_2|V|}\leq T
\end{align*}
For the case $H^*= H$, we have that
\begin{align*}
R(H)&= 1+\sum_{v\in V} R(C_v^*)\\
&\overset{\mathclap{\text{induction}}}{\leq}
1 + \sum_{v\in V} |C_v^*|^{2\log_2|V|}\\
&\leq 1+ |V| \frac{|H|^{2\log_2|V|}}{|V|^2}\leq T&&\text{(using $|C_v^*|\leq\frac{1}{2}|H|$).}
\end{align*}
We consider the running time of one single recursive call.
The algorithm from \cref{cor:canObj0} runs in time $2^{\polylog|V|}$.
Therefore, the total running time is bounded by $(|V|+|H|)^{\polylog|V|}$.
\end{runtime}
\end{proof}

\paragraph{Giants, Johnsons and Cameron Groups}
The groups $\alt(V)$ and $\sym(V)$ are called \emph{giants}.
The groups $\alt(V)[\binom{V}{s}]$ and $\sym(V)[\binom{V}{s}]$ (the alternating group and the symmetric group on $V$ acting on the $s$-element subsets
of $V$) are called \emph{Johnson groups} where $s\leq\frac{1}{2}|V|$. The size $|V|$ is called the \emph{Johnson parameter}.
A group $\Delta\leq\sym(V)$ is called a Cameron group,
if $V=\binom{W}{s}^k$ for some set $W$ and some integers $k\geq 1\leq s\leq\frac{1}{2}|W|> 2$
and $(\alt(W)[\binom{W}{s}])^k\leq\Delta\leq\sym(W)[\binom{W}{s}]\wr\sym(k)$
(primitive wreath product action).
Additionally, we require that the induced homomorphism $h:\Delta\to\sym(k)$ is transitive
and that $s\neq\frac{1}{2}|W|$. These additional requirements ensure that Cameron groups are primitive.

\paragraph{Composition-Width}
For a group $\Delta\leq\sym(V)$, the composition-width of $\Delta$, denoted as $\cw\Delta$,
is the smallest integer $k$ such that all composition factors of $\Delta$ are isomorphic to a subgroup of $\sym(k)$.

\begin{prop}\label{lem:gammaD}
Let $\Delta\leq\sym(\CX)$ be a primitive group on a set $\CX$ with $\cw\Delta\leq d$.
Then at least one of the following is true.
\begin{enumerate}
  \item\label{lem:gammaD1} $|\Delta|\in |\CX|^{\CO(\log d)}$, or
  \item\label{lem:gammaD2} $d!< |\CX|$, or
  \item\label{lem:gammaD3} there is a sparse cover $\CC=\{C_1,\ldots,C_c\}$ of $\CX=C_1\cup\ldots\cup C_{c}$ with $2\leq|\CC|\leq d^3$ which is $\Delta$-invariant.
\end{enumerate}
Moreover, there is a polynomial-time algorithm that determines one of the options that is satisfied
and in case of the third option computes the corresponding cover $\CC$ of $\CX$.
\end{prop}

\begin{proof}
The well known O’Nan-Scott Theorem classifies primitive groups into the following types:
\RN{1}. Affine Groups, \RN{2}. Almost Simple Groups, \RN{3}. Simple Diagonal Action,
\RN{4}. Product Action, \RN{5}. Twisted Wreath Product Action.
For groups $\Delta\leq\sym(\CX)$ of Type \RN{1},\RN{3} or \RN{5} it holds that
$|\Delta|\in|\CX|^{\CO(\log d)}$ \cite{DBLP:conf/focs/GroheNS18}.

Assume that $\Delta$ is of Type \RN{2}.
Then, $|\Delta|\in|\CX|^{\CO(\log d)}$ or $\Delta$
is permutationally isomorphic to a Johnson group with parameter $|V|\leq d$ \cite{DBLP:conf/focs/GroheNS18}.
We identify $\CX=\binom{V}{s}$.
We define a cover $\CC:=\{C_v\mid v\in V\}$ of $\CX=\binom{V}{s}$
where $C_v:=\{X\in\binom{V}{s}\mid v\in X\subseteq V\}$.
Observe that $|\CC|=|V|\leq d$. 
Moreover, this cover is sparse and $\Delta$-invariant.

Assume that $\Delta\leq\sym(\CX)$ is of Type \RN{4}.
Then, $|\Delta|\in|\CX|^{\CO(\log d)}$ or $\Delta$ is a subgroup of a Cameron group $P \wr \Psi$
where $P$ is a Johnson group with parameter $|V|\leq d$
and $\Psi\leq\sym(k)$ is transitive with $\cw\Psi\leq d$.
We identify $\CX=\binom{V}{s}^k$.
We define a cover $\CC:=\{C_{v,i}\mid v\in V,i\in[k]\}$ of $\CX=\binom{V}{s}^k$
where $C_{v,i}:=\{(X_1,\ldots,X_k)\in\binom{V}{s}^k\mid v\in X_i\subseteq V\}$.
Again, $\CC$ is sparse and $\Delta$-invariant.
Observe that $|\CC|\leq|V|\cdot k$.
Since $|\CX|=\binom{|V|}{s}^k$, it follows that $k\leq\log_2|\CX|$.
Furthermore,
we can assume that $d!\geq |\CX|$ because otherwise Option \ref{lem:gammaD2} of the Lemma holds.
Therefore, $\log_2|\CX|\leq d\log_2(d)$.
This leads to $|\CC|\leq d\cdot k\leq d^3$.
\end{proof}

\paragraph{Canonical Generating Sets}
A canonical generating set can be seen as a unique encoding of a group $\Delta^\ord\leq\sym(V^\ord)$
over a linearly ordered set $V^\ord=\{1,\ldots,|V|\}$.

\begin{lem}[\cite{DBLP:journals/siamcomp/AllenderGMMM18}, Lemma 6.2, \cite{DBLP:conf/icalp/GroheNSW18}, Lemma
21 arXiv version]\label{lem:canGen}
There is a polynomial-time algorithm that, given a group $\Delta^\ord\leq\sym(\{1,\ldots,|V|\})$
via a generating set,
computes a generating set for $\Delta^\ord$.
The output only depends
on $\Delta^\ord$ (and not on the given generating set).
\end{lem}
The applications of canonical generating sets to our framework are discussed in \cite{DBLP:conf/stoc/SchweitzerW19}.
Assume that we want to use an algorithm $\LA$ as a black box in our framework which gets as input
an encoding of a permutation group $\Delta\leq\sym(V)$
and produces some output $\LA(\Delta)\in\obj(V)$.
For example, the algorithm from \cref{lem:gammaD} gets as input a group $\Delta\leq\sym(V)$
and might produce a cover $\CC$ of $V$.
Another example could be an algorithm that gets as input a group $\Delta\leq\sym(V)$
and produces a minimal block system $\CB$ for $\Delta$.
When designing a canonization algorithm, it is important that
the subroutines that are used behave in an isomorphism-invariant way.
That means that for all bijections $\phi:V\to V'$ the algorithm satisfies $\LA(\Delta^\phi)=\LA(\Delta)^\phi$.
We can achieve this as follows.
We ensure that black box algorithms are applied to groups $\Delta^\ord\leq\sym(V^\ord)$
over the linearly ordered set $V^\ord=\{1,\ldots,|V|\}$ only. 
The benefit is that isomorphisms $\phi:V\to V'$ act trivially on ordered groups, i.e., $(\Delta^\ord)^\phi=\Delta^\ord$.
For this reason, it remains to ensure that $\LA(\Delta^\ord)$
only depends on $\Delta^\ord$ (and not on the representation of $\Delta^\ord$).
Here, we use canonical generating sets to represent a group uniquely.
We will use this trick in the proof of \cref{lem:canSetSet}.

After we consider canonization problem for explicitly given structures such as $k$-ary relations and hypergraphs,
we will now turn back to sets $J$ consisting of implicitly given labeling cosets.

\problem{\canSetSet\label{prob:CL:SetSet}}
{$(J,L,\alpha,\Delta\rho)\in\obj(V)$ where
$J=\{\Delta_1\rho_1,\ldots,\Delta_t\rho_t\}$,
$L=\{\Lambda_1,\ldots,\Lambda_t\}$,
$\Delta_i\rho_i,\Lambda_i\leq\lab(V)$,
$\alpha:J\to L$ is a function
with $\alpha(\Delta_i\rho_i)=\Lambda_i$,
$\Delta\rho\leq\lab(V)$
and $V$ is an unordered set.
We require that $\Delta\leq\aut(J)$}
{(J,L,\alpha,\Delta\rho)}
{(J^\phi,L^\phi,\alpha^\phi,\phi^{-1}\Delta\rho)}
{\aut((J,L,\alpha,\Delta\rho))}

In fact, a function $\alpha:J\to L$ can be seen as a set consisting of pairs $(\Delta_i\rho_i,\alpha(\Delta_i\rho_i))$ 
is an object in our framework.
In this problem, we assume that for each $\Delta_i\rho_i\in J$ a second labeling coset $\Lambda_i\in L$
is given. Moreover, we assume that the group $\Delta\leq\sym(V)$ already permutes
the set of labeling cosets $J=\{\Delta_1\rho_1,\ldots,\Delta_t\rho_t\}$, i.e., $\Delta\leq\aut(J)$.
The automorphisms of the instance $(J,L,\alpha,\Delta\rho)$ are all permutations $\delta\in\Delta$
such that if $(\Delta_i\rho_i)^\delta=\Delta_j\rho_j$, then $\delta$ also maps the corresponding labeling coset $\Lambda_i$
to the corresponding labeling coset $\Lambda_j$.
Formally, this means $\aut((J,L,\alpha,\Delta\rho))=\{\delta\in\Delta\mid
\forall i,j\in[t]:(\Delta_i\rho_i)^\delta
=\Delta_j\rho_j\implies\Lambda_i^\delta=\Lambda_j\}$

\begin{lem}\label{lem:canSetSet}
A function $\canSetSet$
solving \cref{prob:CL:SetSet} can be computed in time
$(|V|+|J|)^{\polylog|V|}$.
\end{lem}

\begin{proof}
\algorithmDAN{\canSetSet(J,L,\alpha,\Delta\rho)}
\begin{cs}
\case{$|J|\leq 1$}~\\
Compute and return $\canObj((J,L,\alpha,\Delta\rho))$ using \cref{cor:canObj0}.
\end{cs}
\comment{Since $\Delta\leq\aut(J)$, the group $\Delta$ induces a permutation group $\Delta[J]\leq\sym(J)$}
\begin{cs}
\case{$\Delta[J]$ is intransitive}~\\
\comment{We proceed with the partitioning technique}\\
Define an \emph{ordered} partition $\CJ:=(J_1,J_2)$ of $J=J_1\cupdot J_2$
where $J_1$ is the $\Delta[J]$-orbit such that $J_1^\rho$ is minimal w.r.t.~to the ordering ``$\prec$" from \cref{lem:prec}.\\
Define an \emph{ordered} partition $\CL:=(L_1,L_2)$ of $L=L_1\cupdot L_2$ where $L_i:=J_i^{\alpha}$ for both $i=1,2$.\\
Compute $\Lambda_1:=\canSetSet(J_1,L_1,\alpha|_{J_1},\Delta\rho)$ recursively.\\
Compute $\Lambda_2:=\canSetSet(J_2,L_2,\alpha|_{J_2},\Delta\rho)$ recursively.\\
Compute and return $\Lambda:=\canObj((\Lambda_1,\Lambda_2))$ using \cref{lem:canInt} or \cref{cor:canObj0}.

\case{$\Delta[J]$ is transitive}~\\
\comment{We want to find a cover by using \cref{lem:gammaD}. However, the lemma requires a group that is primitive.
For this reason, we will define a minimal block system on which $\Delta$ acts as a primitive permutation group.
Moreover, we do not want that the cover found by \cref{lem:gammaD} depends on the representation of $\Delta$.
For this reason, we use the trick of canonical generating sets and apply the lemma to a group on a linearly ordered set}\\
Define $V^\ord:=\{1,\ldots,|V|\}$.\\
Define $\Delta^\ord:=(\Delta\rho)^\rho=\rho^{-1}\Delta\rho\leq\sym(V^\ord)$.\\
Define $J^\ord:=J^\rho\in\obj(V^\ord)$.\\
\comment{Both $\Delta^\ord$ and $J^\ord$ do not depend on the choice of the representative $\rho$ of $\Delta\rho$}\\
Compute a minimal block system $\CB^\ord:=\{B_1^\ord,\ldots,B_b^\ord\}$ for
$\Delta^\ord[J^\ord]$ acting on $J^\ord=B_1^\ord\cupdot\ldots\cupdot B_b^\ord$.\\
Apply the algorithm from \cref{lem:gammaD} to the primitive group $\Delta^\ord[\CB^\ord]\leq\sym(\CB^\ord)$
of composition-width at most $d:=|V|$.\\
\comment{By using a canonical generating set from \cref{lem:canGen} for $\Delta^\ord$, we ensure that the output of
that algorithm only depends on $\Delta^\ord$ (and not on the representation of $\Delta^\ord$)}\\
Depending on the cases \ref{lem:gammaD1}-\ref{lem:gammaD3} of \cref{lem:gammaD}, we do the following.
\begin{cs}

\case{$|\Delta^\ord[\CB^\ord]|\leq|\CB^\ord|^{\CO(\log_2|V|)}$}~\\
\comment{In this case, the group $\Delta\leq\sym(V)$ acting on the block system is small enough to iterate over all permutations of the blocks}\\
Define $\Psi^\ord:=\stab_{\Delta^\ord}(B_1^\ord,\ldots,B_b^\ord)\leq\sym(V^\ord)$.\\
Decompose $\Delta^\ord$ into left cosets of $\Psi^\ord$ and write
$\Delta^\ord=\bigcup_{\ell\in[s]}\delta^\ord_\ell\Psi^\ord$.\\
\comment{This composition can be computed in time polynomial $|V|$ and in the index $s=|\Delta^\ord[\CB^\ord]|\leq |\CB^\ord|^{\CO(\log_2|V|)}$}\\
Compute $\Theta_\ell\tau_\ell:=\canSetSet(J,L,\alpha,\rho\delta^\ord_\ell\Psi^\ord)$ for each
$\ell\in[s]$ recursively.\\
\comment{The multiplicative cost of the recursion
corresponds to the index $s$ of $\Psi^\ord$ in $\Delta^\ord$,
which is bounded by $s\leq|\CB^\ord|^{\CO(\log_2|V|)}$}\\
Define $\hat J:=\{\Theta_\ell\tau_\ell\mid \ell\in[s]\}$.\\
\comment{We collect the canonical labelings $\Theta_\ell\tau_\ell$ leading to minimal canonical forms of the input}\\
Define $\hat J_{\min}:=\arg\min_{\Theta_\ell\tau_\ell\in \hat J}(J,L,\alpha,\Delta\rho)^{\tau_\ell}\subseteq\hat J$ where the minimum is taken w.r.t.~the ordering ``$\prec$'' from \cref{lem:prec}.\\
Return $\Lambda:=\langle\hat J_{\min}\rangle$.\\
\comment{This is the smallest coset containing all labeling cosets in $\hat J_{\min}$ as defined
in the preliminaries.
The correctness proof for (CL2) is given below the algorithm}

\case{$|V|!< |\CB^\ord|$}~\\
\comment{This case can actually not occur. Since $\Delta^\ord\leq\sym(V^\ord)$, it follows
that $|\Delta^\ord[\CB^\ord]|\leq|\Delta^\ord|\leq |V|!$.
Since $\Delta^\ord$ is transitive on $\CB^\ord$, it follows that $|\CB^\ord|\leq|\Delta^\ord[\CB^\ord]|$.
Therefore, $|\CB^\ord|\leq|V|!$}

\case{there is a sparse cover $\CC_\CB^\ord$ of $\CB^\ord$ with $2\leq|\CC_\CB^\ord|\leq |V|^3$ which is $\Delta^\ord$-invariant}~\\
\comment{We proceed with the covering technique. Observe that the cover we found so far is a cover for $\CB^\ord$
(rather than a cover for $J^\ord$). However, we can easily define a cover for $J^\ord$ as well by taking unions of blocks}\\
Define a sparse cover $\CC^\ord:=\{C_1^\ord,\ldots,C_c^\ord\}$ of $J^\ord=C_1^\ord\cup\ldots\cup C_c^\ord$
where $C_i^\ord:=\bigcup C_{\CB,i}^\ord\subseteq J^\ord$ for each $C_{\CB,i}^\ord\in\CC_\CB^\ord$.\\
\comment{In the next step, we define the cover corresponding to $J$}\\
Define a sparse cover $\CC:=\{C_1,\ldots,C_c\}$ of $J=C_1\cup\ldots\cup C_c$
such that $C_i^\rho=C_i^\ord$ for each $C_i^\ord\in\CC^\ord$.\\
\comment{Observe that $\CC$ does not depend on the choice of the representative $\rho$ of $\Delta\rho$ and is defined in an isomorphism-invariant way}\\
\comment{Next, we will recurse on the cover $\CC$}
\begin{for}{$C_i\in\CC$}
Define $C_{i^*}^\ord\in\CC^\ord$ be the minimal (w.r.t.~``$\prec$'') image of $C_i$ under $\Delta\rho$.\\
Define $\Delta_{C_i}\rho_{C_i}:=\{\lambda\in\Delta\rho\mid C_i^\lambda=C_{i^*}^\ord\}$.\\
\comment{The labeling coset $\Delta_{C_i}\rho_{C_i}$ is essentially a canonical labeling for $(C_i,\Delta\rho)$. Moreover, $\Delta_{C_i}\rho_{C_i}\leq\Delta\rho$
can be computed in polynomial time since the index $(\Delta:\Delta_{C_i})$ is bounded by $|\CC|\leq |V|^3$}\\
Compute $\Theta_i\tau_i:=\canSetSet(C_i,C_i^\alpha,\alpha|_{C_i},\Delta_{C_i}\rho_{C_i})$ recursively.
\end{for}
Define $\CC^\set:=\{\Theta_i\tau_i\mid C_i\in\CC\}$.\\
\comment{We define an ordering according to the isomorphism type of $C_i\in\CC$}\\
Define an \emph{ordered} cover $\C:=(\CC_1,\ldots,\CC_{c'})$ of $\CC^\set=\CC_1\cup\ldots\cup\CC_{c'}$ such that:\\
$(C_k,C_k^\alpha,\alpha|_{C_k},\Delta_{C_k}\rho_{C_k})^{\tau_k}\prec
(C_{k'},C_{k'}^\alpha,\alpha|_{C_{k'}},\Delta_{C_{k'}}\rho_{C_{k'}})^{\tau_{k'}}$,
if and only if $\Theta_k\tau_k\in\CC_i$ and $\Theta_{k'}\tau_{k'}\in\CC_j$ for some $i,j\in[c']$ with $i<j$.\\
Compute $\Lambda_i:=\canSet(\CC_i)$ for each $\CC_i,i\in[c']$ using \cref{lem:canSet0}.\\
\comment{
Since $|\CC_i|\leq|\CC^\set|=|V|^3$,
the algorithm from \cref{lem:canSet0} runs in the desired time bound, i.e., $2^{\polylog|V|}$}\\
Compute and return $\Lambda
:=\canObj((\Lambda_1,\ldots,\Lambda_{c'}))$ using \cref{cor:canObj0}.\\
\comment{
Since $(\Lambda_1,\ldots,\Lambda_{c'})$ is a tuple consisting of atoms, no set is involved in this object.
Therefore, the algorithm from \cref{cor:canObj0} runs in the desired time bound, i.e., $2^{\polylog|V|}$}
\end{cs}
\end{cs}

Condition (CL1) holds as usual.

\begin{cl2}
We have to show that $\Lambda=\aut((J,L,\alpha,\Delta\rho))\pi$ for $\pi\in\Lambda$.
In the intransitive case, we have that $\Lambda_i=\aut((J_i,L_i,\alpha|_{J_i},\Delta\rho))\pi_i$ for both $i=1,2$ by induction.
Then, Condition (CL2) follows from Condition (CL2) of \cref{lem:canInt}.

Consider the case in which Option 1 of \cref{lem:gammaD} holds.
The inclusion $\aut((J,L,\alpha,\Delta\rho))\pi\subseteq\Lambda$ already follows from the isomorphism invariance (Condition (CL1))
of this algorithm, i.e.,
$\can(\CX)=\sigma\can(\CX^\sigma)=\sigma\can(\CX)$ for $\sigma\in\aut(\CX)$ implies
that $\aut(\CX)\pi\subseteq\can(\CX)$ for some $\pi\in\can(\CX)$.
We have to show the reversed inclusion.
By induction, we have that $\Theta_\ell=\aut((J,L,\alpha,\rho\delta^\ord_\ell\Psi^\ord))\subseteq\aut((J,L,\alpha,\Delta\rho))$.
Therefore, we also have the inclusion $\Lambda=\langle\hat J_{\min}\rangle\subseteq\aut((J,L,\alpha,\Delta\rho))\pi$.

The cover case (Option 3) is similar to the recursion in the algorithm of \cref{theo:canRel}.
\end{cl2}

\begin{runtime}
Let $k:=\orb_{J^\ord}(\Delta^\ord)$ be the size of the largest $\Delta^\ord[J^\ord]$-orbit.
Let $c\in\NN$ be the constant from \cref{lem:gammaD} that is hidden in the $\CO$-notation in the exponent.
We claim that the maximum number of recursive calls $R(J,\Delta^\ord)$ is at most $T:=k^{4c\log_2|V|}|J|^2$.
In the intransitive case, this is easy to see by induction:
\begin{align*}
R(J,\Delta^\ord)&= 1+R(J_1,\Delta^\ord)+R(J_2,\Delta^\ord)\\
&\overset{\mathclap{\text{induction}}}{\leq}
1+k^{4c\log_2|V|}(|J_1|^2+|J_2|^2)\leq T.
\end{align*}
We consider the transitive case in which Option 1 of \cref{lem:gammaD} holds.
Since $\Delta[J]$ is transitive, it holds $k=|J^\ord|$.
The recursive calls are done for the subgroup $\Psi^\ord\leq\Delta^\ord$
of index $s\leq |\CB^\ord|^{c\log_2|V|}$.
Moreover, we reduce orbit size for the recursive calls and have $\orb_{J^\ord}(\Psi^\ord)\leq\frac{|J^\ord|}{|\CB^\ord|}$.
This leads to the recurrence
\begin{align*}
R(J,\Delta^\ord)&= 1+ s\cdot R(J,\Psi^\ord)\\
&\overset{\mathclap{\text{induction}}}{\leq}
1+|\CB^\ord|^{c\log_2|V|}\cdot\left(\frac{|J^\ord|}{|\CB^\ord|}\right)^{4c\log_2|V|}|J|^2\leq T.
\end{align*}
In the cover case, we obtain
\begin{align*}
R(J,\Delta^\ord)&= 1+\sum_{C_i\in \CC} R(C_i,\Delta_{C_i}^\ord)\\
&\overset{\mathclap{\text{induction}}}{\leq}
1 + \sum_{C_i\in \CC} |C_i|^{4c\log_2|V|}|J|^2\\
&\leq 1+|V|^3 \frac{|J^\ord|^{4c\log_2|V|}}{|V|^4}|J|^2\leq T&&\text{(using $|\CC|\leq |V|^3$ and $|C_i|\leq\frac{1}{2}|J^\ord|$).}
\end{align*}

We consider the running time of one single recursive call.
The algorithm from \cref{cor:canObj0} runs in time $2^{\polylog|V|}$.
Therefore, the total running time is bounded by $(|V|+|J|)^{\polylog|V|}$.
\end{runtime}
\end{proof}

We consider coset-labeled hypergraphs, which were introduced in \cite{DBLP:conf/icalp/GroheNSW18}.
A coset-labeled hypergraph is essentially a hypergraph for which a labeling coset is given for each hyperedge.
This problem generalizes the canonical labeling problem for hypergraphs, but is not that general as \cref{prob:CL:Set}.

\problem{\canSetHyper\label{prob:CL:SetHyper}}
{$(H,L,\alpha)\in\obj(V)$ where
$H=\{S_1,\ldots,S_t\}$,
$L=\{\Lambda_1,\ldots,\Lambda_t\}$,
$S_i\subseteq V,\Lambda_i\leq\lab(V)$ for
all $i\in[t]$,
$\alpha:H\to L$ is a function with $\alpha(S_i)=\Lambda_i$
and $V$ is an unordered set}
{(H,L,\alpha)}
{(H^\phi,L^\phi,\alpha^\phi)}
{\{\sigma\in\sym(V)\mid\exists\psi\in\sym(t)
\forall
i\in[t]:(S_i,\Lambda_i)^\sigma
=(S_{\psi(i)},\Lambda_{\psi(i)})\}}

Remember that we already have an algorithm that canonizes hypergraphs.
Therefore, the previous lemma implies that we can also canonize hypergraphs for which a labeling coset
is given for each hyperedge.

\begin{lem}\label{lem:canSetHyper}
A function $\canSetHyper$
for \cref{prob:CL:SetHyper} can be computed in time
$(|V|+|H|)^{\polylog|V|}$.
\end{lem}

\begin{proof}
Assume we are given an instance $(H,L,\alpha:H\to L)$.
First, we compute a canonical labeling $\Delta\rho:=\canHyper(H)$ using \cref{theo:canHyper}.
Let $J:=\{\Delta_1\rho_1,\ldots,\Delta_t\rho_t\}$ where $\Delta_i\rho_i$ is a canonical labeling for $S_i$
for each $i\in[t]$.
The set $J$ is polynomial-time computable since each $\Delta_i$ is a direct product of two symmetric groups $\sym(S_i)$
and $\sym(V\setminus S_i)$.
We define $\alpha_J:J\to L$ by setting $\alpha_J(\Delta_i\rho_i):=\alpha(S_i)=\Lambda_i$.
Observe that $\Delta=\aut(H)=\aut(J)$.
We compute and return the canonical labeling $\Lambda:=\canSetSet(J,L,\alpha_J,\Delta\rho)$ using \cref{lem:canSetSet}.
\end{proof}

\section{Canonization of Sets and Objects}\label{sec:sets}

We recall the central problem that we want to solve.

\setcounter{theo}{7}

\problem{\canSet\label{prob:CL:Set}}
{$J\in\obj(V)$ where $J=\{\Delta_1\rho_1,\ldots,\Delta_t\rho_t\}$,
$\Delta_i\rho_i\leq\lab(V)$ for all $i\in[t]$
and $V$ is an unordered set}
{(J)}
{(J^\phi)}
{\aut(J)}

\setcounter{theo}{21}

\paragraph{Giant Representations}
A homomorphism $h:\Delta\to\sym(W)$ is called a \emph{giant representation} if the image of $\Delta$ under $h$ is a giant, i.e., $\alt(W)\leq h(\Delta)\leq\sym(W)$

\begin{theo}\label{theo:canSet}
A function $\canSet$
solving \cref{prob:CL:Set} can be computed in time
$(|V|+|J|)^{\polylog|V|}$.
\end{theo}

\paragraph{Proof Outline}
For the purpose of recursion, our main algorithm $\canSet$ needs some additional input parameters.
The input of the main algorithm is a tuple $(J,A,\Delta^\ord,g^\ord)$ consisting of the following input parameters.
\begin{itemize}
  \item $J$ is a set consisting of labeling cosets,
  \item $A\subseteq V$ is a subset which
is $\Delta_i$-invariant for all $\Delta_i\rho_i\in J$,
\item $\Delta^\ord\leq\sym(V^\ord)$ is a group over the linearly ordered set $V^\ord=\{1,\ldots,|V|\}$, and
\item $g^\ord:\Delta^\ord\to\sym(W^\ord)$ is a giant representation where $W^\ord=\{1,\ldots,|W^\ord|\}$
is a linearly ordered set.
\end{itemize}

We will define the additional parameters besides $J$ in an isomorphism-invariant way.
The additional parameters are used for recursion and can provide information, however,
canonical labelings for an instance $(J,A,\Delta^\ord,g^\ord)$
correspond to canonical labelings for $J$.

Initially, we set $A:=V$
and we let $g^\ord:=\bot$ be undefined.
Furthermore, we require three properties that hold for our input instance:
\begin{enumerate}
  \item[(A)]\label{prop:A} $(\Delta_i\rho_i)|_{V\setminus A}=(\Delta_j\rho_j)|_{V\setminus A}$ for all $\Delta_i\rho_i,\Delta_j\rho_j\in J$, and
  \item[(B)]\label{prop:D} for all $\Delta_i\rho_i\in J$
  it holds that $(\Delta_i\rho_i)^{\rho_i}=\Delta^\ord$
  and there is a subset $A^\ord\subseteq\{1,\ldots,|V|\}$ such that for all $\Delta_i\rho_i\in J$ it holds that
  $A^{\rho_i}=A^\ord$.
  \item[(g)]\label{prop:g} if $g^\ord\neq\bot$ (i.e., $g^\ord$ is defined), then $g^\ord:\Delta^\ord\to\sym(W^\ord)$ is a giant representation
where $|W^\ord|> 2+\log_2 |V|$ and $|W^\ord|$ is greater than some absolute constant
and $\Delta^\ord$ is transitive on $A^\ord$ and $\Delta^\ord_{(A^\ord)}\leq\ker(g^\ord)$ (the pointwise stabilizer of $A^\ord$ in $\Delta^\ord$).
\end{enumerate}

With the initial choice of $A:=V$ Property (A) holds.
Initially, $g^\ord:=\bot$ is undefined and therefore Property (g) also holds.
Furthermore, we can assume that Property (B) holds,
otherwise we can define an \emph{ordered} partition of $J$ and recurse on that, i.e.,

\begin{cs}
\case{Property (B) is not satisfied}~\\
Define $A_i^\ord:=A^{\rho_i}$ for some $\Delta_i\rho_i\in J$.\\
\comment{We will define an ordered partition of $J$ according to the ordering ``$\prec$'' from \cref{lem:prec}
that is defined on the elements $(A_i^\ord,\Delta_i^\ord)$}\\
Define an \emph{ordered} partition $\CJ:=(J_1,\ldots,J_s)$ of $J= J_1\cupdot\ldots\cupdot J_s$ such
that:\\
$(A_i^\ord,\Delta_i^\ord)\prec (A_j^\ord,\Delta_j^\ord)$,
if and only if
$\Delta_i\rho_i\in J_p$ and $\Delta_j\rho_j\in J_q$ for some $p,q\in[s]$ with $p<q$.\\
Recursively compute $\Lambda_i:=\canSet(J_i,A,\Delta^\ord,g^\ord)$ for each $i\in[s]$.\\
Return $\Lambda:=\canObj((\Lambda_1,\ldots,\Lambda_s))$ using \cref{cor:canObj0}.\\
\comment{Since there is no set involved in the tuple $(\Lambda_1,\ldots,\Lambda_s)$, the algorithm
from \cref{cor:canObj0} runs in time $2^{\polylog|V|}(|V|+|J|)^{\CO(1)}$}
\end{cs}

Property (B) also implies that $A$ can be defined out of $J$ in an isomorphism-invariant way.
In particular, $\aut(J,A)=\aut(J)$.

\paragraph{The Measurement of Progress}
By $\orb_{A^\ord}(\Delta^\ord)$, we denote the size of the largest $\Delta^\ord$-orbit on $A^\ord$.
Let $\delta(g^\ord)=1$ if $g^\ord$ is defined and let $\delta(g^\ord)=0$ if $g^\ord=\bot$ is undefined.
We will show that the number of recursive calls $R(J,A,\Delta^\ord,g^\ord)$ of our main algorithm is at most
\begin{align*}\label{runtime}\tag{$T$}
T:=2^{\log_2(|V|+2)^3(2\cdot\log_2(|V|+4)\cdot\log_2|J|+2\cdot\log_2(\orb_{A^\ord}(\Delta^\ord))}\cdot |J|^2\cdot |A|\cdot |V|^{2-2\delta(g^\ord)}.
\end{align*}

The function looks quite complicated, but there are only a few properties that are of importance.
We list these properties.
First, observe that $T\leq (|V|+|J|)^{\polylog|V|}$.
Moreover, if we can show that the number of recursive calls $R$ of our
main algorithm satisfies the recurrences listed below, then it holds that $R\leq T$.
We will allow the following types of recursions for the main algorithm
which we refer to as \emph{progress}.
\begin{itemize}
  \item We split $J$ while preserving $\Delta^\ord$ and $g^\ord$, i.e.,
  \begin{align}\label{prog:linJ}\tag{Linear in $J$}
  R(J,A,\Delta^\ord,g^\ord)&=1+\sum_{i\in[s]}R(J_i,A,\Delta^\ord,g^\ord),
  \end{align}
  where $J=J_1\cupdot\ldots\cupdot J_s$.
  \item We reduce the size of $A$ while preserving $\Delta^\ord$ and $g^\ord=\bot$, i.e.,
  \begin{align}\label{prog:linA}\tag{Linear in $A$}
  R(J,A,\Delta^\ord,\bot)&=1+R(J,A',\Delta^\ord,\bot),
  \end{align}
  where $|A'|<|A|$.
  \item At a multiplicative cost of $2^{\log_2(p)+\log_2(|V|)^4}$, we divide the size $|J|$ by $p$
  and at a multiplicative cost of $2^{\log_2(|V|)^4}$, we reduce the size of $|J|$ to $p$
  while resetting the other parameters $A:=V$ and $g:=\bot$, i.e.,
  \begin{align}\label{prog:J}\nonumber
  R(J,A,\Delta^\ord,g^\ord)&=1+2^{\log_2(p)+\log_2(|V|)^4}\cdot R(J',A,\Delta^\ord,g^\ord)\\
  &\hspace{0.6cm}+2^{\log_2(|V|)^4}\cdot R(J'',V,\Delta^\ord,\bot),\tag{In $J$}
  \end{align}
  where $|J'|\leq\frac{1}{p}|J|$ and $|J''|\leq p$ for some $p\in\NN$ with $1<p\leq\frac{1}{2}|J|$.
  \item At a multiplicative cost of $2^{\log_2(|V|)^3}$, we halve the size of the largest $\Delta^\ord$-orbit while resetting $g^\ord:=\bot$, i.e.,
  \begin{align}\label{prog:A}\tag{In $\Delta^\ord$}
  R(J,A,\Delta^\ord,g^\ord)=1+ 2^{\log_2(|V|)^3}\cdot R(\hat J,A,\Psi^\ord,\bot),
  \end{align}
  where $|\hat J|\leq|J|$ and $\orb_{A^\ord}(\Psi^\ord)\leq\frac{1}{2}\orb_{A^\ord}(\Delta^\ord)$.
  \item At a multiplicative cost of $|V|$, we find a giant representation, i.e.,
  \begin{align}\label{prog:giant}\tag{In $g^\ord$}
  R(J,A,\Delta^\ord,\bot)=1+|V|\cdot R(\hat J,A,\Psi^\ord,g^\ord),
  \end{align}
  where $|\hat J|\leq |J|$ and $\orb_{A^\ord}(\Psi^\ord)\leq \orb_{A^\ord}(\Delta^\ord)$ and $g^\ord$ is defined.
\end{itemize}

The main algorithm calls the subroutines $\toJohnson$, $\processJohnson$
and $\processAut$ described in \cref{lem:toJohnson}, \cref{lem:processJohnson} and \cref{lem:processAut}, respectively.
These subroutines in turn use the subroutines $\recPart$ and $\toGroup$ given in \cref{lem:recPart} and \cref{lem:toGroup}.
We will ensure that progress is achieved whenever the main algorithm is called recursively.
See \cref{fig:flow} for a flowchart diagram.

\begin{figure}[h]
\begin{center}
\scalebox{.98}{\tikzstyle{block} = [rectangle, draw, text width=4.8cm, text centered, rounded corners, minimum height=2.5em]
\tikzstyle{decision} = [diamond, draw]
\begin{tikzpicture}
\node (n1) at (0,0) [block]  {Main algorithm $\canSet$};
\node (d) at (0,-3) [decision,align=center] {$g^\ord$ is\\ defined?};
\node (n2) at (-3,-5) [block]  {$\toJohnson$};
\node (n3) at (3,-5) [block]  {$\processJohnson$};
\node (n3b) at (7,-8) [block]  {$\processAut$};

\draw[-{Triangle[scale=1.5]}] (n1) to (d);
\draw[-{Triangle[scale=1.5]}] (d) to node[yshift=0.4cm] {No} (n2);
\draw[-{Triangle[scale=1.5]}] (d) to node[yshift=0.4cm] {Yes} (n3);
\draw[-{Triangle[scale=1.5]}] (n3) to node[xshift=-1.5cm,align=center] {Certificate\\ found} (n3b);
\draw[-{Triangle[scale=1.5]},bend left=30] (n2) to node[xshift=-1.5cm,text width=2.8cm] {Progress via\\ \cref{prog:linJ}\\ or \cref{prog:linA}\\ or \cref{prog:J}\\ or \cref{prog:A}\\ or \cref{prog:giant}} (n1);
\draw[-{Triangle[scale=1.5]},bend right=30] (n3) to node[xshift=2cm,text width=3cm] {Progress via\\ \cref{prog:linJ}\\ or \cref{prog:J}\\ or \cref{prog:A}} (n1);
\draw[-{Triangle[scale=1.5]},bend right=50] (n3b) to node[xshift=2.1cm,text width=2.21cm] {Progress via\\ \cref{prog:linJ}\\ or \cref{prog:J}\\ or \cref{prog:A}} (n1);
\end{tikzpicture}}
\caption{Flowchart of the algorithm for \cref{theo:canSet}\label{fig:flow}.}
\end{center}
\end{figure}
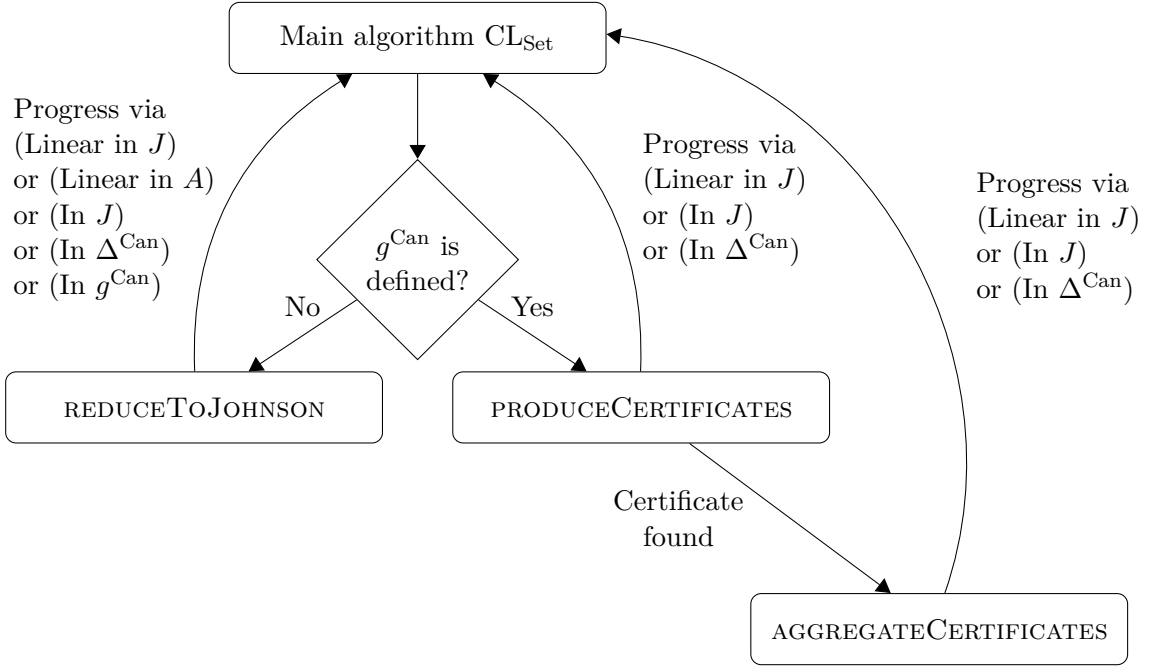

\paragraph{Equipartitions and Partition Families}
An \emph{equipartition} is a partition $\CP$ in which all parts $P_i\in\CP$ have the same size $|P_i|$.
A \emph{partition family} of $\CX\in\obj(V)$ is a family $\PP:=\{\CP_k\}_{k\in K}$ where each member $\CP_k=\{P_{k,1},\ldots,P_{k,p_k}\}$ is a partition of
$\CX=P_{k,1}\cupdot\ldots\cupdot P_{k,p_k}$.
A partition family $\PP$ is called \emph{trivial} if all partitions $\CP_k\in\PP$ are trivial.
The notion of partition families generalizes the notion of covers.
More precisely, for each cover $\CC=\{C_1,\ldots,C_c\}$ of $\CX$ we can define a partition family
$\PP:=\{\CP_i\}_{i\in [c]}$ by setting $\CP_i:=\{C_i,\CX\setminus C_i\}$ for each $i\in[c]$.
In this case, we say that $\PP$ is \emph{induced} by $\CC$.

The next lemmas shows that we can exploit partition families $\{\CP_k\}_{k\in K}$ of $J$ algorithmically.

\begin{lem}\label{lem:recPart}
There is an algorithm $\recPart$ that gets a input a pair $(\CX,\PP)$
where $\CX=(J,A,\Delta^\ord,g^\ord)$ is a tuple
for which Property (A), (B) and (g) hold
and $\PP=\{\CP_k\}_{k\in K}$ is a non-trivial partition family.
In time $2^{\polylog(|V|+|K|)}$, the algorithm reduces the canonical labeling problem of the instance $\CX$ to canonical labeling of either
\begin{enumerate}
  \item\label{lem:recPartTwo} two instances $(J_1,A,\Delta^\ord,g^\ord)$ and $(J_2,A,\Delta^\ord,g^\ord)$ with $|J_1|+|J_2|=|J|$, or
  \item\label{lem:recPartP} $|K|p$-many instances $(J_{k,i},A,\Delta^\ord,g)$
  of size $|J_{k,i}|\leq \frac{1}{p}|J|$ and to additionally $|K|$-many instances $(J_k,V,\Delta_k^\ord,\bot)$ of size $|J_k|\leq p$ for some $p\in\NN$ with $1<p\leq\frac{1}{2}|J|$.
\end{enumerate}
\end{lem}
In case that $|K|$ is quasipolynomially bounded (or more precisely, bounded by $2^{\log_2(|V|)^4}$) the the lemma facilitates a recursion that leads to progress \cref{prog:J}.

In the following, we sketch the idea how to exploit a partition family.

\paragraph{The Partition-Family Technique}
Extending the covering technique, we suggest a technique for handling partition families that we use to prove \cref{lem:recPart}.
In this setting, we assume that we are given a set $J\in\obj(V)$ consisting of labeling cosets for which we can define a non-trivial partition family $\PP=\{\CP_k\}_{k\in K}$
in an isomorphism-invariant way.
We do not require any bound on the size of the partitions $\CP_k$.
The goal is the computation of a canonical labeling of $J$ using an efficient recursion.

Let $\PP':=\{\CP_k\in\PP\mid \CP_k\text{ is non-trivial}\}$ be the non-empty set of non-trivial partitions.
We can assume that $\PP'=\PP$, otherwise we continue with $\PP:=\PP'$.
We distinguish between two cases.

Case 1: There is a partition $\CP_k=\{P_{k,1},\ldots,P_{k,p_k}\}\in\PP$ that is an equipartition of $J$.
Again, we assume each $\CP_k\in\PP$ is an equipartition, otherwise consider the partition family $\PP:=\{\CP_k\in\PP\mid\CP_k\text{ is an equipartition}\}$.
Moreover, we can assume that all parts have the same size $|P_{k,i}|$ even across all equipartitions,
otherwise we would consider a subset $\PP:=\arg\min_{\CP_k\in\PP}|\CP_k|$.
Now, we use recursion and compute a canonical labeling $\Theta_{k,i}\tau_{k,i}$ for each part $P_{k,i}\subseteq J$.
For simplicity, we assume that all parts $P_{k,i}$ and $P_{k,j}$ are isomorphic.
Let $\CP_k^\set:=\{\Theta_{k,i}\tau_{k,i}\mid P_{k,i}\in\CP_k\}$.
By object replacement (\cref{lem:rep}), a canonical labeling for $\CP_k^\set$
also defines a canonical labeling for $\CP_k$.
To compute a canonical labeling $\Theta_k\tau_k$ for $\CP_k^\set$, we use recursion again.

Next, we compute a canonical labeling $\Theta\tau$ for $J$.
We choose the canonical labelings which lead to a minimal canonical form for $J$.
More precisely, let $K:=\arg\min_{\Theta_k\tau_k} J^{\tau_k}$ where the minimum is taken w.r.t.~the ordering ``$\prec$'' from \cref{lem:prec}.
We let $\Theta\tau$ be the labeling coset that is generated by all $\Theta_k\tau_k$ for $k\in K$.

We analyze the recurrence of this approach.
Let $p\in\NN$ be the size $p=|\CP_k|$ which is uniform over all partitions $\CP_k\in\PP$.
We have $|K|p$-many recursive calls for instances $P_{k,i}$ of size $\frac{1}{p}|J|$.
After that, we have $|K|$-many recursive calls for instances $\CP_k^\set$ of size $p$.
In case that $|K|\leq 2^{\polylog|V|}$,
this recurrence is progress via \cref{prog:J}.

Case 2: Each partition $\CP_k\in\PP$ is not an equipartition.
If a partition $\CP_k$ is not an equipartition of $J$, then $\CP_k$ induces a non-trivial \emph{ordered} partition
$\tilde\CP_k:=(P_k^1,\ldots,P_k^{|J|})$ of $J$ where
$P_k^x:=\bigcup_{P_{k,i}\in \CP_k,|P_{k,i}|=x}P_{k,i}$ for $x\in\{1,\ldots,|J|\}$.
Moreover, let $P_k^*:=P_k^x$ where $x\in\NN$ is the smallest number such that $1\leq |P_k^x|\leq\frac{1}{2}|J|$.
We define $J^*:=\bigcup_{\CP_k\in\PP}P_k^*$.

In the case in which $J^*\subsetneq J$, we found a non-trivial \emph{ordered} partition of
$J=J^*\cupdot J\setminus J^*$ and proceed with the partitioning technique.
This will lead to progress via \cref{prog:linJ}.

In the other case in which $J^*=J$, we found a cover of $J=\bigcup_{\CP_k\in\PP}P_k^*$ and proceed with the covering technique.
In case that $|K|\leq 2^{\polylog|V|}$, this will ensure progress via \cref{prog:J}.

\begin{proof}[Proof of \cref{lem:recPart}]
\algorithmDAN{\recPart(J,A,\Delta^\ord,g^\ord,\PP)}~\\
\comment{We simplify to the case in which all partitions $\CP_k\in\PP$ are non-trivial}\\
Define $\PP:=\{\CP_k\in\PP\mid\CP_k\text{ is non-trivial}\}$.
\begin{cs}
\case{there is a partition $\CP_k\in\PP$ that is an equipartition}~\\
\comment{We simplify to the case in which all $\CP_k\in\PP$ are equipartitions}\\
Define $\PP:=\{\CP_k\in\PP\mid\CP_k\text{ is an equipartition}\}$.\\
\comment{We simplify to the case in which $|\CP_k|$ are equal for all $\CP_k\in\PP$}\\
Define $\PP:=\arg\min_{\CP_k\in\PP}|\CP_k|$.\\
\comment{Now, there is a number $p\in\NN$ such that $p=|\CP_k|$ for all partitions $\CP_k\in\PP$}
\begin{for}{$\CP_k\in\PP$}
\comment{We show how to compute a canonical labeling $\Theta_k\tau_k$ for the pair $(J,\CP_k)$ for each partition $\CP_k\in\PP$.
Roughly speaking, the instance $(J,\CP_k)$ can be seen as an individualization of $J$ obtained by individualizing one partition
$\CP_k\in\PP$}\\
Recursively compute $\Theta_{k,i}\tau_{k,i}:=\canSet(P_{k,i},A,\Delta^\ord,g^\ord)$ for each part $P_{k,i}\in\CP_k$.\\
\comment{We have a multiplicative cost of $|\PP|\cdot p$
and recursive instances of size $|P_{k,i}|=|J|/p$}\\
Define $\CJ_k^\set:=\{(\Theta_{k,i}\tau_{k,i},P_{k,i}^{\tau_{k,i}})\mid P_{k,i}\in\CP_k\}$.\\
\comment{In previous algorithms, we computed a canonical labeling for $\CJ_k^\set$ by using \cref{cor:canObj0}.
However, in this case, the size $|\CJ_k^\set|=p$ might not be bounded by a quasipolynomial.
For this reason, we use a recursive approach to compute a canonical labeling for $\CJ_k^\set$.
First, we define an ordering according to the isomorphism type of the parts $P_{k,i}\in\CP_k$}\\
Define an \emph{ordered} partition $\JJ_k^\set:=(\CJ_{k,1}^\set,\ldots,\CJ_{k,m_k}^\set)$ of $\CJ_k^\set=\CJ_{k,1}^\set\cupdot\ldots\cupdot\CJ_{k,m_k}^\set$
such that:\\
$P_{k,i}^{\tau_{k,i}}\prec
P_{k,j}^{\tau_{k,j}}$, if and only if
$(\Theta_{k,i}\tau_{k,i},P_{k,i}^{\tau_{k,i}})\in\CJ_{k,p}^\set$ and $(\Theta_{k,j}\tau_{k,j},P_{k,j}^{\tau_{k,j}})\in \CJ_{k,q}^\set$ for some
$p<q\in[m_k]$.\\
Define $\pi_1(\CJ_{k,\ell}^\set):=\{\Theta_{k,i}\tau_{k,i}\mid (\Theta_{k,i}\tau_{k,i},P_{k,i}^{\tau_{k,i}})\in\CJ_{k,\ell}^\set\}$
for each $\ell\in[m_k]$.\\
\comment{The ordering ensures that Property (B) holds for all instances $\pi_1(\CJ_{k,p}^\set)$
for some group $\Theta_{k,p}^\ord$}\\
Recursively compute $\Lambda_{k,\ell}:=\canSet(\pi_1(\CJ_{k,\ell}^\set)),V,\Theta_{k,\ell}^\ord,\bot)$ for each $\ell\in[m_k]$.\\
\comment{
For our running time $T(|J|)$ given in \cref{runtime}, we have that $\sum_{\ell\in[m_k]} T(|\CJ_{k,\ell}^\set|)\leq
T(\sum_{\ell\in[m_k]}|\CJ_{k,\ell}^\set|)=T(|\CJ_k^\set|)$.
Therefore, in the worst case, we have $m_k=1$ and $|\CJ_{k,1}^\set|=|\CJ_k^\set|=p$.
Therefore, we have a multiplicative cost of $|\PP|\cdot m_k=|\PP|$
and recursive instances of size $p$}\\
Compute $\Theta_k\tau_k:=\canSet((\Lambda_{k,1},\ldots,\Lambda_{k,m_k}))$ using \cref{cor:canObj0}.\\
\comment{Observe that $\Theta_k\tau_k$ is a canonical labeling for $(J,\CP_k)$}
\end{for}
Define $\PP^\set:=\{\Theta_k\tau_k\mid\CP_k\in\PP\}$.\\
\comment{To obtain a canonical labeling $\Lambda$ for $J$ for the given $\Theta_k\tau_k$ we will proceed as follows.
We compare the canonical forms $J^{\tau_k}$ that we obtain for each individualized
partition $\CP_k\in\PP$.
Then, we collect the canonical labelings leading to a minimal canonical form w.r.t.~``$\prec$''}\\
Define $\PP_{\min}^\set:=\arg\min_{\Theta_k\tau_k\in\PP^\set}J^{\tau_k}\subseteq\PP^\set$ where the minimum is taken w.r.t.~the ordering ``$\prec$'' from \cref{lem:prec}.\\
Return $\Lambda:=\langle\PP_{\min}^\set\rangle$.\\
\comment{This is the smallest coset that contains all labeling cosets in $\PP_{\min}^\set$ as defined
in the preliminaries. The correctness proof for (CL2) is similar to the (CL2)-proof of \cref{lem:canSetSet}}
\end{cs}
\comment{Now, each partition $\CP_k\in\PP$ of $J$ is not an equipartition}
\begin{for}{$\CP_k\in\PP$}
Define an \emph{ordered} partition
$\tilde\CP_k:=(P_k^1,\ldots,P_k^{|J|})$ of $J$ where
$P_k^x:=\bigcup_{P_{k,i}\in \CP_k,|P_{k,i}|=x}P_{k,i}$ for $x\in\{1,\ldots,|J|\}$.\\
Define $P_k^*:=P_k^x\subseteq J$ where $x\in\NN$ is the smallest number such that $1\leq |P_k^x|\leq\frac{1}{2}|J|$.
\end{for}
Define $P^*:=\bigcup_{\CP_k\in\PP}P_k^*$.
\begin{cs}
\case{$P^*\subsetneq J$}~\\
\comment{We found an ordered partition of $J$ and proceed with the partitioning technique}\\
Define an \emph{ordered} partition $\CP=(P^*,P^\circ)$ of $J=P^*\cupdot P^\circ$ where $P^\circ:=J\setminus P^*$.\\
\comment{The partition is non-trivial since $P^*$ is non-empty by the definition of each part $P_k^*\subseteq J$}\\
Recursively compute $\Lambda_1:=\canSet(P^*,A,\Delta^\ord,g^\ord)$.\\
Recursively compute $\Lambda_2:=\canSet(P^\circ,A,\Delta^\ord,g^\ord)$.\\
\comment{We have that $|P^*|+|P^\circ|=|J|$ and therefore Option \ref{lem:recPartTwo} of \cref{lem:recPart} is satisfied}\\
Compute and return $\Lambda:=\canObj((\Lambda_1,\Lambda_2))$ using \cref{lem:canInt} or \cref{cor:canObj0}.\\
\comment{The algorithm from \cref{lem:canInt} and \cref{cor:canObj0} runs in time $2^{\polylog|V|}$}

\case{$P^*= J$}~\\
\comment{We found a sparse cover of $J$ and proceed with the covering technique}\\
Define a sparse cover $\CC:=\{C_k\mid \CP_k\in\PP\}$ of $J=\bigcup_{k\in K}C_k$ where $C_k:=P_k^*$.\\
For each $C_k\in\CC$, compute $\Theta_k\tau_k:=\canSet(C_k,A,\Delta^\ord,g^\ord)$ recursively.\\
\comment{We have a multiplicative cost of $|\PP|$
and recursive instances of size $|C_k|\leq\frac{1}{2}|J|$
and therefore Option \ref{lem:recPartP} of \cref{lem:recPart} is satisfied}\\
Define $\PP^\set:=\{\Theta_k\tau_k\mid \CP_k\in\PP\}$.\\
Define an \emph{ordered} cover $\C:=(\CC_1,\ldots,\CC_c)$ of $\PP^\set=\CC_1\cup\ldots\cup\CC_c$ such that:\\
$(C_k)^{\tau_k}\prec (C_{k'})^{\rho_{k'}}$, if and only if $\Theta_k\tau_k\in\CP_i$ and $\Theta_{k'}\tau_{k'}\in\CP_j$ for some $i,j\in[c]$ with $i<j$.\\
\comment{In fact, $\CC$ might not be a partition since $\Theta_k\tau_k=\Theta_{k'}\tau_{k'}$ for $k\neq k'$ might hold}\\
Compute $\Lambda_i:=\canSet(\CC_i)$ for each $\CC_i,i\in[c]$ using \cref{lem:canSet0}.\\
\comment{Since $|\CC_i|\leq|\PP^\set|=|\PP|$, the algorithm from \cref{lem:canSet0} runs in time $2^{\polylog(|V|+|\PP|)}$}\\
Compute and return $\Lambda
:=\canObj((\Lambda_1,\ldots,\Lambda_c))$ using \cref{cor:canObj0}.\\
\comment{
Since $(\Lambda_1,\ldots,\Lambda_c)$ is a tuple consisting of atoms, no set is involved in this object.
Therefore, the algorithm from \cref{cor:canObj0} runs in time $2^{\polylog|V|}(|V|+|\PP|)^{\CO(1)}$}
\end{cs}
\end{proof}

\paragraph{Relative Minimal Base Size}
Recall that the pointwise stabilizer of a subset $X\subseteq V$
in a group $\Delta\leq\sym(V)$ is denoted by $\Delta_{(X)}$.
The \emph{minimal base size} of a group $\Delta\leq\sym(V)$ relative to a subgroup $\Psi\leq\Delta$, denoted by $\pid(\Delta,\Psi)$, is
the smallest cardinality $|X|$ among all subsets $X\subseteq V$ such that $\Delta_{(X)}\leq\Psi$.

\begin{exa}\label{exa:pid}
We give some examples.
\begin{enumerate}
  \item The minimal base size of $\Delta$ is defined as $\pidb(\Delta):=\pid(\Delta,1)$ where $1\leq\sym(V)$ denotes the trivial group.
  It can easily be seen that $\pid(\Delta,\Psi)\leq \pidb(\Delta)\leq\log_2|\Delta|$.
  \item\label{exa:pid1} Let $\Psi:=\stab_{\Delta}(A)$ where $A\subseteq V$.
We show that $\pid(\Delta,\Psi)\leq \log_2(\Delta:\Psi)$.
We assume that $\Psi<\Delta$, otherwise $\pid(\Delta,\Psi)=0=\log_2(\Delta:\Psi)$.
Since the $\Psi$-orbit partition is a refinement of the $\Delta$-orbit, there is a
$\Psi$-orbit $U$ and a $\Delta$-orbit $W$ with $U\subsetneq W$ and $|U|\leq\frac{1}{2}|W|$.
Let $v\in U$. It holds that $(\Delta:\Psi)\cdot(\Psi:\Psi_{(v)})=(\Delta:\Delta_{(v)})\cdot(\Delta_{(v)}:\Psi_{(v)})$.
Moreover, $(\Psi:\Psi_{(v)})=|U|$ and $(\Delta:\Delta_{(v)})=|W|$.
Therefore, $(\Delta_{(v)}:\Psi_{(v)})\leq\frac{1}{2}(\Delta:\Psi)$.
By induction on the index, it holds that $\pid(\Delta,\Psi)\leq\pid(\Delta,\Delta_{(v)})+\pid(\Delta_{(v)},\Psi_{(v)})\leq
1+\log_2(\Delta_{(v)}:\Psi_{(v)})\leq 1+\log_2(\frac{1}{2}(\Delta:\Psi))=\log_2(\Delta:\Psi)$.
\item\label{exa:pid2} Let $\Psi:=\stab_{\Delta}(B_1,\ldots,B_b)$ where $\CB:=\{B_1,\ldots,B_b\}$ is a partition of $V=B_1\cupdot\ldots\cupdot B_b$.
We show that $\pid(\Delta,\Psi)\leq 2\cdot\log_2(\Delta:\Psi)$.
Let $\Theta:=\aut(\CB)\cap\Delta$.
Since $\Psi\leq\Theta\leq\Delta$, it follows that $\pid(\Delta,\Psi)\leq\pid(\Delta,\Theta)+\pid(\Theta,\Psi)$.
By the definition of $\Theta$, for each $\theta\in\Theta$ and each $B\in\CB$ it holds that $B^\theta\in\CB$
and therefore $B^\theta$ is equal to $B$ or disjoint from $B$ (if $\Theta$ would be transitive on $V$, then $\CB$ is a block system for $\Theta$).
Therefore, fixing a point $v\in B\in\CB$ also fixes the set $B\in\CB$, i.e., $\Theta_{(v)}[\CB]\leq\Theta[\CB]_{(B)}$ for all $v\in B\in\CB$.
Let $\CX\subseteq\CB$ and $X:=\bigcup_{B\in\CX}B\subseteq V$
and assume $\Theta[\CB]_{(X)}\leq\Psi[\CB]$.
Then, $\Theta_{(X)}[\CB]\leq\Theta[\CB]_{(X)}\leq\Psi[\CB]=1[\CB]$ and thus $\Theta_{(X)}\leq\Psi$.
This gives
$\pid(\Theta,\Psi)\leq\pid(\Theta[\CB],\Psi[\CB])$.
Moreover, $\pid(\Theta[\CB],\Psi[\CB])=\pidb(\Theta[\CB])\leq\log_2|\Theta[\CB]|=\log_2(\Theta:\Psi)$.
Next, we show $\pid(\Delta,\Theta)\leq 2\cdot\log_2(\Delta:\Theta)$.
Consider the permutation groups $\Theta[V^2]$ and $\Delta[V^2]$ induced on $V^2$.
Fixing two points $v,w$ in the domain of $\Delta$ also fixes the point $(v,w)$ in the domain of $\Delta[V^2]$, i.e.,
$\Delta_{(\{v,w\})}[V^2]\leq\Delta[V^2]_{((v,w))}$ for all $v,w\in V$.
Moreover, $\Delta_{(X)}[V^2]\leq\Theta[V^2]$ implies that $\Delta_{(X)}\leq\Theta$ for all $X\subseteq V$.
Therefore, $\pid(\Delta,\Theta)\leq 2\cdot\pid(\Delta[V^2],\Theta[V^2])$.
Let $A:=\{(v,w)\in V^2\mid \{v,w\}\subseteq B_i$ for some $B_i\in\CB\}$.
Then, $\Theta[V^2]=\stab_{\Delta[V^2]}(A)$.
Therefore, $2\cdot\pid(\Delta[V^2],\Theta[V^2])\leq 2\cdot\log_2(\Delta[V^2]:\Theta[V^2])=2\cdot\log_2(\Delta:\Theta)$.
\item\label{exa:pid3} Let $\Psi:=\alt(V)\leq\Delta:=\sym(V)$. This is an example where the relative base size is large compared to the index of the subgroup.
It is easy to see that $\pid(\Delta,\Psi)=|V|-1$.
\end{enumerate}
\end{exa}

The next lemma facilitate a subgroup reduction, similar as in Luks's framework.
The multiplicative cost of this recursion corresponds to the index of the subgroup.

\begin{lem}\label{lem:toGroup}
There is an algorithm $\toGroup$ that gets as input a pair $(\CX,\Psi^\ord)$
where $\CX=(J,A,\Delta^\ord,g^\ord)$ is a tuple
for which Property (A), (B) and (g) hold
and $\Psi^\ord\leq\Delta^\ord$ is a subgroup.
Let $c_\ind:=(\Delta^\ord:\Psi^\ord)$ and $c_\pid:=\pid(\Delta^\ord,\Psi^\ord)$.
In time polynomial in the input and output size, the algorithm either
\begin{enumerate}
  \item\label{lem:toGroupPart} finds a non-trivial partition family $\PP=\{\CP_k\}_{k\in K}$ of $J$
  with $|K|\leq c_\ind\cdot|V|^{c_\pid}$, or
  \item\label{lem:toGroupSub} reduces the canonical labeling problem of $\CX$ to the canonical labeling problem of $c_\ind$-many instances $(\hat J_i,A,\Psi^\ord,\bot)$ with $|\hat J_i|\leq |J|$
  for $i\in[c_\ind]$.
\end{enumerate}
\end{lem}

In contrast to Luks's subgroup reduction, the present reduction splits all labeling cosets in $J$
simultaneously.
We describe the idea of this algorithm.

\paragraph{Intuition of the Subgroup Recursion}
We consider the decomposition into left cosets of
$\Delta^\ord=\bigcup_{\ell\in[s]}\delta^\ord_\ell\Psi^\ord$
and define $\hat J:=\{\rho_i\delta^\ord_\ell\Psi^\ord\mid i\in[t],\ell\in[s]\}$.
Surprisingly, we can show that $\aut(\hat J)=\aut(J)$.
This means that a canonical labeling for $\hat J$ defines a canonical labeling for $J$
as well and vice versa.
Therefore, the first idea that comes to mind would be a recursion on the instance $(\hat J,A,\Psi^\ord,\bot)$.
However, there are two problems when recursing on $\hat J$.
First, the instance $\hat J$ does not necessarily satisfy Property (A).
To ensure, Property (A) for the recursive instance, one could reset $A:=V$, but this would not lead to the desired
recursion.
Second, it holds that $|\hat J|>|J|$ (assumed that $\Psi^\ord<\Delta^\ord$ is a proper subgroup).
Also this blow-up in the instance size would not lead to the desired recursion.
The given subroutine is designed to fix exactly these two problems.
In particular, we construct a decomposition of $\hat J=\hat J_1\cupdot\ldots\cupdot\hat J_r$
such that $r\leq c_\ind$ and $|\hat J_i|\leq|J|$ and such that Property (A) holds for each instance
$(\hat J_i,A,\Psi^\ord,\bot)$.

\begin{proof}[Proof of \cref{lem:toGroup}]
\algorithmDAN{\toGroup(J,A,\Delta^\ord,g^\ord,\Psi^\ord)}\\
Decompose $\Delta^\ord=\bigcup_{\ell\in[s]}\delta^\ord_\ell\Psi^\ord$ into left cosets of $\Psi^\ord$.\\
Define $\hat J:=\{\rho_i\delta^\ord_\ell\Psi^\ord\mid i\in[t],\ell\in[s]\}$.\\
\comment{We claim that $\aut(\hat J)=\aut(J)$. It is not difficult to see that $\aut(J)\subseteq\aut(\hat J)$ since
$\hat J$ is defined in an isomorphism-invariant way.
On the other side, let $\sigma\in\aut(\hat J)$. Therefore, for each labeling coset $\rho_i\delta_\ell^\ord\Psi^\ord\in\hat J$ there is a
labeling coset
$\rho_{i'}\delta_{\ell'}^\ord\Psi^\ord\in\hat J$ such that $(\rho_i\delta_\ell^\ord\Psi^\ord)^\sigma=\rho_{i'}\delta_{\ell'}^\ord\Psi^\ord$
or equivalently $\sigma\in\rho_i\delta_\ell^\ord\Psi^\ord{\delta_{\ell'}^\ord}^{-1}\rho_{i'}^{-1}$.
In particular, $\sigma\in\rho_i\Delta^\ord\rho_{i'}^{-1}$ or equivalently $(\rho_i\Delta^\ord)^\sigma=\rho_{i'}\Delta^\ord$.
Therefore, $\sigma\in\aut(J)$}\\
Let $X^\ord=(x_1^\ord,\ldots,x_{c_\pid}^\ord)\in (V^\ord)^{c_\pid}$ be the minimal (w.r.t.~the ordering ``$\prec$'') tuple
such that $\Delta^\ord_{(\{x_1^\ord,\ldots,x_{c_\pid}^\ord\})}\leq\Psi^\ord$.\\
We say that $X\in V^{c_\pid}$ \emph{identifies} the subcoset $\rho_i\delta^\ord_\ell\Psi^\ord\leq\Delta_i\rho_i$
if $X^{\rho_i\delta^\ord_\ell\psi^\ord}=X^\ord$ for some $\psi^\ord\in\Psi^\ord$.\\
\comment{We claim that each $X\in V^{c_\pid}$ identifies at most one subcoset $\rho_i\delta^\ord_\ell\Psi^\ord\leq\Delta_i\rho_i$
of each $\Delta_i\rho_i\in J$.
Assume that $X\in V^{c_\pid}$ identifies both
$\rho_i\delta^\ord_\ell\Psi^\ord,\rho_i\delta^\ord_{\ell'}\Psi^\ord\leq\Delta_i\rho_i$
for some $\Delta_i\rho_i\in J$.
We show that $\ell=\ell'$.
There are $\psi^\ord,{\psi^\ord}'\in\Psi^\ord$ such that $X^{\rho_i\delta^\ord_\ell\psi^\ord}=X^\ord=X^{\rho_i\delta^\ord_{\ell'}{\psi^\ord}'}$.
This implies $(\delta^\ord_\ell\psi^\ord)^{-1}\delta^\ord_{\ell'}{\psi^\ord}'\in\Delta^\ord_{(\{x_1^\ord,\ldots,x_{c_\pid}^\ord\})}\leq\Psi^\ord$
and therefore $(\delta^\ord_\ell)^{-1}\delta^\ord_{\ell'}\in\Psi^\ord$ and thus $\delta^\ord_\ell=\delta^\ord_{\ell'}$}\\
Define an (unordered) partition $\hat \CJ:=\{\hat J_1,\ldots,\hat J_r\}$ of 
$\hat J=\hat J_1\cupdot\ldots\cupdot\hat J_r$ such that:\\
$\rho_i\delta^\ord_\ell\Psi^\ord,\rho_{i'}\delta^\ord_{\ell'}\Psi^\ord\in\hat J_k$ for some $\hat J_k\in\hat\CJ$,
iff
$(\rho_i\delta^\ord_\ell\Psi^\ord)|_{V\setminus A}=(\rho_{i'}\delta^\ord_{\ell'}\Psi^\ord)|_{V\setminus A}$ and
there is a tuple $X\in V^{c_\pid}$ that identifies both
$\rho_i\delta^\ord_\ell\Psi^\ord$ and $\rho_{i'}\delta^\ord_{\ell'}\Psi^\ord$.\\
\comment{As already observed, each $X\in V^{c_\pid}$ identifies at most one subcoset 
$\rho_i\delta^\ord_\ell\Psi^\ord\leq\Delta_i\rho_i$ of $\Delta_i\rho_i$.
For this reason $|\hat J_k|\leq |J|$ for each $\hat J_k\in\hat \CJ$. On the other side, $|\hat\CJ|\leq c_\ind\cdot |V|^{c_\pid}$}\\
Define a cover $\CC:=\{C_1,\ldots,C_r\}$ of $J=C_1\cup\ldots\cup C_r$ such that:\\
$\Delta_i\rho_i\in C_k$
if there are $\ell\in[s]$ such that
$\rho_i\delta^\ord_\ell\Psi^\ord\in\hat J_k$.\\
Define $\PP:=\{\CP_k\}_{k\in[r]}$ as partition family induced by $\CC$, i.e.,
$\CP_k:=\{P_{k,1},P_{k,2}\}$ where
$P_{k,1}:=C_k$ and $P_{k,2}:=J\setminus C_k$ for $k\in[r]$.
\begin{cs}
\case{$\PP$ is non-trivial}
Return $\PP$.\\
\comment{In this case, Option \ref{lem:toGroupPart} of \cref{lem:toGroup} is satisfied}
\case{there is $\CP_k\in\PP$ that is the partition into singletons}~\\
Return $\Lambda:=\canObj(J)$ using \cref{cor:canObj0}.\\
\comment{Since $\CP_k\in\PP$ is the partition into singletons and has size $|\CP_k|\leq 2$, it follows that $|J|\leq 2$}

\end{cs}
\comment{Now, each $\CP_k\in\PP$ is the singleton partition. This means that for each $\hat J_k\in\hat\CJ$
and each $\Delta_i\rho_i\in J$ there is a subcoset $\rho_i\delta^\ord_\ell\Psi^\ord\leq\Delta_i\rho_i$
that is contained in $\hat J_k$.
The same argument that shows $\aut(\hat J)\leq\aut(J)$
also shows that $\aut(\hat J_k)\leq\aut(J)$ for each $\hat J_k\in\hat\CJ$.
Roughly speaking, this means that $\hat J_k$ can be seen as an individualization of $\hat J$}\\
Compute $\Theta_k\tau_k:=\canSet(\hat J_k,A,\Psi^\ord,\bot)$ for each
$\hat J_k\in\hat\CJ$ recursively.\\
\comment{In this case, we satisfy Option \ref{lem:toGroupSub} of \cref{lem:toGroup}}\\
Define $\hat\CJ^\set:=\{\Theta_k\tau_k\mid\hat J_k\in \hat\CJ\}$.\\
\comment{We collect the canonical labelings $\Theta_k\tau_k$ of $\hat J_k$ leading to minimal canonical forms of the input}\\
Define $\hat\CJ_{\min}^\set:=\arg\min_{\Theta_k\tau_k\in \hat \CJ^\set}J^{\tau_k}\subseteq\hat \CJ^\set$ where the minimum is taken w.r.t.~the ordering ``$\prec$'' from \cref{lem:prec}.\\
Return $\Lambda:=\langle\hat\CJ_{\min}^\set\rangle$.\\
\comment{This is the smallest coset containing all labeling cosets in $\hat\CJ_{\min}$ as defined
in the preliminaries. The correctness proof for (CL2) is similar to the (CL2)-proof of \cref{lem:canSetSet}}
\end{proof}

\begin{theo}[\cite{DBLP:journals/corr/Babai15}, Theorem 3.2.1.]\label{theo:cameron}
Let $\Delta\leq\sym(V)$ be a primitive group of order $|\Delta| \geq |V|^{1+\log_2 |V|}$ where $|V|$ is greater
than some absolute constant. Then $\Delta$ is a Cameron group and has a normal subgroup $N$ of index at most $|V|$ such that $N$
has a system of imprimitivity on which $N$ acts as a Johnson group.
Moreover, $N$ and the system of imprimitivity in question can be found in polynomial time.
\end{theo}

\begin{lem}\label{lem:pid4}
Let $N\leq\Delta\leq\sym(V)$ be the group from \cref{theo:cameron}.
Then, $\pid(\Delta,N)\leq\log_2|V|$.
\end{lem}

\begin{proof}
As $\Delta\leq\sym(V)$ is a Cameron group, we have
$(\alt(W)[\binom{W}{s}])^k\leq\Delta\leq\sym(W)[\binom{W}{s}] \wr \sym(k)$.
We identify $V=\binom{W}{s}^k$.
We have an induced homomorphism $h:\Delta\to\sym(k)$.
It follows from the proof of \cref{theo:cameron} that $N=\ker(h)$.
For each $i\in[k]$, we choose two points $A_i=(a_1,\ldots,a_k),B_i=(b_1,\ldots,b_k)\in\binom{W}{s}^k$ such that $a_i=b_i$
and $a_i\neq b_j$ for $i\neq j$. We define $X:=\bigcup_{i\in[k]}\{A_i,B_i\}$.
Observe that $|X|=2k\leq 2\frac{\log_2 |V|}{\log_2|W|}\leq \log_2|V|$.
We claim that $\Delta_{(X)}\leq N$.
Observe that $h(\Delta_{(\{A_i,B_i\})})\leq\sym(k)_{(i)}$ for all $i\in[k]$.
Therefore, $h(\Delta_{(X)})\leq\sym(k)_{(\{1,\ldots,k\})}=1$ and thus $\Delta_{(X)}\leq\ker(h)=N$.
\end{proof}

\begin{lem}\label{lem:toJohnson}
There is an algorithm $\toJohnson$ that gets as input an instance $(J,A,\Delta^\ord,\bot)$
for which Property (A), (B) and (g) hold.
In time $(|V|+|J|)^{\polylog|V|}$, the algorithm reduces the canonical labeling problem of $(J,A,\Delta^\ord,\bot)$ to
canonical labeling of either
\begin{itemize}
  \item (progress \cref{prog:linJ}) two instances $(J_1,A,\Delta^\ord,\bot)$ and $(J_2,A,\Delta^\ord,\bot)$ with $|J_1|+|J_2|=|J|$, or
  \item (progress \cref{prog:linA}) one instance $(J,A',\Delta^\ord,\bot)$ with $|A'|<|A|$, or
  \item (progress \cref{prog:J}) $2^{\log_2 p+\log_2(|V|)^4}$-many instances $(J_{k,i},A,\Delta^\ord,\bot)$
  of size $|J_{k,i}|\leq \frac{1}{p}|J|$ and to additionally $2^{\log_2(|V|)^4}$-many instances $(J_k,V,\Delta_k^\ord,\bot)$ of size $|J_k|\leq p$ for some $p\in\NN$
  with $1<p\leq\frac{1}{2}|J|$, or
  \item (progress \cref{prog:A}) $2^{\log_2(|V|)^3}$-many instances $(\hat J_i,A,\Psi^\ord,\bot)$ with $|\hat J_i|\leq |J|$ and
  such that
  $\orb_{A^\ord}(\Psi^\ord)\leq\frac{1}{2}\orb_{A^\ord}(\Delta^\ord)$, or
  \item (progress \cref{prog:giant}) $|V|$-many instances $(\hat J_i,A,\Psi^\ord,g^\ord)$ where
  $|\hat J_i|\leq |J|$ and such that $\orb_{A^\ord}(\Psi^\ord)\leq \orb_{A^\ord}(\Delta^\ord)$ and $g^\ord$ is defined.
\end{itemize}
\end{lem}

\paragraph{Intuition of the Johnson Reduction}
First of all, we want to reduce to the case in which all $\Delta_i\leq\sym(V)$ are transitive on $A\subseteq V$.
To achieve transitivity, Babai's algorithm uses Luks's idea of orbit-by-orbit processing.
However, the orbit-by-orbit recursion is a tool that is developed for strings and
needs a non-trivial adaption when dealing with a set of labeling cosets $J$.
To achieve transitivity, the present algorithm uses an adaption of the orbit-by-orbit recursion that was developed
in \cite{DBLP:conf/stoc/SchweitzerW19}.
In the transitive case, we proceed similarly to Babai's algorithm.
First, we define a block system $\CB^\ord$ on which $\Delta^\ord$ acts primitively.
If the primitive group acting on $\CB^\ord$ is small, we use the subgroup reduction
from \cref{lem:toGroup} to reduce to a subgroup $\Psi^\ord\leq\Delta^\ord$ that is defined as the kernel of that action.
In case that the primitive group is large, we use Cameron’s classification of large primitive groups
which implies that the primitive group is a Cameron group.
Using \cref{theo:cameron}, we reduce the Cameron group to a Johnson group
by using the subgroup reduction from \cref{lem:toGroup}.
The Johnson group (acting on subsets of a set $W^\ord$)
in turn can be used to define a giant representation $g^\ord:\Delta^\ord\to\sym(W^\ord)$.

\begin{proof}[Proof of \cref{lem:toJohnson}]
\algorithmDAN{\toJohnson(J,A,\Delta^\ord,\bot)}
\begin{cs}
\case{$|A|$ is smaller than some absolute constant}~\\
Return $\canObj(J)$ using \cref{cor:canObj0}.\\
\comment{We claim that Property (A) and (B) imply that $|J|$ is smaller than some absolute constant.
By Property (B), it holds that $\Delta_i\rho_i=\rho_i\Delta^\ord$ for all $\Delta_i\rho_i\in J$.
Let $\Lambda:=\{\lambda\in\lab(V)\mid\lambda|_{V\setminus A}=\rho_1|_{V\setminus A}\}$.
By definition, $|\Lambda|\leq |A|!$.
By Property (A), for all $\rho_i\Delta^\ord$ there is a representative $\rho_i^*\in\rho_i\Delta^\ord$ with $\rho_i^*\in\Lambda$.
The representatives $\rho_i^*,\rho_j^*$ for $i\neq j$ are pairwise distinct since otherwise $\rho_i\Delta^\ord=\rho_i^*\Delta^\ord=\rho_j^*\Delta^\ord=\rho_j\Delta^\ord$.
Therefore, $|J|=|\{\rho_1^*,\ldots,\rho_t^*\}|\leq|\Lambda|\leq |A|!$ which proves the claim.
Therefore, the algorithm from \cref{cor:canObj0} runs in constant time}

\case{$\Delta_i$ is intransitive on $A$ for some (and because of (B) for
all) $\Delta_i\rho_i\in J$}~\\
Define ${A^\ord}^*\subsetneq A^\ord$ as the $\Delta^\ord$-orbit on $A^\ord$ that is minimal w.r.t.~the ordering ``$\prec$'' from \cref{lem:prec}.\\
For each $\Delta_i\rho_i\in J$, define $A_i^*\subsetneq A$ as the
$\Delta_i$-orbit such that
$(A_i^*)^{\rho_i}={A^\ord}^*$.\\
Define an (unordered) partition $\CP:=\{P_1,\ldots,P_p\}$ of $J= P_1\cupdot\ldots\cupdot P_p$ such
that:\\
$\Delta_i\rho_i,\Delta_j\rho_j\in P_\ell$ for some $P_\ell\in\CP$,
if and only if $A_i^*=A_j^*$.
\begin{cs}
\case{$\CP$ is non-trivial}~\\
\comment{The singleton $\{\CP\}$ can be seen as a non-trivial partition family consisting of one single partition}\\
Return $\Lambda:=\recPart(J,A,\Delta^\ord,\bot,\{\CP\})$ using \cref{lem:recPart}.\\
\comment{Since $|\{\CP\}|=1$, we make progress \cref{prog:J} or \cref{prog:linJ}}

\case{$\CP$ is the partition into singletons, i.e., $A_i^*\neq A_j^*$ for all $\Delta_i\rho_i\neq \Delta_j\rho_j\in J$}~\\
\comment{In this case, we can define a coset-labeled hypergraph $(H,J,\alpha)$}\\
Define the hypergraph $H:=\{A_1^*,\ldots,A_t^*\}$.\\
Define $\alpha:H\to J$ by setting $\alpha(A_i^*):=\Delta_i\rho_i$ for each $A_i^*\in H$.\\
Return $\Lambda:=\canSetHyper(J,H,\alpha)$ using \cref{lem:canSetHyper}.\\
\comment{The algorithm from \cref{lem:canSetHyper} runs in time $(|V|+|J|)^{\polylog|V|}$}

\case{$\CP$ is the singleton partition, i.e., $A_i^*= A_j^*$ for all $\Delta_i\rho_i,\Delta_j\rho_j\in J$}~\\
Define $A^*:=A_i^*$ for some $\Delta_i\rho_i\in J$.\\
\comment{The set $A^*$ is well-defined and does not depend on the choice of $\Delta_i\rho_i\in J$}\\
Define $\Lambda_i:=(\Delta_i\rho_i)|_{V\setminus A^*}$ for each $\Delta_i\rho_i\in J$.\\
Define an (unordered) partition $\CQ:=\{Q_1,\ldots,Q_q\}$ of $J= Q_1\cupdot\ldots\cupdot Q_q$ such
that:\\
$\Delta_i\rho_i,\Delta_j\rho_j\in Q_\ell$ for some $Q_\ell\in\CQ$,
if and only if
$\Lambda_i=\Lambda_j$.

\begin{cs}
\case{$\CQ$ is non-trivial}~\\
\comment{The singleton $\{\CQ\}$ can be seen as a non-trivial partition family consisting of one single partition}\\
Return $\Lambda:=\recPart(J,A,\Delta^\ord,\bot,\{\CQ\})$ using \cref{lem:recPart}.\\
\comment{Since $|\{\CQ\}|=1$, we make progress \cref{prog:J} or \cref{prog:linJ}}

\case{$\CQ$ is the singleton partition, i.e., $\Lambda_i=\Lambda_j$ for all $\Delta_i\rho_i,\Delta_j\rho_j\in J$}~\\
Recurse and return $\Lambda:=\canSet(J,A^*,\Delta^\ord,\bot)$.\\
\comment{By definition of the partition, Property (A) also holds with $A^*\subsetneq A$ in place of $A$. We have progress \cref{prog:linA}}

\case{$\CQ$ is the partition into singletons, i.e., $\Lambda_i\neq\Lambda_j$ for all $\Delta_i\rho_i\neq \Delta_j\rho_j\in J$}~\\
Define $\Delta_i^\circ:=\Delta_i[V\setminus A^*]\times\sym(A^*)\geq\Delta_i$
and define $J^\circ:=\{\Delta_1^\circ\rho_1,\ldots,\Delta_t^\circ\rho_t\}$.\\
\comment{Since $\CQ$ is the partition into singletons, $|J^\circ|=|J|$}\\
Define ${\Delta^\circ}^\ord:=\rho_i^{-1}\Delta_i^\circ\rho_i$ for some $i\in[t]$.\\
Recursively compute $\Delta\rho:=\canSet(J^\circ,A\setminus A^*,{\Delta^\circ}^\ord,\bot)$.\\
\comment{We claim that Property (A) holds for this instance with $A\setminus A^*$ in place of $A$.
Observe that $V\setminus (A\setminus A^*)=(V\setminus A)\cupdot A^*$.
Since $\Delta_i^\circ$ is a direct product, we can consider both direct factors separately and obtain
$(\Delta_i\rho_i)|_{V\setminus A}=(\Delta_j\rho_j)|_{V\setminus A}$
and $(\Delta_i\rho_i)|_{A^*}=\sym(A^*)\rho_i|_{A^*}=\sym(A^*)\rho_j|_{A^*}=(\Delta_j\rho_j)|_{A^*}$
for all $\Delta_i\rho_i\neq \Delta_j\rho_j\in J$.
Since $A\setminus A^*\subsetneq A$, we have progress \cref{prog:linA}}\\
Define $\alpha:J^\circ\to J$ by setting $\alpha(\Delta_i^\circ\rho_i):=\Delta_i\rho_i$ for each $\Delta_i^\circ\rho_i\in J^\circ$.\\
Return $\Lambda:=\canSetSet(J^\circ,J,\alpha,\Delta\rho)$ using \cref{lem:canSetSet}.\\
\comment{The algorithm from \cref{lem:canSetSet} runs in time $(|V|+|J|)^{\polylog|V|}$}
\end{cs}
\end{cs}

\case{$\Delta_i$ is transitive on $A$ for some
(and because of (B) for all)
$\Delta_i\rho_i\in J$}~\\
\comment{We reduce the group to the primitive case}\\
Compute a minimal block system for $\CB^\ord=\{B_1^\ord,\ldots,B_b^\ord\}$ for $\Delta^\ord$ acting on $A^\ord$.\\
\comment{By using a canonical generating set from \cref{lem:canGen} for $\Delta^\ord$, we can ensure that the block system $B^\ord$
only depends on $\Delta^\ord$ (and not on the representation of $\Delta^\ord$).
Observe that $\Delta^\ord[\CB^\ord]\leq\sym(\CB^\ord)$ is a primitive group}
\begin{cs}
\case{$\Delta^\ord[\CB^\ord]$ is smaller than or equal to $|V|^{3+\log_2 |V|}$}~\\
Define $\Psi^\ord:=\stab_{\Delta^\ord}(B_1^\ord,\ldots,B_b^\ord)$.\\
\comment{The group can be computed using a membership test as
stated in the preliminaries}\\
Apply $\toGroup(J,A,\Delta^\ord,\bot,\Psi^\ord)$ using \cref{lem:toGroup}.
\begin{cs}
\case{$\toGroup$ returns a non-trivial partition family $\PP$}~\\
Return $\Lambda:=\recPart(J,A,\Delta^\ord,\bot,\PP)$.\\
\comment{It holds that $|\PP|\leq c_\ind\cdot |V|^{c_\pid}$ where $c_\ind:=(\Delta^\ord:\Psi^\ord)\leq |V|^{3+\log_2|V|}$.
As in \cref{exa:pid}.\ref{exa:pid2}, we have $c_\pid:=\pid(\Delta^\ord,\Psi^\ord)\leq 2\cdot\log_2(c_\ind)$. This leads to progress
\cref{prog:J} or \cref{prog:linJ}}
\case{$\toGroup$ reduces to $c_\ind$-many instances $(\hat J_i,A,\Psi^\ord,\bot)$}~\\
Recurse on these $c_\ind$-many instances $(\hat J_1,A,\Psi^\ord,\bot),\ldots,(\hat J_{c_\ind},A,\Psi^\ord,\bot)$
as suggested by the subroutine.\\
\comment{We analyze the recurrence.
The multiplicative cost is $c_\ind=(\Delta^\ord:\Psi^\ord)\leq |V|^{3+\log_2|V|}$.
Moreover, $\orb_{A^\ord}(\Psi^\ord)=|A^\ord|/|\CB^\ord|\leq\frac{1}{2}\orb_{A^\ord}(\Delta^\ord)$. This leads to
progress \cref{prog:A}}
\end{cs}
\end{cs}
\begin{cs}
\case{$\Delta^\ord[\CB^\ord]$ is greater than $|V|^{3+\log_2|V|}$}~\\
\comment{Since $|\CB^\ord|!\geq |\Delta^\ord[\CB^\ord]|>|V|^{3+\log_2|V|}\geq |A|^{3+\log_2|A|}$
and $|A|$ is greater than some absolute constant we can apply \cref{theo:cameron}.
It follows that $\Delta^\ord[\CB^\ord]$ is a Cameron group.
Next, we will reduce the group to the Johnson case}\\
Define $N^\ord\trianglelefteq\Delta^\ord[\CB^\ord]\leq\sym(\CB^\ord)$ as the subgroup of index at most $b\leq|V|$ which has a system
of imprimitivity on which $N^\ord$ acts as a Johnson group as in \cref{theo:cameron}.\\
Define $\Psi^\ord\trianglelefteq\Delta^\ord\leq\sym(V^\ord)$ as the corresponding normal subgroup for which $\Psi^\ord[\CB^\ord]=N^\ord$ holds.\\
\comment{Also $\Psi^\ord\trianglelefteq\Delta^\ord$ is of index at most $b\leq|V|$ and has a system
of imprimitivity on which it acts as a Johnson group.
By using a canonical generating set from \cref{lem:canGen} for $\Delta^\ord[\CB^\ord]$, we can ensure that $N^\ord$
and $\Psi^\ord$ only depend on $\Delta^\ord[\CB^\ord]$ (and not on the representation of $\Delta^\ord[\CB^\ord]$)}\\
Apply $\toGroup(J,A,\Delta^\ord,\bot,\Psi^\ord)$ using \cref{lem:toGroup}.
\begin{cs}
\case{$\toGroup$ returns a non-trivial partition family $\PP$}~\\
Return $\Lambda:=\recPart(J,A,\Delta^\ord,\bot,\PP)$.\\
\comment{It holds that $|\PP|\leq c_\ind\cdot |V|^{c_\pid}$ where $c_\ind:=(\Delta^\ord:\Psi^\ord)\leq b\leq|V|$.
By \cref{lem:pid4}, it follows that $c_\pid:=\pid(\Delta^\ord,\Psi^\ord)\leq \log_2|V|$. This leads to progress
\cref{prog:J} or \cref{prog:linJ}}
\case{$\toGroup$ reduces to $b$-many instances $(\hat J_i,A,\Psi^\ord,\bot)$}~\\
Now, there is a system
of imprimitivity on which $\Psi^\ord$ acts as a Johnson group and therefore
there is a homomorphism $h^\ord:\Psi^\ord\to\sym(W^\ord)[\binom{W^\ord}{s}]$.\\
Define $g^\ord:\Psi^\ord\to\sym(W^\ord)$ as the giant representation obtained from $h^\ord$ whose image is acting on $W^\ord$
(rather than acting on subsets of $W^\ord$).\\
\comment{Since $|\Delta^\ord[\CB^\ord]|>|V|^{3+\log_2|V|}$, it follows from the proof of \cref{theo:cameron} that $|W^\ord|> 2+\log_2|V|$}\\
Recurse on the $b$-many instances $(\hat J_1,A,\Psi^\ord,g^\ord),\ldots,(\hat J_b,A,\Psi^\ord,g^\ord)$.\\
\comment{Observe that the algorithm recurses on
the instances
$(\hat J_i,A,\Psi^\ord,g^\ord)$
rather than
$(\hat J_i,A,\Psi^\ord,\bot)$.
We analyze the recurrence. We have a multiplicative cost of at most $b\leq |V|$ and recursive instances
where $g^\ord$ is defined. This leads to progress \cref{prog:giant}}
\end{cs}
\end{cs}
\end{cs}
\end{proof}

\begin{defn}[\cite{DBLP:conf/stoc/Babai16}]
Let $\Delta\leq\sym(V)$ and let $g:\Delta\to\sym(W)$ be a giant representation.
We say that $v\in V$ is \emph{affected} by $g$ if $g$ does not map $\Delta_{(v)}$, the pointwise stabilizer of $v$ in $\Delta$,
onto a giant, i.e., it does not hold $\alt(W)\leq g(\Delta_{(v)})\leq\sym(W)$.
A set $S\subseteq V$ consisting of affected points is called \emph{affected} set.
\end{defn}

\begin{theo}[\cite{DBLP:conf/stoc/Babai16}, Theorem 6]\label{theo:unaffected}
Let $\Delta\leq\sym(V)$ be a permutation group and let $k$ denote the length of the largest $\Delta$-orbit of $V$.
Let $g:\Delta\to\sym(W)$ be a giant representation. Let $U\subseteq V$ denote the set of all elements of
$V$ that are not affected by $g$. Then the following holds.

\begin{enumerate}
  \item\label{theo:unaffected1} (Unaffected Stabilizer Theorem)  Assume $|W|>\max\{8,2+\log_2 k\}$. Then $g$ maps $G_{(U)}$, the pointwise stabilizer of $U$ in $G$,
  onto $\alt(W)$ or $\sym(W)$ (so $g:G_{(U)}\to\sym(W)$ is still a giant representation). In particular, $U\subsetneq V$ (at least one element is affected). 
  \item\label{theo:unaffected2} (Affected Orbits Lemma) Assume $|W|\geq 5$. If $S$ is an affected $\Delta$-orbit, i.e., $S\cap U =\emptyset$,
then $\ker(g)$ is not transitive on $S$; in fact, each orbit of $\ker(g)$ in $S$ has length at most $|S|/|W|$.
\end{enumerate}
\end{theo}

\begin{defn}[Certificates of Fullness]
A group $G\leq\sym(V)$ is called \emph{certificate of fullness} for an instance $(J,A,\Delta^\ord,g^\ord)$
if
\begin{enumerate}
  \item $G\leq\aut(J)$,
  \item $G^\ord:=G^{\rho_i}\leq\Delta^\ord$ does not depend on the choice of $\Delta_i\rho_i\in J$, and
  \item $g^\ord:G^\ord\to\sym(W^\ord)$ is still a giant representation.
\end{enumerate}
\end{defn}

\begin{lem}\label{lem:processJohnson}
There is an algorithm $\processJohnson$ that gets a input an instance $(J,A,\Delta^\ord,g^\ord)$
for which Property (A), (B) and (g) hold where $g^\ord$ is defined.
In time $(|V|+|J|)^{\polylog|V|}$, the algorithm reduces the canonical labeling problem of $(J,A,\Delta^\ord,g^\ord)$ to
canonical labeling of either
\begin{itemize}
  \item (progress \cref{prog:linJ}) two instances $(J_1,A,\Delta^\ord,g^\ord)$ and $(J_2,A,\Delta^\ord,g^\ord)$ with $|J_1|+|J_2|=|J|$, or
  \item (progress \cref{prog:J}) $2^{\log_2 p+\log_2(|V|)^4}$-many instances $(J_{k,i},A,\Delta^\ord,g^\ord)$
  of size $|J_{k,i}|\leq \frac{1}{p}|J|$ and to additionally $2^{\log_2(|V|)^4}$-many instances $(J_k,V,\Delta_k^\ord,\bot)$ of size $|J_k|\leq p$ for some $p\in\NN$
  with $1<p\leq\frac{1}{2}|J|$, or
  \item (progress \cref{prog:A}) $2^{\log_2(|V|)^3}$-many instances $(\hat J_i,A,\Psi^\ord,\bot)$ with $|\hat J_i|\leq |J|$ and such that
  $\orb_{A^\ord}(\Psi^\ord)\leq\frac{1}{2}\orb_{A^\ord}(\Delta^\ord)$, or
  \item (Fullness certificate) finds a certificate of fullness $G\leq\sym(V)$ for the input instance.
\end{itemize}
\end{lem}

\paragraph{Intuition of the Certificate Producing Algorithm}
We describe the idea of the algorithm.
The algorithm picks a subset $T^\ord\subseteq W^\ord$ of logarithmic size.
We call this set $T^\ord$ a \emph{canonical test set}.
Next, we define the group $\Delta_T^\ord\leq\Delta^\ord$ which stabilizes $T^\ord$
in the image under $g^\ord$.
By doing so, we can define a giant representation $g_T^\ord:\Delta_T^\ord\to\sym(T^\ord)$.
Let $S^\ord,U^\ord\subseteq V^\ord$ be set of elements affected and unaffected by $g_T^\ord$, respectively.
We have a technical difference in our algorithm in contrast to Babai's method.
In Babai's method of local certificates, he processes a giant representation $g:\Delta\to\sym(W)$
and considers multiple test sets $T\subseteq W$ (one test set for each subset of logarithmic size).
In our framework, we define the giant representation for a group $\Delta^\ord$ over a linearly ordered set $V^\ord$.
This allows us to choose one single (canonical) test set $T^\ord\subseteq W^\ord$ only.
Here, canonical means that the subset is chosen minimal with respect to the ordering ``$\prec$''.
However, when we translate the ordered structures $V^\ord$ to unordered structures over $V$, we implicitly consider multiple test sets
and giant representations. More precise, by applying inverses of labelings in $\Delta_i\rho_i\in J$
to the ordered group $\Delta_T^\ord\leq\sym(V^\ord)$, we obtain a set of groups over $V$, i.e.,
$\{\lambda_i\Delta_T^\ord\lambda_i^{-1}\mid\lambda_i\in\Delta_i\rho_i\}$.
Similarly, we can define a set of giant representations
$\{(g_T^\ord)^{\lambda_i^{-1}}\mid\lambda_i\in\Delta_i\rho_i\}$ (where $(g_T^\ord)^{\lambda_i^{-1}}(\delta_i):=g_T^\ord(\lambda_i^{-1}\delta_i\lambda_i)$ for $\delta_i\in\Delta_i$) and a set of affected points
$H_i:=\{S\subseteq V\mid S^{\lambda_i}=S^\ord$ for some $\lambda_i\in\Delta_i\rho_i\}$.
Therefore, when dealing over unordered structures, we need to consider multiple groups and homomorphisms.
It becomes even more complex, since we are dealing with a set $J$ consisting of labeling cosets 
rather than one single group only.
In fact, we obtain a set of affected point sets $H_i$ for each labeling coset $\Delta_i\rho_i\in J$.
However, it turns out that the hardest case occurs when $H_i=H_j$ for all $\Delta_i\rho_i,\Delta_j\rho_j\in J$.
Roughly speaking, we will apply the following strategy.

We restrict each labeling coset in $J$ to some set of affected points $S\in H_i$
and define a set of local restrictions $J_S^*$ that ignore the vertices outside $S$.
The precise definition of $J_S^*$ is given in the algorithm.
Intuitively, the algorithms tries to analyze the labeling cosets locally.

Case 1: The local restrictions $J_S^*$ are pairwise distinct.
In this case, we canonize the
local restrictions $J_S^*$ recursively.
Observe that a canonical labeling $\Delta\rho$ for $J_S^*$ does not necessarily define a canonical labeling for $J$.
However, we can define a function $\alpha:J_S^*\to J$ that assigns each local restriction its corresponding labeling coset $\Delta_i\rho_i\in J$.
This function is well-defined since we assumed the local restrictions to be pairwise distinct.
Now, we can use the algorithm from \cref{lem:canSetSet} to canonize the instance $(J_S^*,J,\alpha,\Delta\rho)$.

Case 2: Some local restrictions in $J_S^*$ are pairwise different and some local restrictions in $J_S^*$
are pairwise equal.
In this case, we can define a non-trivial partition of $J$ in the following way.
We say that two labeling cosets $\Delta_i\rho_i,\Delta_j\rho_j$ are in the same part,
if and only if the corresponding local restrictions in $J_S^*$ coincide.
Actually, this leads to a family of partitions since we obtain one partition for each choice of an affected set $S\in H_i$.
We exploit this partition family by recursing using the subroutine $\recPart$ from \cref{lem:recPart}.

Case 3: The local restrictions $J_S^*$ are pairwise equal.
In this case, it is possible to find automorphisms $G_S\leq\sym(V)$ of $J$
which fix the unaffected points $V\setminus S$.
In fact, we can find such automorphisms for all choices of $S\in H_i$,
otherwise we are in a situation of a previous case.
Finally, we consider the group of automorphisms $G\leq\aut(J)$ generated by all $G_S$ for $S\in H_i$.
We can show that $G$ is indeed a certificate of fullness.

\begin{proof}[Proof of \cref{lem:processJohnson}]
\algorithmDAN{\processJohnson(J,A,\Delta^\ord,g^\ord)}~\\
Let $g^\ord:\Delta^\ord\to\sym(W^\ord)$ be the giant representation.\\
\comment{
By Property (g), the set $A^\ord$ is an orbit, $|W^\ord|>2+\log_2|V|\geq 2+\log_2|A^\ord|$,
$|W^\ord|$ is greater than some absolute constant
and $\Delta^\ord_{(A^\ord)}\leq\ker(g^\ord)$.
By the Unaffected Stabilizer Theorem \ref{theo:unaffected},
and since $\Delta^\ord_{(A^\ord)}\leq\ker(g^\ord)$, at least one element in $A^\ord$ is affected by $g^\ord$}\\
Define $\Pi^\ord$ as the kernel of $g^\ord$.\\
\comment{By the Affected Orbits Lemma \ref{theo:unaffected}, the orbits of $\Pi^\ord$ on $A^\ord$ have size at most $|A^\ord|/|W^\ord|$}\\
Define $T^\ord:=\{1,\ldots,3+\lfloor\log_2 |V|\rfloor\}\subseteq W^\ord$.\\
\comment{The set $T^\ord$ was referred to as canonical test set in the above paragraph}\\
Define $\Delta_T^\ord:=\{\delta^\ord\in\Delta^\ord\mid g^\ord(\delta^\ord)\in\stab_{\sym(W^\ord)}(T^\ord)\}$.\\
Define $g_T^\ord:\Delta_T^\ord\to \sym(T^\ord)$ as the giant representation that is obtained by restricting the image of $g^\ord$.\\
\comment{By the Unaffected Stabilizer Theorem \ref{theo:unaffected}, at least one element in $V^\ord$ is affected by $g_T^\ord$.
Moreover, since we assume that $\Delta^\ord_{(A^\ord)}\leq\ker(g^\ord)$, it follows that at least one element in $A^\ord$ is affected by $g_T^\ord$}\\
Decompose $V^\ord:=S^\ord\cupdot U^\ord$ where $S^\ord$ contains the points affected by $g_T^\ord$
and where $U^\ord:=V^\ord\setminus S$ contains the unaffected points.\\
\comment{By the Unaffected Stabilizer Theorem \ref{theo:unaffected}, $g_T^\ord:\Delta_{T,(U^\ord)}^\ord\to \sym(T^\ord)$ is still a giant representation}\\
\comment{We have the subgroup chain $\Pi^\ord,\Delta_{T,(U^\ord)}^\ord\leq\Delta_T^\ord\leq\Delta^\ord\leq\sym(V^\ord)$.
However, $\Pi^\ord$ and $\Delta_{T,(U^\ord)}^\ord$ might be incomparable under the subgroup relation}\\
Define $\Psi^\ord:=\stab_{\Delta^\ord}(S^\ord)$.\\
\comment{Observe that $\Delta_T^\ord\leq\Psi^\ord\leq\Delta^\ord$}\\
Decompose $\Delta^\ord=\bigcup_{\ell\in[s]}\delta^\ord_\ell\Psi^\ord$ into left cosets of $\Psi^\ord$.\\
\comment{This can be done in time polynomial in $|V|$ and $(\Delta^\ord:\Psi^\ord)\leq|V|^{3+\log_2|V|}$}\\
Define the hypergraph $H_i:=\{S\subseteq V\mid S^{\rho_i\delta^\ord_\ell}=S^\ord\text{ for some }\ell\in[s]\}$ for each $\Delta_i\rho_i\in J$.\\
\comment{The hypergraph $H_i$ can be seen as the preimages of affected points for each $\Delta_i\rho_i\in J$.
By definition of $\Psi^\ord$, the hypergraph $H_i$ does not depend on the choice of the representative $\rho_i$ of $\Delta_i\rho_i$.
However, $H_i$ might depend on the choice of the labeling coset $\Delta_i\rho_i\in J$.
We want to reduce to the case in which $H_i=H_j$ for all $\Delta_i\rho_i,\Delta_j\rho_j\in J$}\\
Define an (unordered) partition $\CP:=\{P_1,\ldots,P_p\}$ of $J=P_1\cupdot\ldots\cupdot P_p$ such
that:\\
$\Delta_i\rho_i,\Delta_j\rho_j\in P_\ell$ for some $P_\ell\in\CP$,
if and only if
$H_i=H_j$.
\begin{cs}
\case{$\CP$ is non-trivial}~\\
\comment{The singleton $\{\CP\}$ can be seen as a non-trivial partition family consisting of one single partition}\\
Compute and return $\Lambda:=\recPart(J,A,\Delta^\ord,g^\ord,\{\CP\})$ using \cref{lem:recPart}.\\
\comment{Since $|\{\CP\}|=1$, we make progress \cref{prog:J} or \cref{prog:linJ}}

\case{$\CP$ is the partition into singletons, i.e., $H_i\neq H_j$ for all $\Delta_i\rho_i\neq\Delta_j\rho_j\in J$}~\\
\comment{It holds that $|H_i|\leq (\Delta^\ord:\Psi^\ord)\leq |V|^{3+\log_2|V|}$.
We want to use the hypergraphs $H_i$ to define a partition family of $J$}\\
Define $K:=\{(k_1,k_2)\mid k_1,k_2\subseteq V, |k_1|,|k_2|\leq c\}$ as the set of pairs of subsets of $V$
of size at most $c:=\log_2(|V|^{3+\log_2|V|})$.\\
\comment{Observe that $|K|\leq 2^{\log_2(|V|)^4}$ since $|V|\geq|A|$ is greater than some absolute constant}\\
We say that $(k_1,k_2)\in K$ is \emph{compatible} with a set $S\subseteq V$ if $k_1\subseteq S$ and $k_2\subseteq V\setminus S$.\\
We say that $(k_1,k_2)\in K$ \emph{identifies} the hyperedge $S\in H_i$
in the hypergraph $H_i$ if $(k_1,k_2)$ is compatible with $S$ and $(k_1,k_2)$ is not compatible with each $S'\in H_i$ with $S'\neq S$.\\
\comment{We claim that for each hypergraph $H_i$ there is a $k\in K$ that identifies a hyperedge in $H_i$.
Let $H_i$ be a hypergraph with $\log_2(|H_i|)\leq c$.
We prove the claim by induction on $|H_i|$.
If $|H_i|=1$, then $(\emptyset,\emptyset)\in K$ identifies the hyperedge in $H_i$.
Assume that $|H_i|\geq 2$. Let $v\in V$ such that the partition $\{H_{i,v},H_{i,\bar v}\}$ of $H$
is non-trivial where
$H_{i,v}:=\{S\in H_i\mid v\in S\}$ and $H_{i,\bar v}:=\{S\in H_i\mid v\notin S\}$.
Assume that $1\leq|H_{i,v}|\leq\frac{1}{2}|H_i|$ and therefore $\log_2(|H_{i,v}|)\leq c-1$.
By induction, there is a $k=(k_1,k_2)$ with $|k_1|,|k_2|\leq c-1$ that identifies a hyperedge $S\in H_{i,v}$ in $H_{i,v}$.
Therefore, $(k_1\cup\{v\},k_2)\in K$ identifies the hyperedge $S\in H_i$ in $H_i$.
The other case in which $1\leq|H_{i,\bar v}|\leq\frac{1}{2}|H_i|$ is analogous}\\
\comment{We reduce to the case in which there is a $k\in K$ that identifies a hyperedge in each hypergraph $H_i$}\\
Define a cover $\CC:=\{C_k\mid k\in K\}$ of $J=\bigcup C_k$ such that:\\
$\Delta_i\rho_i\in C_k$
if $k\in K$ identifies a hyperedge $S\in H_i$ in the hypergraph $H_i$.\\
Define $\PP:=\{\CP_k\}_{k\in K}$ as partition family induced by $\CC$, i.e.,
$\CP_k:=\{P_{k,1},P_{k,2}\}$ where
$P_{k,1}:=C_k$ and $P_{k,2}:=J\setminus C_k$ for $k\in K$.
\begin{cs}
\case{$\PP$ is non-trivial}~\\
Return $\Lambda:=\recPart(J,A,\Delta^\ord,g^\ord,\PP)$ using \cref{lem:recPart}.\\
\comment{Since $|\PP|= |K|\leq 2^{\log_2(|V|)^4}$, we make progress \cref{prog:J} or \cref{prog:linJ}}
\case{there is a partition $\CP_k\in\PP$ that is the partition into singletons}~\\
Return $\Lambda:=\canObj(J)$ using \cref{cor:canObj0}.\\
\comment{Since $\CP_k\in\PP$ is the partition into singletons and has size $|\CP_k|\leq 2$, it follows that $|J|\leq 2$}
\end{cs}
\comment{Therefore, there is a singleton partition $\CP_k\in\PP$.
This means that there is a $k\in K$ that identifies a hyperedge in each hypergraph $H_i$.
We simplify to the case in which each $k\in K$ identifies a hyperedge in each hypergraph $H_i$}\\
Define $K:=\{k\in K\mid \CP_k=\{J\}=\{C_k\}$ is the singleton partition $\}$.\\
\comment{Now, each $k\in K$ identifies a hyperedge in each hypergraph $H_i$.
By definition, each $k\in K$ identifies exactly one hyperedge $S\in H_i$ in each hypergraph $H_i$}\\
Define $E_k:=\{S\mid k$ identifies the hyperedge $S\in H_i$ in some hypergraph $H_i\}$ for each $k\in K$.\\
\comment{By definition, $|E_k\cap H_i|= 1$ for all $k\in K$ and all hypergraphs $H_i$}\\
Define a partition family $\QQ:=\{\CQ_k\}_{k\in K}$
of $J=Q_{k,1}\cupdot\ldots\cupdot Q_{k,q_k}$
where $\CQ_k:=\{Q_{k,1},\ldots,Q_{k,q_k}\}$ such that:\\
$\Delta_i\rho_i,\Delta_j\rho_j\in Q_{k,x}$ for some $Q_{k,x}\in\CQ_k$, if and only if
$E_k\cap H_i =E_k\cap H_j$.
\begin{cs}
\case{$\CQ$ is non-trivial}~\\
Return $\Lambda:=\recPart(J,A,\Delta^\ord,g^\ord,\QQ)$ using \cref{lem:recPart}.\\
\comment{Since $|\QQ|=|K|\leq 2^{\log_2(|V|)^4}$, we make progress \cref{prog:J} or \cref{prog:linJ}}

\case{there is a partition $\CQ_k\in\QQ$ that is the singleton partition}~\\
\comment{This means that there is a $k\in K$ such that $E_k\cap H_i=E_k\cap H_j$ for all hypergraphs}\\
Define $\CS:=\{S\in E_k\cap H_i\mid\CQ_k\in\QQ$ is the singleton partition$\}$.\\
\comment{We have that $\CS\subseteq H_i$ for all hypergraphs $H_i$. The case that $\CS=H_i$ for all hypergraphs
cannot occur since we are
in a situation with $H_i\neq H_j$ for all $\Delta_i\rho_i\neq \Delta_j\rho_j\in J$}\\
Define $H_i':=H_i\setminus \CS$ for all hypergraphs $H_i$.\\
Go to the outer case with $H_i'$ in place of $H_i$.\\
\comment{Again, we have $H_i'\neq H_j'$ for all $\Delta_i\rho_i\neq \Delta_j\rho_j\in J$}

\case{all partitions $\CQ_k\in\QQ$ are partitions into singletons}~\\
\comment{This means that for all $k\in K$, the sets $E_k\cap H_1,\ldots,E_k\cap H_t$ are pairwise distinct}
\begin{for}{$k\in K$}
\comment{We will compute a canonical labeling for $(J,k)$.
We will define a coset-labeled hypergraph}\\
Define a function $\alpha_k:E_k\to J$ by setting $\alpha_k(S):=\Delta_i\rho_i$ for $\{S\}=E_k\cap H_i$.\\
\comment{This is well-defined, since $|H_i\cap E_k|=1$ and the sets $E_k\cap H_1,\ldots,E_k\cap H_t$ are pairwise distinct}\\
Compute $\Theta_k\tau_k:=\canSetHyper(E_k,J,\alpha_k)$ using \cref{lem:canSetHyper}.\\
\comment{The algorithm from \cref{lem:canSetHyper} runs in time $(|V|+|J|)^{\polylog|V|}$}
\end{for}
Define $K^\set:=\{\Theta_k\tau_k\mid k\in K\}$.\\
\comment{We collect the canonical labelings $\Theta_k\tau_k$ leading to minimal canonical forms of the input}\\
Define $K_{\min}^\set:=\arg\min_{\Theta_k\tau_k\in K^\set}J^{\tau_k}\subseteq K^\set$ where the minimum is taken w.r.t.~the ordering ``$\prec$'' from \cref{lem:prec}.\\
Return $\Lambda:=\langle K_{\min}^\set\rangle$.\\
\comment{This is the smallest coset containing all labeling cosets in $K_{\min}^\set$ as defined
in the preliminaries. The correctness proof for (CL2) is similar to the (CL2)-proof of \cref{lem:canSetSet}}
\end{cs}
\end{cs}
\comment{Now, the partition $\CP$ is the singleton partition, i.e., $H_i= H_j$ for all $\Delta_i\rho_i,\Delta_j\rho_j\in J$}\\
Define $H:=H_i$ for some $\Delta_i\rho_i\in J$.\\
\comment{This does not depend on the choice of $\Delta_i\rho_i\in J$}\\
For each $S\in H$, define a representative $\lambda_{i,S}\in\Delta_i\rho_i$ such that $S^{\lambda_{i,S}}=S^\ord$.\\
Define $\hat J:=\{\rho_i\delta^\ord_\ell\Psi^\ord\mid
i\in[t],\ell\in[s]\}=\{\lambda_{i,S}\Psi^\ord\mid i\in[t],S\in H\}$.\\
\comment{It follows that $\aut(\hat J)=\aut(J)$ using the same argument as in the $\toGroup$ subroutine}\\
Define $\hat J_S:=\{\lambda_{i,S}\Psi^\ord\mid i\in[t]\}$ for each $S\in H$ and define $\hat\CJ:=\{\hat J_S\mid S\in H\}$.\\
\comment{We claim that $\aut(\hat J_S)\leq\aut(J)$ for all $\hat J_S\in \hat\CJ$.
For all $\hat J_S\in \hat\CJ$ and all $\Delta_i\rho_i\in J$
there is a subcoset $\lambda_{i,S}\Psi^\ord\leq\Delta_i\rho_i$ in $\hat J_S$.
This proves the claim with the same
argument as in the $\toGroup$ subroutine}\\
Define a partition family $\PP:=\{\CP_S\}_{S\in H}$
of $J=P_{S,1}\cupdot\ldots\cupdot P_{S,p_S}$
where $\CP_S:=\{P_{S,1},\ldots,P_{S,p_S}\}$ such that:\\
$\Delta_i\rho_i,\Delta_j\rho_j\in P_{S,\ell}$ for some $P_{S,\ell}\in\CP_S$, if and only if
$(\lambda_{i,S}\Psi^\ord)|_{V\setminus A}=(\lambda_{j,S}\Psi^\ord)|_{V\setminus A}$.
\begin{cs}
\case{$\PP$ is non-trivial}~\\
Compute and return $\Lambda:=\recPart(J,A,\Delta^\ord,g^\ord,\PP)$ using \cref{lem:recPart}.\\
\comment{We have that $|\PP|=|H|\leq (\Delta^\ord:\Psi^\ord)\leq |V|^{3+\log_2|V|}$ and therefore we make progress
\cref{prog:J} or \cref{prog:linJ}}

\case{there is a partition $\CP_S\in\PP$ that is the partition into singletons}~\\
Return $\canObj(J)$ using \cref{cor:canObj0}.\\
\comment{Since $\CP_S$ is the partition into singletons and  $|\CP_S|\leq (\Delta^\ord:\Psi^\ord)\leq |V|^{3+\log_2|V|}$ is bounded, it follows that
$|J|\leq |V|^{3+\log_2|V|}$. Therefore, the algorithm from \cref{cor:canObj0} runs in time $2^{\polylog|V|}$}

\end{cs}
\comment{Now, all partitions $\CP_S\in\PP$ are singleton partitions.
This means that Property (A) holds for each instance $\hat J_S\in\hat\CJ$.
In the next steps, we analyze the sets $\lambda_{i,S}\Psi^\ord$ locally.
More precisely, we consider the restrictions $(\lambda_{i,S}\Psi^\ord)|_S$.
We consider different cases depending on whether these local restrictions coincide or not. We define the following partition family}\\
Define a partition family $\QQ:=\{\CQ_S\}_{S\in H}$
of $J=Q_{S,1}\cupdot\ldots\cupdot Q_{S,q_S}$
where $\CQ_S:=\{Q_{S,1},\ldots,Q_{S,q_S}\}$ such that:\\
$\Delta_i\rho_i,\Delta_j\rho_j\in Q_{S,\ell}$ for some $Q_{S,\ell}\in\CQ_S$, if and only if
$(\lambda_{i,S}\Psi^\ord)|_S=(\lambda_{j,S}\Psi^\ord)|_S$.
\begin{cs}
\case{$\QQ$ is non-trivial}~\\
Compute and return $\Lambda:=\recPart(J,A,\Delta^\ord,g^\ord,\QQ)$ using \cref{lem:recPart}.\\
\comment{We have that $|\QQ|=|H|\leq (\Delta^\ord:\Psi^\ord)\leq |V|^{3+\log_2|V|}$ and therefore we make progress
\cref{prog:J} or \cref{prog:linJ}}

\case{there is a partition $\CQ_S\in\QQ$ that is the partition into singletons}~\\
\comment{This means that there is $S\in H$
such that $(\lambda_{i,S}\Psi^\ord)|_S\neq(\lambda_{j,S}\Psi^\ord)|_S$ are pairwise distinct for all $\Delta_i\rho_i,\Delta_j\rho_j\in J$.
We simplify to the case in which $\CQ_S$ is the partition into singletons for all $\CQ_S\in\QQ$}\\
Define $H:=\{S\in H\mid \CQ_S\in\QQ$ is the partition into singletons$\}$.\\
\comment{Now, for all $S\in H$ the local restrictions are pairwise distinct}
\begin{for}{$S\in H$}
\comment{We will compute a canonical labeling for $(J,S)$}\\
Define ${\Psi^*}^\ord:=\Psi^\ord[S^\ord]\times\sym(V^\ord\setminus S^\ord)\geq\Psi^\ord$.\\
Define $\hat J_S^*:=\{\lambda_{i,S}{\Psi^*}^\ord\mid\lambda_{i,S}\Psi^\ord\in\hat J_S\}$.\\
\comment{Since $\CQ_S$ is the partition into singletons, it follows that $|\hat J_S^*|=|\hat J_S|$}\\
Define $A_S:=A\cap S$.\\
\comment{Since Property (A) holds for $\hat J_S$ with $A_S$ in place of $A$, it follows that
Property (A) also holds for the instance $\hat J_S^*$ (with $A_S$ in place of $A$)}\\
Define $\Theta^\ord\leq\Delta_T^\ord\leq\Psi^\ord$ to be the kernel of $g_T^\ord:\Delta_T^\ord\to\sym(T^\ord)$.\\
\comment{Observe that all points in $A_S^\ord\subseteq S^\ord$ are affected by $g_T^\ord$.
By the Affected Orbit Lemma \ref{theo:unaffected}, the $\Theta^\ord$-orbits of $A_S^\ord$
have size at most $|A_S^\ord|/|T^\ord|$}\\
Define ${\Theta^*}^\ord:=\stab_{{\Psi^*}^\ord}(A_1^\ord,\ldots,A_a^\ord)$ be the stabilizer of those orbits.\\
Let $c_\ind:=({\Psi^*}^\ord:{\Theta^*}^\ord)$ and let $c_\pid:=\pid({\Psi^*}^\ord,{\Theta^*}^\ord)$.\\
Apply the subroutine $\toGroup(\hat J_S^*,A_S,{\Psi^*}^\ord,{\Theta^*}^\ord)$ using \cref{lem:toGroup}.
\end{for}
\comment{We consider two cases depending on which option of \cref{lem:toGroup} is satisfied for the subroutine $\toGroup$}
\begin{cs}
\case{for all $S\in H$ the subroutine reduces to $c_\ind$-many instances with the subgroup ${\Theta^*}^\ord$}~\\
For each $S\in H$, define $\Delta_S\rho_S:=\toGroup(\hat J_S^*,A_S,{\Psi^*}^\ord,{\Theta^*}^\ord)$.\\
\comment{We analyze the recurrence. We have a multiplicative cost of $H\cdot c_\ind\leq \binom{|V|}{|T^\ord|}\cdot|T^\ord|!\leq
|V|^{|T^\ord|}\leq|V|^{3+\log_2|V|}$ and recursive
instances with $\orb_{A^\ord}({\Theta^*}^\ord)\leq|A_S^\ord|/|T^\ord|\leq\frac{1}{2}|A^\ord|$.
Therefore, we make progress \cref{prog:A}}\\
For each $S\in H$, define $\alpha_S:\hat J_S^*\to \hat J_S$ by setting $\alpha_S(\lambda_{i,S}{\Psi^*}^\ord):=\lambda_{i,S}\Psi^\ord$.\\
\comment{Observe that $\Delta_S\leq\aut(\hat J_S^*)$}\\
For each $S\in H$, compute $\Theta_S\tau_S:=\canSetSet(\hat J_S^*,\hat J_S,\alpha_S,\Delta_S\rho_S)$ using \cref{lem:canSetSet}.\\
\comment{The algorithm from \cref{lem:canSetSet} runs in time $(|V|+|\hat J_S^*|)^{\polylog|V|}$}\\
Define $H^\set:=\{\Theta_S\tau_S\mid S\in H\}$.\\
\comment{We collect the canonical labelings $\Theta_S\tau_S$ leading to minimal canonical forms of the input}\\
Define $H_{\min}^\set:=\arg\min_{\Theta_S\tau_S\in H^\set}J^{\tau_S}\subseteq H^\set$ where the minimum is taken w.r.t.~the ordering ``$\prec$'' from \cref{lem:prec}.\\
Return $\Lambda:=\langle H_{\min}^\set\rangle$.\\
\comment{This is the smallest coset containing labeling cosets in $H_{\min}^\set$ as defined
in the preliminaries. The correctness proof for (CL2) is similar to the (CL2)-proof of \cref{lem:canSetSet}}
\case{for some $S\in H$ the subroutine returns a non-trivial partition family $\hat\PP_S$ of $\hat J_S^*$}~\\
\comment{We simplify to the case in which we have a non-trivial partition family for all $S\in H$}\\
Define $H:=\{S\in H\mid$ the subroutine returns a partition family $\hat\PP_S$ for $\hat J_S^*\}$.\\
\comment{The partition family $\hat\PP_S$ for $\hat J_S$ also induces a partition family $\PP_S$ of $J$}\\
For each $S\in H$, define a non-trivial partition family $\PP_S:=\{\CP_S\mid\hat\CP_S\in\hat\PP_S\}$ of $J$ where
$\CP_S:=\{P_S\mid\hat P_S\in\hat\CP_S\}$ such that:
$\Delta_i\rho_i\in P_S$, if and only if
$\lambda_{i,S}{\Psi^*}^\ord\in\hat P_S$.\\
\comment{By taking a union, we combine all partition families into one single partition family $\PP$}\\
Define a non-trivial partition family $\PP:=\bigcup_{S\in H}\PP_S$ of $J$.\\
Return $\Lambda:=\recPart(J,A,\Delta^\ord,g^\ord,\PP)$.\\
\comment{We analyze the recurrence. In this case $|\PP|\leq |H|\cdot|\PP_S|\leq |H|\cdot c_\ind\cdot |V|^{c_\pid}$.
We have $c_\ind\leq |V|^{3+\log_2|V|}$ and by \cref{exa:pid}.\ref{exa:pid2}, we have
$c_\pid\leq 2\cdot\log_2(c_\ind)$.
In total, we have $|\PP|\leq 2^{\log_2(|V|)^4}$ which leads
to progress \cref{prog:J} or \cref{prog:linJ}}
\end{cs}
\end{cs}
\comment{Now it holds that $\CQ_S$ is the singleton partition for each $S\in H$.
This means that the local restrictions pairwise coincide.
More precisely, this means that
$(\lambda_{i,S}^{-1}\lambda_{j,S})[S^\ord]\in\Psi^\ord[S^\ord]$ for all $\lambda_{i,S}\Psi^\ord,\lambda_{j,S}\Psi^\ord\in \hat J_S$}
\begin{for}{$S\in H$}
\comment{Now, we compute automorphisms $G_S\leq\aut(\hat J_S)\leq\aut(J)$ for each $\hat J_S\in\hat\CJ$}\\
Define $G_S^\ord:=\langle(\Delta_{T,(U^\ord)}^\ord)^{\Psi^\ord}\rangle\trianglelefteq\Psi^\ord$, the normal closure of $\Delta_{T,(U^\ord)}^\ord$ in $\Psi^\ord$.\\
Define $G_S:=\lambda_{i,S} G_S^\ord\lambda_{i,S}^{-1}\leq\Delta_i\leq\sym(V)$ for some
$\Delta_i\rho_i\in J$.\\
\comment{We claim that $G_S$ depends neither on the choice of $\Delta_i\rho_i\in J$ nor on the choice of
the representative $\lambda_{i,S}\in\Delta_i\rho_i$.
First, we show that $G_S$ does not depend on the choice of the representative $\lambda_{i,S}\in\Delta_i\rho_i$.
Let $\lambda_{i,S}'\in\Delta_i\rho_i$ be a second representative.
Observe that $\lambda_{i,S}^{-1}\lambda_{i,S}'\in\Psi^\ord$ and since $G_S^\ord\trianglelefteq\Psi^\ord$
the permutation $\lambda_{i,S}^{-1}\lambda_{i,S}'$ normalizes $G_S^\ord$.
Equivalently, this means that $\lambda_{i,S}G_S^\ord\lambda_{i,S}^{-1}=\lambda_{i,S}'G_S^\ord{\lambda_{i,S}'}^{-1}$
which was what we wanted to show.
We show that $G_S$ does not depend on the choice of $\Delta_i\rho_i\in J$.
Let $\Delta_j\rho_j\in J$.
We have that $(\lambda_{i,S}^{-1}\lambda_{j,S})[S^\ord]\in\Psi^\ord[S^\ord]$.
and since $G_S^\ord[S^\ord]\trianglelefteq\Psi^\ord[S^\ord]$
the permutation $(\lambda_{i,S}^{-1}\lambda_{j,S})[S^\ord]$ normalizes $G_S^\ord[S^\ord]$.
Moreover, the permutation $(\lambda_{i,S}^{-1}\lambda_{j,S})[U^\ord]$ obviously
normalizes $G_S^\ord[U^\ord]=\Delta_{T,(U^\ord)}^\ord[U^\ord]=1$.
In total, $\lambda_{i,S}^{-1}\lambda_{j,S}$ normalizes
$G_S^\ord$ or equivalently
$\lambda_{i,S}G_S^\ord\lambda_{i,S}^{-1}=\lambda_{j,S}G_S^\ord\lambda_{j,S}^{-1}$}\\
\comment{In particular, $G_S\leq \aut(\hat J_S)\leq\aut(J)$}\\
\comment{In particular, $G_S^{\lambda_{i,S}}=G_S^\ord$ for all $\lambda_{i,S}\Psi^\ord\in\hat J$}
\end{for}
\comment{In the next step, we will consider the group of automorphisms $G$ generated by all groups $G_S$
and show that $G$ is a certificate of fullness}\\
Define $G\leq\sym(V)$ as the group generated by all $G_S$ for all $S\in H$.\\
\comment{We have $G\leq \aut(J)$ since $G_S\leq\aut(J)$ for all $S\in H$}\\
Define $G^\ord:=\langle(\Delta_{T,(U^\ord)}^\ord)^{\Delta^\ord}\rangle\trianglelefteq\Delta^\ord$, the normal closure of $\Delta_{T,(U^\ord)}^\ord$ in $\Delta^\ord$.\\
\comment{We claim that $G^{\rho_i}=G^\ord$ for all $\Delta_i\rho_i\in J$.
Let $\Delta_i\rho_i\in J$.
We have $G^{\rho_i}=\rho_i^{-1}\langle\lambda_{i,S} G_S^\ord\lambda_{i,S}^{-1}\mid S\in H\rangle\rho_i
=\langle\delta^\ord_\ell G_S^\ord{\delta^\ord_\ell}^{-1}\mid \ell\in[s]\rangle=\langle (G_S^\ord)^{\Delta^\ord}\rangle=G^\ord$}\\
\comment{We claim that $g^\ord:G^\ord\to\sym(W^\ord)$ is a giant representation.
Since $G^\ord\trianglelefteq\Delta^\ord$, it follows that $g^\ord(\Delta_{T,(U^\ord)}^\ord)\leq g^\ord(G^\ord)\trianglelefteq g^\ord(\Delta^\ord)$.
Moreover, each non-trivial normal subgroup of the giant $g^\ord(\Delta^\ord)$ is a giant as well}\\
Return the certificate of fullness $G$.
\end{proof}

\paragraph{Automorphism Lemma}
For an object $\CX\in\obj(V)$ and a group $G\leq\sym(V)$, we define $\CX^G:=\{\CX^g\mid g\in G\}\in\obj(V)$.

\begin{lem}[Automorphism Lemma]\label{lem:aut}
Let $\CX\in\obj(V)$ be an object,
let $G\leq\sym(V)$ be a group and let $\can$ be a canonical labeling function.
Assume that $\aut(\CX)\leq\aut(\CX^G)$.
Then, $\canObj(\CX^G):=G\can(\CX)$ defines a canonical labeling for $\CX^G$.
\end{lem}

\begin{proof}
We claim that $\aut(\CX^G)=G\aut(\CX)$.
The inclusion $G\aut(\CX)\leq\aut(\CX^G)$ follows by the assumption
Conversely, we show $\aut(\CX^G)\leq G\aut(\CX)$.
Let $\sigma\in\aut(\CX^G)$. Therefore, $\CX^{\sigma^{-1}}=\CX^{g^{-1}}$ for some $g\in G$.
This implies $g^{-1}\sigma\in\aut(\CX)$ and thus $\sigma\in G\aut(\CX)$.
\end{proof}

\begin{lem}\label{lem:processAut}
There is an algorithm $\processAut$ that gets as input a pair $(\CX,G)$
where $\CX=(J,A,\Delta^\ord,g^\ord)$ is a tuple
for which Property (A), (B) and (g) hold where $g^\ord$ is defined
and $G\leq\sym(V)$ is a fullness certificate.
In time $(|V|+|J|)^{\polylog|V|}$, the algorithm reduces the canonical labeling problem of $\CX$ to
canonical labeling of either
\begin{itemize}
  \item (progress \cref{prog:linJ}) two instances $(J_1,A,\Delta^\ord,g^\ord)$ and $(J_2,A,\Delta^\ord,g^\ord)$ with $|J_1|+|J_2|=|J|$, or
  \item (progress \cref{prog:J}) $2^{\log_2 p+\log_2(|V|)^4}$-many instances $(J_{k,i},A,\Delta^\ord,g^\ord)$
  of size $|J_{k,i}|\leq \frac{1}{p}|J|$ and to additionally $2^{\log_2(|V|)^4}$-many instances $(J_k,V,\Delta_k^\ord,\bot)$ of size $|J_k|\leq p$ for some $p\in\NN$
  with $1<p\leq\frac{1}{2}|J|$, or
  \item (progress \cref{prog:A}) $2^{\log_2(|V|)^3}$-many instances $(\hat J_i,A,\Psi^\ord,\bot)$ with $|\hat J_i|\leq |J|$ and such that
  $\orb_{A^\ord}(\Psi^\ord)\leq\frac{1}{2}\orb_{A^\ord}(\Delta^\ord)$.
\end{itemize}
\end{lem}

\paragraph{Intuition of Certificate Aggregation}

We describe the overall strategy of this subroutine.
Let us consider the less technical case in which $g^\ord(G^\ord)$ is the symmetric group (rather than the alternating group).
In this case, it holds that $G^\ord\Psi^\ord=\Delta^\ord$
where $\Psi^\ord$ is the kernel of $g^\ord$.
Similarly to the $\toGroup$ and $\processJohnson$ subroutine, we consider the decomposition of
$\Delta^\ord=\bigcup_{\ell\in[s]}\delta^\ord_\ell\Psi^\ord$ into left cosets of the kernel
and define $\hat J:=\{\rho_i\delta^\ord_\ell\Psi^\ord\mid i\in[t],\ell\in[s]\}$.
Again, we have $\aut(\hat J)=\aut(J)$.
The key observation is that $G$ is transitive on $\hat J$ since $(\rho_i\delta^\ord_\ell\Psi^\ord)^{g^{-1}}=\rho_ig^{\rho_i}\delta^\ord_\ell\Psi^\ord$
for all $g\in G$ and $G^\ord\Psi^\ord=\Delta^\ord$.

First, consider an easy case in which $J=\{\Delta_1\rho_1\}$ consists of one single labeling coset.
In this case, we have a set of automorphisms $G$ acting transitively on the subcosets $\hat J=\{\rho_1\delta^\ord_\ell\Psi^\ord\mid\ell\in[s]\}$.
Moreover, each subcoset satisfies $\aut(\rho_1\delta^\ord_\ell\Psi^\ord)\leq\aut(J)$ and can be seen as an individualization of $J$.
This means, we can choose (arbitrarily) a subcoset $\rho_1\delta^\ord_\ell\Psi^\ord\leq\Delta_1\rho_1$ and recurse on that.
Since the automorphisms in $G$ can map each subcoset to each other subcoset it does not matter which subcoset we choose.
By recursing on one single subcoset only, we can measure significant progress.
At the end, we return $G\hat\Lambda$ where $\hat\Lambda$ is a canonical labeling for the (arbitrarily) chosen subcoset
and $G$ is the group of automorphisms (acting transitively on the set of all subcosets).

However, the situation becomes more difficult when dealing with more labeling cosets $J=\{\Delta_1\rho_1,\ldots,\Delta_t\rho_t\}$ for $t\geq 2$.
The first idea that comes to mind is the following generalization.
We choose (arbitrarily) some $\ell\in[s]$ and define the set of subcosets $\hat J_\ell:=\{\rho_i\delta^\ord_\ell\Psi^\ord\mid i\in[t]\}\subseteq\hat J$.
The set $\hat J_\ell$ contains exactly one subcoset $\rho_i\delta^\ord_\ell\Psi^\ord\leq\Delta_i\rho_i$
of each $\Delta_i\rho_i\in J$.
However, the partition $\hat\CJ:=\{\hat J_\ell\mid \ell\in[s]\}$ might not be $G$-invariant
and $G$ might not be transitive on it.
The goal of the algorithm is to find a suitable partition $\hat\CJ:=\{\hat J_1,\ldots,\hat J_r\}$ of the subcosets $\hat J$ on which $G$ is transitive.

\begin{proof}[Proof of \cref{lem:processAut}]
\algorithmDAN{\processAut(J,A,\Delta^\ord,g^\ord,G)}~\\
Define $\Pi^\ord\trianglelefteq\Delta^\ord$ as the kernel of $g^\ord:\Delta^\ord\to\sym(W^\ord)$.\\
Define $M^\ord:=\sym({W^\ord})_{(\{3,\ldots,|W^\ord|\})}$, the pointwise stabilizer of all points excluding $1,2\in\NN$.\\
Define $\Psi^\ord:={g^\ord}^{-1}(M^\ord)\leq\Delta^\ord$.\\
\comment{It holds that $\Pi^\ord\leq\Psi^\ord\leq\Delta^\ord$, where the former subgroup relation is of index 2.
Moreover, $G^\ord\Psi^\ord=\Delta^\ord$}\\
We consider (but not compute) the decomposition of $\Delta^\ord=\bigcup_{\ell\in[s]}\delta^\ord_\ell\Psi^\ord$ into left cosets
and define $\hat J:=\{\rho_i\delta^\ord_\ell\Psi^\ord\mid \Delta_i\rho_i\in J,\ell\in[s]\}$.\\
\comment{This decomposition is for the analysis only and its computation is not part of the algorithm}
\begin{cs}
\case{$\Pi^\ord[V^\ord\setminus A^\ord]<\Delta^\ord[V^\ord\setminus A^\ord]$ is a subgroup of index greater than 2}~\\
Define the homomorphism $h:\Delta^\ord\to\Delta^\ord[V^\ord\setminus A^\ord]$ by restricting the image to $V^\ord\setminus A^\ord$.\\
Define $N^\ord:=\ker(h)\leq\Delta^\ord$ as the kernel of the homomorphism $h$.\\
\comment{We claim that $N^\ord\leq\Pi^\ord$.
Since $\Pi^\ord,N^\ord\trianglelefteq\Delta^\ord$, we have that $\Pi^\ord\leq \Pi^\ord N^\ord\trianglelefteq\Delta^\ord$.
Observe that $\Delta^\ord/\Pi^\ord$ is isomorphic to a giant
and all normal subgroups of a giant with index greater than 2 are trivial.
By assumption, $(\Delta^\ord:\Pi^\ord N^\ord)\geq (h(\Delta^\ord):h(\Pi^\ord N^\ord))> 2$.
By the Correspondence Theorem, $\Pi^\ord N^\ord=\Pi^\ord$ which proves the claim}\\
We consider (but not compute) the decomposition $\hat\CJ:=\{\hat J_1,\ldots,\hat J_r\}$ of $\hat J=\hat J_1\cupdot\ldots\cupdot \hat J_r$
such that:\\
$\rho_i\delta^\ord_\ell\Psi^\ord,\rho_{i'}\delta^\ord_{\ell'}\Psi^\ord\in \hat J_k$ for some $\hat J_k\in\hat\CJ$,
if and only if $(\rho_i\delta^\ord_\ell\Psi^\ord)|_{V\setminus A}$ equals $(\rho_{i'}\delta^\ord_{\ell'}\Psi^\ord)|_{V\setminus A}$.\\
\comment{We claim that $|\hat J_k|=|J|$ and $\aut(\hat J_k)\leq\aut(J)$.
We show a stronger statement, i.e.,
for each $k\in[r],\Delta_i\rho_i\in J$ there is exactly one $\ell\in[s]$ such that $\rho_i\delta^\ord_\ell\Psi^\ord\in \hat J_k$.
Property (A) implies that
for all $k\in[r],\Delta_i\rho_i\in J$ there is at least one $\ell\in[s]$ such that $\rho_i\delta^\ord_\ell\Psi^\ord\in \hat J_k$.
On the other side, let $\rho_i\delta^\ord_\ell\Psi^\ord,\rho_i\delta^\ord_{\ell'}\Psi^\ord\in\hat J_k$.
By definition of $\hat J_k$, it holds that $((\delta^\ord_\ell)^{-1}\delta^\ord_{\ell'})[V^\ord\setminus A^\ord]\in\Psi^\ord[V^\ord\setminus A^\ord]$
or equivalently $h((\delta^\ord_\ell)^{-1}\delta^\ord_{\ell'})\in h(\Psi^\ord)$.
Therefore, $(\delta^\ord_\ell)^{-1}\delta^\ord_{\ell'}\in\Psi^\ord N^\ord=\Psi^\ord$. Thus, $\ell=\ell'$ which proves the claim}\\
Define $\hat J_0:=\hat J_k$ for some arbitrarily chosen $k\in[r]$ (which can depend on the choice of $k\in[r]$).\\
\comment{To compute $\hat J_0$, one can use the Schreier-Sims algorithm as follows.
First, we pick $\rho_1\Psi^\ord$ and define $\hat J_0$ as the part $\hat J_k$ such that
$\rho_1\Psi^\ord\in\hat J_k$.
Then, we compute the coset $\Delta_i':=(\Delta_i\rho_i\rho_1^{-1})_{(V\setminus A)}$ (which is non-empty)
and pick an element $\delta'\in\Delta_i$ for each $\Delta_i\rho_i\in J$.
Then, $\delta'\rho_1\Psi^\ord\leq\Delta_i\rho_i$ and $\delta'\rho_1\Psi^\ord|_{V\setminus A}=\rho_1\Psi^\ord|_{V\setminus A}$
and therefore $\delta'\rho_1\Psi^\ord$ also belongs to $\hat J_0$. Therefore, we can compute the entire set $\hat J_0$.
We claim that $G$ is transitive on $\hat\CJ$ and therefore $(\hat J_0)^G=\hat\CJ$.
This follows from the fact that $G$ is transitive on $\hat J$ and that $\hat\CJ$ is an automorphism-invariant partition of $\hat J$}\\
Compute $\hat\Lambda:=\canSet(\hat J_0,A,\Psi^\ord,\bot)$ recursively.\\
\comment{As already observed, it holds that $\orb_{A^\ord}(\Psi^\ord)\leq
2\cdot\orb_{A^\ord}(\Pi^\ord)\leq 2|A^\ord|/|W^\ord|\leq\frac{1}{2}\orb_{A^\ord}(\Delta^\ord)$.
This leads to progress \cref{prog:A}}\\
Return $\Lambda:=G\hat\Lambda$.\\
\comment{Since $(\hat J_0)^G=\hat \CJ$, it follows that $\aut(\Lambda)=\aut(\hat\CJ)=\aut(\hat J)=\aut(J)$ by \cref{lem:aut}}
\end{cs}
\begin{cs}
\case{$\Pi^\ord[V^\ord\setminus A^\ord]\leq\Delta^\ord[V^\ord\setminus A^\ord]$ has index 1 or 2}~\\
For each $\Delta_i\rho_i\in J$, define $\Pi_i:=\rho_i\Pi^\ord\rho_i^{-1}\trianglelefteq\Delta_i\leq\sym(V)$.\\
\comment{The group $\Pi_i$ does not depend on the representative of $\Delta_i\rho_i$,
because the kernel $\Pi^\ord\trianglelefteq\Delta^\ord$ is a normal subgroup}\\
For each $\Delta_i\rho_i\in J$ define the $\Pi_i$-orbit partition $\CB_i:=\{B\subseteq A\mid B$ is a $\Pi_i$-orbit$\}$ of $A$.\\
Define an (unordered) partition $\CP:=\{P_1,\ldots,P_p\}$ of $J= P_1\cupdot\ldots\cupdot P_p$ such that:\\
$\Delta_i\rho_i,\Delta_j\rho_j\in P_\ell$ for some $P_\ell\in\CP$, if and only if
$\CB_i=\CB_j$.
\begin{cs}
\case{$\CP$ is a non-trivial partition}~\\
Return $\Lambda:=\recPart(J,A,\Delta^\ord,g^\ord,\{\CP\})$ using \cref{lem:recPart}.\\
\comment{We have $|\{\CP\}|=1$ which leads to progress
\cref{prog:J} or \cref{prog:linJ}}

\case{$\CP$ is the partition into singletons}~\\
\comment{This means that $\CB_i\neq \CB_j$ for all $\Delta_i\rho_i\neq\Delta_j\rho_j\in J$. We will define a non-trivial cover $\CC$}\\
Define a cover $\CC:=\{C_{vw}\mid (v,w)\in A^2\}$ of $J=\bigcup_{(v,w)\in A^2}C_{vw}$ such that:\\
$\Delta_i\rho_i\in C_{vw}$, if and only if $\{v,w\}\subseteq B$ for some $B\in\CB_i$.\\
Define $\PP:=\{\CP_{vw}\}_{(v,w)\in A^2}$ as partition family induced by $\CC$, i.e.,
$\CP_{vw}:=\{P_{vw,1},P_{vw,2}\}$ where
$P_{vw,1}:=C_{vw}$ and $P_{vw,2}:=J\setminus C_{vw}$ for $(v,w)\in A^2$.\\
Return $\Lambda:=\recPart(J,A,\Delta^\ord,g^\ord,\PP)$ using \cref{lem:recPart}.\\
\comment{We have $|\PP|=|A^2|\leq|V|^2$ which leads to
progress \cref{prog:J} or \cref{prog:linJ}}

\end{cs}
\comment{Now, the partition $\CP$ is the singleton partition. This means that $\CB_i=\CB_j$ for all $\Delta_i\rho_i,\Delta_j\rho_j\in J$}\\
Define $\CB:=\CB_i$ for some $\Delta_i\rho_i\in J$.\\
\comment{The partition $\CB$ does not depend on the choice of $\Delta_i\rho_i\in J$}\\
Define an (unordered) partition $\CQ:=\{Q_1,\ldots,Q_q\}$ of $J= Q_1\cupdot\ldots\cupdot Q_q$ such that:\\
$\Delta_i\rho_i,\Delta_j\rho_j\in Q_\ell$ for some $Q_\ell\in\CQ$, if and only if
$(\Delta_i\rho_i)[\CB]=(\Delta_j\rho_j)[\CB]$.

\begin{cs}
\case{$\CQ$ is a non-trivial partition}~\\
Compute and return $\Lambda:=\recPart(J,A,\Delta^\ord,g^\ord,\{\CQ\})$.\\
\comment{We have $|\{\CQ\}|=1$ which leads to progress \cref{prog:J} or \cref{prog:linJ}}

\case{$\CQ$ is a partition into singletons}~\\
\comment{This means that $\Delta_i\rho_i[\CB]\neq\Delta_j\rho_j[\CB]$ for all $\Delta_i\rho_i\neq\Delta_j\rho_j\in J$}\\
\comment{We will use the following strategy. We individualize a labeling coset $\Delta_1\rho_1\in J$
at a multiplicative cost of $|J|$. Then, we choose arbitrarily a subcoset $\rho_1\delta^\ord_\ell\Psi^\ord\leq\Delta_1\rho_1$
(no multiplicative cost since the group of automorphisms $G$ is transitive on the set of all possible chosen subcosets).
Again, we individualize a subcoset $\Gamma_k:=\rho_1\delta^\ord_\ell\psi^\ord_k\Pi^\ord\leq\rho_1\delta^\ord_\ell\Psi^\ord$
at a multiplicative cost of $2$.
With respect to the individualized subcoset $\Gamma_k$, we can define a linear ordering on $J$ and solve the canonization problem without further recursive calls}
\begin{for}{$\Delta_i\rho_i\in J$}
\comment{We will compute a canonical labeling for $(J,\Delta_i\rho_i)$}\\
We consider (but not compute) $\hat J_i:=\{\rho_i\delta^\ord_\ell\Psi^\ord\mid\ell\in[s]\}$ of $\Delta_i\rho_i$.\\
Define $\Gamma_{i,0}:=\rho_i\delta^\ord_\ell\Psi^\ord\in\hat J_i$ for some arbitrarily chosen $\ell\in[s]$ (which can depend on the choice of $\ell\in[s]$).\\
\comment{We will compute a canonical labeling for $(J,\Gamma_{i,0})$.
Again, $G$ is transitive on $\hat J_i$ and therefore $\Gamma_{i,0}^G=\hat J_i$}\\
Compute the decomposition of $\Psi^\ord=\psi^\ord_1\Pi^\ord\cupdot\psi^\ord_2\Pi^\ord$ into left cosets.\\
Decompose $\Gamma_{i,0}=\Gamma_{i,0,1}\cupdot\Gamma_{i,0,2}$ where $\Gamma_{i,0,2}:=\rho_i\delta^\ord_\ell\psi^\ord_k\Pi^\ord\leq\Delta_i\rho_i$ for $k=1,2$.
\begin{for}{$k=1,2$}
\comment{We will compute a canonical labeling for $(J,\Gamma_{i,0,k})$}\\
Compute $\Theta_{i,0,k,j}\tau_{i,0,k,j}:=\canSet(\Gamma_{i,0,k},\Delta_j\rho_j)$ using \cref{lem:canInt} or \cref{cor:canObj0} for each $\Delta_j\rho_j\in J$.\\
Rename indices $[t]$ such that:
$(\Gamma_{i,0,k},\Delta_1\rho_1)^{\tau_{i,0,k,1}}\prec\ldots\prec(\Gamma_{i,0,k},\Delta_t\rho_t)^{\tau_{i,0,k,t}}$.\\
\comment{We claim that the ordering is strict.
Assume that $(\Gamma_{i,0,k},\Delta_j\rho_j)^{\tau_{i,0,k,j}}=(\Gamma_{i,0,k},\Delta_{j'}\rho_{j'})^{\tau_{i,0,k,j'}}$.
On the one side, $\Gamma_{i,0,k}^{\tau_{i,0,k,j}}=\Gamma_{i,0,k}^{\tau_{i,0,k,j'}}$ implies $\tau_{i,0,k,j'}\tau_{i,0,k,j}^{-1}[\CB]=1$
and on the other side $(\Delta_j\rho_j)^{\tau_{i,0,k,j}}=(\Delta_{j'}\rho_{j'})^{\tau_{i,0,k,j'}}$ implies
$\tau_{i,0,k,j'}\tau_{i,0,k,j}^{-1}[\CB]\in\rho_{j'}\Delta^\ord\rho_j^{-1}[\CB]$.
Since $\Delta_j\rho_j[\CB]\neq\Delta_{j'}\rho_{j'}[\CB]$ for all $\Delta_j\rho_j\neq\Delta_{j'}\rho_{j'}\in J$, it follows
that $j=j'$ which proves the claim}\\
Define $\Theta_{i,0,k}\tau_{i,0,k}:=\canObj((\Delta_1\rho_1,\ldots,\Delta_t\rho_t))$ using \cref{cor:canObj0}.\\
\comment{Observe that $\Theta_{i,0,k}\tau_{i,0,k}$ defines a canonical labeling for $(J,\Gamma_{i,0,k})$}
\end{for}
Define $\Theta_{i,0}:=
\begin{cases}
\langle\Theta_{i,0,1}\tau_{i,0,1},\Theta_{i,0,2}\tau_{i,0,2}\rangle,&\text{if }J^{\tau_{i,0,1}}=J^{\tau_{i,0,2}}\\
\Theta_{i,0,k}\tau_{i,0,k},&\text{if }J^{\tau_{i,0,k}}\prec J^{\tau_{i,0,3-k}}.
\end{cases}$.\\
\comment{Observe that $\Theta_{i,0}\tau_{i,0}$ defines a canonical labeling for $(J,\Gamma_{i,0})$}\\
Define $\Theta_i\tau_i:=G\Theta_{i,0}\tau_{i,0}$.\\
\comment{We claim that $\Theta_i\tau_i$ defines a canonical labeling for $(J,\Delta_i\rho_i)$.
By \cref{lem:aut}, we have that $\Theta_i\tau_i$ defines a canonical labeling for $(J,\hat J_i)$
since $\Gamma_{i,0}^G=\hat J_i$. Moreover, $\hat J_i$ is an isomorphism-invariant partition of $\Delta_i\rho_i$ which proves the claim}
\end{for}
\comment{Next, we compute a canonical labeling $\Lambda$ for $J$}\\
Define $J^\set:=\{\Theta_i\tau_i\mid \Delta_i\rho_i\in J\}$\\
\comment{We collect the canonical labelings $\Theta_i\tau_i$ leading to minimal canonical forms of the input}\\
Define $J_{\min}^\set:=\arg\min_{\Theta_i\tau_i\in J^\set}J^{\tau_i}\subseteq J^\set$ where the minimum is taken w.r.t.~the ordering ``$\prec$'' from \cref{lem:prec}.\\
Return $\Lambda:=\langle J_{\min}^\set\rangle$.\\
\comment{This is the smallest coset containing labeling cosets in $J_{\min}^\set$ as defined
in the preliminaries. The correctness proof for (CL2) is similar to the (CL2)-proof of \cref{lem:canSetSet}}
\end{cs}
\comment{Now, $\CQ$ is the singleton partition.
This means that $\Delta_i\rho_i[\CB]=\Delta_j\rho_j[\CB]$ for all $\Delta_i\rho_i,\Delta_j\rho_j\in J$}\\
Define $\CB^\ord:=\{B_1^\ord,\ldots,B_b^\ord\}$ as the $\Pi^\ord$-orbit partition of $A^\ord$.\\
\comment{By definition, $\CB^{\rho_i}=\CB^\ord$ for all $\Delta_i\rho_i\in J$}\\
Define the homomorphism $h:\Delta^\ord\to\Delta^\ord[\CB^\ord]$.\\
Define $N^\ord:=\ker(h)\leq\Delta^\ord$ as the kernel of the homomorphism $h$.\\
\comment{Again, we claim that $N^\ord\leq\Pi^\ord$.
Since $\Pi^\ord,N^\ord\trianglelefteq\Delta^\ord$, we have that $\Pi^\ord\leq \Pi^\ord N^\ord\trianglelefteq\Delta^\ord$.
Observe that $\Delta^\ord/\Pi^\ord$ is isomorphic to a giant
and all normal subgroups of a giant with index greater than 2 are trivial.
We have $(\Delta^\ord:\Pi^\ord N^\ord)\geq (h(\Delta^\ord):h(\Pi^\ord N^\ord))=|h(\Delta^\ord)|$.
Since $\Delta^\ord$ is transitive on $A^\ord$ (Property (g)) it holds that $h(\Delta^\ord)=\Delta^\ord[\CB^\ord]$
is transitive on $\CB^\ord$ and therefore
$h(\Delta^\ord)|\geq |\CB^\ord|$.
By the Affected Orbits Lemma \ref{theo:unaffected},
each $B_i^\ord\in\CB^\ord$ has size at most $|A^\ord|/|W^\ord|$
and therefore $|\CB^\ord|\geq|W^\ord|$.
Moreover, $|W^\ord|\geq 3$ is assumed the be greater than some absolute constant.
By the Correspondence Theorem, $\Pi^\ord N^\ord=\Pi^\ord$ which proves the claim}\\
We consider (but not compute) the decomposition $\hat\CJ:=\{\hat J_1,\ldots,\hat J_r\}$ of $\hat J=\hat J_1\cupdot\ldots\cupdot \hat J_r$
such that:\\
$\rho_i\delta^\ord_\ell\Psi^\ord,\rho_{i'}\delta^\ord_{\ell'}\Psi^\ord\in \hat J_k$ for some $\hat J_k\in\CJ$,
iff $\rho_i\delta^\ord_\ell\Psi^\ord[\CB]=\rho_{i'}\delta^\ord_{\ell'}\Psi^\ord[\CB]$.\\
\comment{Again, we claim that $|\hat J_k|=|J|$ and $\aut(\hat J_k)\leq\aut(J)$.
We show a stronger statement, i.e.,
for each $k\in[r],\Delta_i\rho_i\in J$ there is at exactly one $\ell\in[s]$ such that $\rho_i\delta^\ord_\ell\Psi^\ord\in \hat J_k$.
Because of $\Delta_i\rho_i[\CB]=\Delta_j\rho_j[\CB]$, it holds for all $k\in[r],\Delta_i\rho_i\in J$ there is at least one $\ell\in[s]$ such that $\rho_i\delta^\ord_\ell\Psi^\ord\in \hat J_k$.
On the other side, let $\rho_i\delta^\ord_\ell\Psi^\ord,\rho_i\delta^\ord_{\ell'}\Psi^\ord\in\hat J_k$.
By definition of $\hat J_k$, it holds $((\delta^\ord_\ell)^{-1}\delta^\ord_{\ell'})[\CB^\ord]\in\Psi^\ord[\CB^\ord]$
or equivalently $h((\delta^\ord_\ell)^{-1}\delta^\ord_{\ell'})\in h(\Psi^\ord)$.
Therefore, $(\delta^\ord_\ell)^{-1}\delta^\ord_{\ell'}\in\Psi^\ord N^\ord=\Psi^\ord$. Thus, $\ell=\ell'$ which proves the claim}\\
Define $\hat J_0:=\hat J_k$ for some arbitrarily chosen $k\in[r]$ (which can depend on the choice of $k\in[r]$).\\
\comment{Again, $G$ is transitive on $\hat\CJ$ and therefore $(\hat J_0)^G=\hat\CJ$}\\
Define $K:=\{(\rho_i\delta^\ord_\ell\Psi^\ord)|_{V\setminus A}\mid \rho_i\delta^\ord_\ell\Psi^\ord\in\hat J_0\}$.
\begin{cs}
\case{$|K|=1$}~\\
\comment{In this case, Property (A) is satisfied for $\hat J_0$}\\
Compute $\hat\Lambda:=\canSet(\hat J_0,A,\Psi^\ord,\bot)$ recursively.\\
\comment{As before, it holds that $\orb_{A^\ord}(\Psi^\ord)\leq\frac{1}{2}\orb_{A^\ord}(\Delta^\ord)$.
This leads to progress \cref{prog:A}}\\
Return $\Lambda:=G\hat\Lambda$.\\
\comment{Since $(\hat J_0)^G=\hat\CJ$, follows that $\aut(\Lambda)=\aut(\hat\CJ)=\aut(\hat J)=\aut(J)$ by \cref{lem:aut}}
\case{$|K|\geq 2$}~\\
\comment{Actually, we are in a case in which $|K|=2$ since we are still in the case in which $\Pi^\ord[V^\ord\setminus A^\ord]\leq\Delta^\ord[V^\ord\setminus A^\ord]$ has index 1 or 2}\\
For both $k\in K$, define $\hat J_{0,k}:=\{\rho_i\delta^\ord_\ell\Psi^\ord\in\hat J_0\mid
(\rho_i\delta^\ord_\ell\Psi^\ord)|_{V\setminus A}=k\}$.\\
\comment{Now, Property (A) is satisfied for both $\hat J_{0,k}$}\\
Compute $\Theta_k\tau_k:=\canSet(\hat J_{0,k},A,\Psi^\ord,\bot)$ recursively for both $k\in K$.\\
\comment{Again, $\orb_{A^\ord}(\Psi^\ord)\leq\frac{1}{2}\orb_{A^\ord}(\Delta^\ord)$.
The multiplicative cost is 2 which leads to progress \cref{prog:A}}\\
Define $\hat\CJ_0^\set:=\{(\Theta_k\tau_k,(\hat J_{0,k})^{\tau_k})\mid k\in K\}$.\\
Compute $\hat\Lambda:=\canObj(\hat\CJ_0^\set)$ using \cref{cor:canObj0}.\\
\comment{By \cref{lem:rep},
it follows that $\hat\Lambda$ defines a canonical labeling for $\hat J_0$}\\
Return $\Lambda:=G\hat\Lambda$.\\
\comment{Since $(\hat J_0)^G=\hat\CJ$, follows that $\aut(\Lambda)=\aut(\hat\CJ)=\aut(\hat J)=\aut(J)$ by \cref{lem:aut}}
\end{cs}
\end{cs}
\end{proof}

We have all tools together to give the algorithm for \cref{theo:canSet}.

\begin{proof}[Proof of \cref{theo:canSet}]
\algorithmDAN{\canSet(J,A,\Delta^\ord,g^\ord)}
\begin{cs}
\case{Property (B) is not satisfied}~\\
We recurse as described in the beginning of this section.

\case{$g^\ord=\bot$ is undefined}~\\
Recurse and return $\Lambda:=\toJohnson(J,A,\Delta^\ord,\bot)$ using \cref{lem:toJohnson}.
\case{$g^\ord$ is defined}~\\
Apply the subroutine $\processJohnson(J,A,\Delta^\ord,g^\ord)$ using \cref{lem:processJohnson}.
\begin{cs}
\case{the subroutine returns a certificate of fullness $G\leq\sym(V)$}~\\
Return $\Lambda:=\processAut(J,A,\Delta^\ord,g^\ord,G)$ using \cref{lem:processAut}.

\case{the subroutine finds a canonical labeling $\Lambda$ using recursion}~\\
Return $\Lambda$.
\end{cs}
\end{cs}
\begin{runtime}
The number of recursive calls of the algorithm $\canSet$ is bounded $T\leq(|V|+|J|)^{\polylog|V|}$
where $T$ is the function given in \cref{runtime}.
Also each recursive call takes time bounded in $(|V|+|J|)^{\polylog|V|}$.
\end{runtime}
\end{proof}

By improving the running time of \cref{prob:CL:Set}, we also obtain an improved version of \cref{cor:canObj0}.

\begin{cor}\label{cor:canObj}
Canonical labelings for combinatorial objects can be computed in time $n^{\polylog|V|}$
where $n$ is the input size and $V$ is the ground set of the object.
\end{cor}
\section{Isomorphism of Graphs Parameterized by Treewidth}\label{sec:treewidth}

\paragraph{Graph Theory}
We write $N_G(v):=\{w\in V(G)\mid \{v,w\}\in E(G)\}$ to denote the (open) neighborhood of $v\in V(G)$
in a graph $G$.
We also write $N_G(S):=\bigcup_{v\in S}N_G(v)$ to denote the (open) neighborhood of a subset $S\subseteq V(G)$.
We write $G[U]$ to denote the subgraph induced by $U\subseteq V(G)$ in $G$.

\begin{defn}[Tree Decomposition]
A \emph{tree decomposition} of a graph $G$ is a pair $(T,\beta)$ where
$T$ is a tree and $\beta: V(T)\to 2^{V (G)}$ is a function
that assigns each node $t\in V(T)$ a subset $\beta(t)\subseteq V(G)$, called \emph{bag}, such that:
\begin{enumerate}[(\textnormal{T}1)]
\item for each vertex $v\in V (G)$, the induced subtree
$T[\{t \in V (T)\mid v \in \beta(t)\}]$
is non-empty and connected, and
\item for each edge $e\in E(G)$, there exists $t\in V(T)$
such that $e\subseteq\beta(t)$.
\end{enumerate}
\end{defn}
The sets $\beta(s)\cap\beta(t)$
for $\{s,t\} \in E(T)$ are called
the \emph{adhesion sets}.
The \emph{width} of a tree decomposition $T$ is equal to its maximum
bag size decremented by one, i.e. $\max_{t\in V (T)} |\beta(t)| - 1$.
The \emph{treewidth} of a graph,
denoted by $\tw G$, is equal to the minimum width
among all its tree decompositions.

\paragraph{Separations and Separators}
Let $G=(V,E)$ be a graph and let $v,w\in V(G)$.
A pair $(A,B)$ is called a \emph{$(v,w)$-separation} if $A\cup B=V(G)$ and $v\in A\setminus B,w\in B\setminus A$ and there are no edges with one vertex in $V(G)\setminus A$
and the other vertex in $V(G)\setminus B$.
In this case, $A\cap B$ is called a \emph{$(v,w)$-separator}.
A separator $A\cap B$ is called \emph{clique separator} if $A\cap B$ is a clique in $G$.
Among all $(v,w)$-separations $(A,B)$ with minimal $|A\cap B|$ there is a unique separation $(A^*,B^*)$ with an inclusion minimal $A^*$.
In this case, $S_{v,w}:=A^*\cap B^*$ is called the \emph{leftmost minimal $(v,w)$-separator}.
It is known that $S_{v,w}$ can be computed in polynomial time using the Ford-Fulkerson algorithm.

\paragraph{Improved Graphs}
The \emph{$k$-improvement} of a graph $G$ is the graph $G^k$ obtained from $G$
by connecting every pair of non-adjacent
vertices $v,w$ for which there are more than $k$ pairwise internally vertex disjoint
paths connecting $v$ and $w$\footnote{In \cite{DBLP:journals/siamcomp/LokshtanovPPS17}, a slightly different notion of improvement is used
where an edge is also added when there are exactly $k$ pairwise internally vertex disjoint paths connecting non-adjacent vertices.}.
The \emph{separability} of a graph $G$, denoted by $\sep G$, is the smallest integer
$k$ such that $G^k=G$.
Equivalently, $\sep G$ equals the maximum size $|S_{v,w}|$ of a leftmost minimal separator among all non-adjacent vertices $v,w\in V(G)$.

The next lemma says that one can $k$-improve a graph for some $k\geq \tw G$
and reduce the separability of that graph while preserving the treewidth of that graph. 

\begin{lem}[\cite{DBLP:journals/siamcomp/LokshtanovPPS17}]\label{lem:improve}
Let $G$ be a graph and $k\in\NN$.
\begin{enumerate}
  \item There is a polynomial-time algorithm that for a given $(G,k)$ computes $G^k$.
  \item It holds that $(G^k)^k=G^k$ and therefore $\sep G^k\leq k$.
  \item Every tree decomposition of $G$ of width at most $k$ is also a tree decomposition of $G^k$
and therefore $\tw G\leq k$ implies $\tw G^k=\tw G$.
\end{enumerate}
\end{lem}

The next theorem says that one can decompose a graph into clique-separator-free graphs.
By possibly introducing new bags, we can assume that the adhesion sets inside each bag are either pairwise equal
or pairwise distinct. This ensures the third property in the following theorem.

\begin{theo}[\cite{DBLP:journals/dm/Leimer93},\cite{DBLP:conf/stacs/ElberfeldS16}]
\label{theo:cliqueDecomposition}
Let $G$ be a graph.
There is an algorithm that, given a graph $G$, computes a
tree decomposition $(T,\beta)$ with the following properties.

\begin{enumerate}
\item For every $t\in V(T)$ the graph $G[\beta(t)]$
is clique-separator free,
\item each adhesion set of $(T,\beta)$ is a clique in $G$, and
\item for each bag $\beta(t)$ either the adhesion sets are all equal and $|\beta(t)|\leq (\tw G)+1$ or
the adhesion sets are pairwise distinct.
\end{enumerate}
The algorithm runs in polynomial time and the output of the algorithm is isomorphism-invariant.
\end{theo}

We make use of the bounded-degree graph isomorphism algorithm given by Grohe, Neuen and Schweitzer.
In fact, they proved a stronger statement and designed a string isomorphism algorithm
for groups of bounded composition-width.
This implies the following result.

\begin{theo}[\cite{DBLP:conf/focs/GroheNS18}]\label{theo:boundedDegree}
Let $G_1,G_2$ be two graphs and let $\Delta\phi\leq\iso(V(G_1);V(G_2))$.
There is an algorithm that, given a triple $(G_1,G_2,\Delta\phi)$,
computes the set of isomorphisms $\iso(G_1;G_2)\cap\Delta\phi$ in time $|V(G_1)|^{\polylog(\cw\Delta)}$.
\end{theo}

We give an isomorphism algorithm
for the clique-separator-free graphs.
The algorithm uses the ideas from \cite{DBLP:conf/icalp/GroheNSW18}.

\begin{lem}\label{lem:isoBasic}
Let $G_1,G_2$ be two clique-separator-free graphs.
There is an algorithm that, given a pair $(G_1,G_2)$,
computes the set of isomorphisms $\iso(G_1;G_2)$ in time $|V(G_1)|^{\polylog(\tw G_1+\sep G_1)}$.
Moreover, there is a vertex $v_1\in V(G_1)$ such that $\cw\aut(G_1)_{(v_1)}\leq\max(\tw G_1,\sep G_1)$.
\end{lem}

\begin{proof}
Let $\mindeg G_i:=\min_{v\in V(G_i)}|N_{G_i}(v)|$ be the minimal degree among all vertices.
It is well-known that $\mindeg G_i\leq \tw G_i$ for all graphs $G_i$.
Let $\CS_i:=\{N_{G_i}(v)\subseteq V(G_1)\mid v\in V(G_i),|N_{G_i}(v)|=\mindeg G_i\}$ be the non-empty set of minimal size neighborhoods for both $i=1,2$.
We assume $|\CS_1|=|\CS_2|$, otherwise we reject isomorphism.
Since $G_i$ does not have clique separators, it follows that each $S_i\in\CS_i$ is not a clique for both $i=1,2$.
Since $\iso(G_1;G_2)=\bigcup_{S_1\in\CS_1,S_2\in\CS_2}\iso(G_1,S_1;G_2,S_2)$, it suffices to compute the isomorphisms
from $G_1$ to $G_2$ that map $S_1$ to $S_2$ for all possible choices of $S_1\in\CS_1,S_2\in\CS_2$.

We give an algorithm that gets as input $(G_1,S_1,G_2,S_2,\Delta\phi)$
where $S_i\subseteq V(G_i)$ is not a clique for both $i=1,2$ and $\Delta\phi\geq\iso(G_1,S_1;G_2,S_2)[S_1]$ with $\cw\Delta\leq \max(\tw G_1,\sep G_1)$.
The algorithm outputs $\iso(G_1,S_1;G_2,S_2)$.
Initially, we call the algorithm for some $S_1\in\CS_1,S_2\in\CS_2$
and $\Delta\phi:=\sym(S_1)\phi$ for some bijection $\phi:S_1\to S_2$.

\algorithmDAN{\isoBasic(G_1,S_1,G_2,S_2,\Delta\phi)}
\begin{cs}
\case{$S_1\subsetneq V(G_1)$}~\\
Let $S_i':=S_i\cup\bigcup_{v,w\in S_i,\{v,w\}\nin E(G_i)} S_{v,w}$ for both $i=1,2$. We claim that $S_i'\supsetneq S_i$ for both $i=1,2$.
Let $Z_i\subseteq V(G_i)$ be the vertex set of a connected component of $G_i-S_i$.
Since $G_i$ does not have clique separators, it follows that $N_{G_i}(Z_i)$ is not a clique.
Therefore, there are $v,w\in N_{G_i}(Z_i)\subseteq S_i$ with $\{v,w\}\nin E(G_i)$.
Moreover, there is a path from $v$ to $w$ with all internally vertices lying in $Z_i$.
Therefore, $S_{v,w}\cap Z_i\neq\emptyset$
and thus $S_i'\supseteq S_i\cupdot(S_{v,w}\cap Z_i)\supsetneq S_i$.
Observe that $|S_{v,w}|\leq \sep G_i$ for all $v,w\in S_i$ and both $i=1,2$.

First, we ensure for all $\phi\in\Delta\phi$ that $S_{v,w}$ and $S_{\phi(v,w)}$ have the same cardinality.
To do so, we define an edge relation $X_i^k:=\{(v,w)\mid |S_{v,w}|=k\}$ for each $k\leq n:=|V(G_i)|$ and both $i=1,2$.
We compute $\Delta\phi:=\iso(X_1^1,\ldots,X_1^n;X_2^1,\ldots,X_2^n)\cap\Delta\phi$ using \cref{theo:boundedDegree}.

Second, we define a wreath product with  $\sym(S_{v,w})$ and $\Delta$.
More precisely, we define $\hat S_i:=S_i\cupdot\bigcupdot_{v,w\in S_i,\{v,w\}\nin E(G_i)} \hat S_{v,w}$
where $\hat S_{v,w}:=S_{v,w}\times\{(v,w)\}$ is a disjoint copy of $S_{v,w}$ for both $i=1,2$.
We define $\hat\Delta\hat\phi\leq\iso(\hat S_1;\hat S_2)$ as
$\{\hat\phi:\hat S_1\to\hat S_2\mid\hat\phi[S_1]\in\Delta\phi,\forall v,w\in S_i:
\hat\phi(\hat S_{v,w})=\hat S_{\hat\phi(v),\hat\phi(w)}\}$.
Observe that $\cw\hat\Delta\leq\max(\max_{v,w\in S_1}|S_{v,w}|,\cw\Delta)\leq\max(\tw G_1,\sep G_1)$.

Third, we define the group $\Delta'\phi'\geq\iso(G_1,S_1';G_2,S_2')[S_1']$
by identifying the corresponding vertices.
More precisely, we define an edge relation $X_i:=\{((s,v,w),(s,v',w'))\in\hat S_{v,w}\times\hat S_{v',w'}\}\cup\{((s,v,w),s)\in \hat S_{v,w}\times S_i\}$ for both $i=1,2$.
Observe that $S_i'$ can be identified with the equivalence classes of $X_i$ for both $i=1,2$.
Now, compute $\Delta'\phi':=(\iso(X_1;X_2)\cap\hat\Delta\hat\phi)[S_1']$
using \cref{theo:boundedDegree}.

Finally, we compute and return $\isoBasic(G_1,S_1',G_2,S_2',\Delta'\phi')$ recursively.

\case{$S_1=V(G_1)$}~\\
Compute and return $\iso(G_1;G_2)\cap\Delta\phi$ using \cref{theo:boundedDegree}.

\end{cs}
\begin{runtime}
The number of recursive calls is bounded by $|V(G_1)|$.
In each call, we use the algorithm from \cref{theo:boundedDegree} which runs in time $|V(G_1)|^{\polylog(\tw G_1+\sep G_1)}$.
\end{runtime}
\end{proof}

With the above algorithm it is possible to compute the isomorphisms between the clique-separator-free parts
of the decomposition from \cref{theo:cliqueDecomposition}.
The adhesion sets (which are the intersections between two clique-separator-free graphs)
are cliques in the graph.
The next lemma is used in order to respect the adhesion sets of the clique-separator-free parts.
Also this lemma uses an idea similar to \cite{DBLP:conf/icalp/GroheNSW18}, Lemma 14 arXiv version.

\begin{lem}\label{lem:isoBasicClique}
Let $G_1,G_2$ be two clique-separator-free graphs
and let $H_1\subseteq 2^{V(G_1)},H_2\subseteq 2^{V(G_2)}$
be sets that contain cliques in the graphs $G_1,G_2$, respectively.
There is an algorithm that, given a tuple $(G_1,H_1,G_2,H_2)$,
computes the set of isomorphisms $\iso(G_1,H_1;G_2,H_2)$ in time $|V(G_1)|^{\polylog(\tw G_1+\sep G_1)}$.
\end{lem}

\begin{proof}
In the first step, we define a cover capturing all cliques.
More precisely, we claim that there is a function $\alpha:V(G)\to K$
where $K\subseteq 2^{V(G)}$
such that 
\begin{enumerate}
  \item $|S|\leq(\tw G)+1$ for all $S\in K$, and
  \item if $C\subseteq V(G)$ is a clique, then there is a $S\in K$ with $C\subseteq S$.
\end{enumerate}
Moreover, this function is polynomial-time computable and is defined in an isomorphism-invariant way.
This can be shown by induction on $|V(G)|$.
If $|V(G)|\leq 1$ the statement is trivially true.
So assume $|V(G)|\geq 2$.
Let $\mindeg G:=\min_{v\in V(G)}|N_G(v)|$ be the minimal degree among all vertices.
It is well-known that $\mindeg G\leq\tw G$.
Let $U:=\arg\min_{v\in V(G)}|N_G(v)|\subseteq V(G)$ be the non-empty set of vertices of minimal degree.
Let $G':=G[V(G)\setminus U]$ and let $\alpha':V(G')\to K'$ be the function obtained inductively.
Now, define $\alpha:V(G)\to K$ as $\alpha(v):=\alpha'(v)$ if $v\in V(G')$
and as $\alpha(v):=N_G(v)\cup\{v\}$ if $v\in U$.
It easily follows that $|\alpha(v)|\leq \tw G_1+1$.
Moreover, if $C\subseteq V(G)$ is a clique, then
either $C\subseteq V(G')$ or there is a vertex $v\in C\cap U$.
In the latter case, $C\subseteq N_G(v)\cup\{v\}=\alpha(v)$.
This proves the claim.
We will use this claim in order to compute the isomorphisms between the instances.

First of all, we compute the isomorphism $\Delta\phi:=\iso(G_1;G_2)$ between the two graphs $G_1$ and $G_2$
using \cref{lem:isoBasic}.

It remains to respect the hyperedges.
We compute $\alpha_i:V(G_i)\to K_i$ for both $i=1,2$.
We define hypergraphs $H_{v_i}:=\{S\in H_i\mid S\subseteq\alpha_i(v_i)\}\subseteq H_i$ for $v_i\in V(G_i)$ and both $i=1,2$.
For each pair of vertices $(v_1,v_2)\in V(G_1)\times V(G_2)$,
we compute the isomorphisms $\Delta_{v_1}\phi_{v_1,v_2}:=\iso(H_{v_1};H_{v_2})$ (seen as hypergraphs on $\alpha_1(v_1),\alpha_2(v_2)$, respectively)
using \cref{theo:canHyper}.
Since $|\alpha_1(v_1)|\leq(\tw G_1)+1$, the algorithm runs in the desired time bound.

First, we ensure for all $\phi\in\Delta\phi$ that the hypergraphs $H_{v_1}$ and $H_{\phi(v_1)}$ are isomorphic.
To do so, we define a vertex-colored graph $X_i$ that colors a vertex $v_i\in V(X_i)$ according to the isomorphism type
of $H_{v_i}$ for both $i=1,2$.
We compute $\Delta\phi:=\iso(X_1;X_2)\cap\Delta\phi$ using \cref{theo:boundedDegree}.
To analyze the running time of the algorithm from \cref{theo:boundedDegree},
we observe that $\cw\Delta$ is not necessarily bounded by $\max(\tw G_1,\sep G_1)$.
However, there is a point $v_1\in V(G_1)$ such that $\cw\Delta_{(v_1)}\leq\max(\tw G_1,\sep G_1)$.
By applying \cref{theo:boundedDegree} to each coset of the subgroup $\Delta_{(v_1)}$, we still achieve a running time
of $|V(G_1)|^{\polylog(\tw G_1+\sep G_1)}$.

Second, we define a wreath product $\Delta_{v_1}\phi_{v_1,v_2}$ with $\Delta\phi$.
More precisely, we define $U_i:=V(G_i)\cupdot\bigcupdot_{v\in V(G_i)} S_v$ where $S_v:=\alpha_i(v)\times\{v\}$ is
a disjoint copy of $\alpha_i(v)$ for both $i=1,2$.
We define $\hat\Delta\hat\phi\leq\iso(U_1;U_2)$ as $\{\hat\phi:U_1\to U_2\mid\hat\phi[V(G_1)]\in\Delta\phi,
\forall v_1\in V(G_1)\exists\phi_{v_1}\in\Delta_{v_1}\phi_{v_1,\hat\phi(v_1)}:\hat\phi(u,v_1)=(\phi_{v_1}(u),\hat\phi(v_1))\}$.
Again, there is a point $v_1\in V(G_1)$ such that $\cw\hat\Delta_{(v_1)}\leq \max((\tw G_1)+1,\sep G_1)$, which
can be seen as follows.
Consider the homomorphism $h$ that restricts $\hat\Delta_{(v_1)}$ to the set $V(G_1)$.
By construction, the image of $h$ is $\Delta_{(v_1)}$ which composition-width is bounded by $\max(\tw G_1,\sep G_1)$.
Moreover, the kernel of $h$ has orbits bounded by $|S_v|=|\alpha(v)|\leq(\tw G_1)+1$ for some $v\in V(G_1)$, which gives us the bound.

Third, we define the group $\Delta\phi:=\iso(G_1,H_1;G_2,H_2)$
by identifying the corresponding vertices.
More precisely, we define an edge relation $X_i:=\{((u,v),(u,v'))\in S_v\times S_{v'}\}\cup\{((u,v),u)\in S_v\times V(G_i)\}$ for both $i=1,2$.
Observe that $V_i$ can be identified with the equivalence classes of $X_i$ for both $i=1,2$.
Then, we compute and return $\Delta\phi:=(\iso(X_1;X_2)\cap\hat\Delta\hat\phi)[V(G_1)]$
using \cref{theo:boundedDegree}.
By applying \cref{theo:boundedDegree} to each coset of $\Delta_{(v)}\leq\Delta$, we still achieve a running time of $|V(G_1)|^{\polylog(\tw G_1+\sep G_1)}$.
\end{proof}

We recall the coset-labeled hypergraphs from
\cref{lem:canSetHyper}.
A coset-labeled hypergraph is a tuple $(V,H,J,\alpha)$ where $H$ is a set of hyperedges $S_i\subseteq V$
and $J$ is a set of labeling cosets $\Delta_i\phi_i\leq\lab(V)$ and $\alpha:H\to J$ is a function with $\alpha(S_i)=\Delta_i\phi_i$.

\begin{lem}[\cite{DBLP:conf/icalp/GroheNSW18}, Lemma 9, see \cite{neuenDiss}]\label{lem:cosetHyper}
Let $\CH_1=(V_1,H_1,J_1,\alpha_1),\CH_1=(V_2,H_2,J_2,\alpha_2)$ be two coset-labeled hypergraphs and let $\Delta\phi\leq\iso(H_1;H_2)$
be a coset that maps the hyperedges in $H_1$ to hyperedges in $H_2$.
There is an algorithm that, given a triple $(\CH_1,\CH_2,\Delta\phi)$,
computes the set of isomorphisms $\iso(\CH_1;\CH_2)\cap\Delta\phi$ in time $(|V_1|+|H_1|)^{\polylog(\cw\Delta)}$.
\end{lem}

\begin{theo}\label{theo:isoGraph}
Let $G_1,G_2$ be two connected graphs.
There is an algorithm that, given a pair $(G_1,G_2)$,
computes the set of isomorphisms $\iso(G_1;G_2)$ in time $|V(G_1)|^{\polylog(\tw G_1)}$.
\end{theo}

\begin{proof}
It is known that the treewidth can be approximated (up to a logarithmic factor) in polynomial time \cite{DBLP:conf/uai/Amir01}.
Let $k\in(\tw G_1)^{\CO(1)}$ be such an upper bound on the treewidth of $G_1$.
We compute the $k$-improved graphs $G_i^k$ using \cref{lem:improve} for both $i=1,2$.
We compute the tree decompositions $(T_i,\beta_i)$ from \cref{theo:cliqueDecomposition} for
the $k$-improvements $G_i^k$ for both $i=1,2$.
In particular, $(T_i,\beta_i)$ is also a tree decomposition for $G_i$ for both $i=1,2$.
Let $(T_i,\beta_i,r_i)$ denote the tree decomposition of $G_i$ rooted at $r_i\in V(T_i)$ for both $i=1,2$.
It suffices to give an algorithm for rooted tree decompositions.
We give an algorithm that gets as input $(\hat\CX_1,\hat\CX_2)$
where $\hat\CX_i=(\hat G_i,T_i,\beta_i,r_i,S_i)$ for both $i=1,2$
and outputs the set of isomorphisms $\iso(\hat\CX_1;\hat\CX_2)$.

\algorithmDAN{\isoTree(\hat\CX_1,\hat\CX_2)}~\\
Define $G_i:=\hat G_i[\beta_i(r_i)]$ as the induced subgraph for both $i=1,2$.\\
Define $Y_i:=\hat G_i^k[\beta_i(r_i)]$ as the induced subgraph of the $k$-improved graph $\hat G_i^k$ for both $i=1,2$.\\
Define $C_i:=\{t_{i,1},\ldots,t_{i,|C_i|}\}\subseteq V(T_i)$ as the set of children of $r_i$ for both $i=1,2$.\\
Define $H_i:=\{\beta_i(r_i)\cap \beta_i(t)\mid t\in C_i\}$ as the set of adhesion sets of the root $r_i$ for both $i=1,2$.\\
Let $(T_{i,t},\beta_{i,t})$ be the tree decomposition of the subtree rooted at $t\in C_i$ and let
$\hat G_{i,t}$ be the subgraph of $\hat G_i$ corresponding to $(T_{i,t},\beta_{i,t})$.\\
Define $\hat\CX_{i,t}:=(\hat G_{i,t},T_{i,t},\beta_{i,t},t,\beta_i(r_i)\cap\beta_i(t))$ for all $t\in C_i$ and both $i=1,2$.\\
Recursively, compute $\hat\Delta_{t_1}\hat\phi_{t_1,t_2}:=\isoTree(\hat\CX_{1,t_1},\hat\CX_{2,t_2})$ for all $t_1\in C_1,t_2\in C_2$.\\
Define $\Delta_{t_1}\phi_{t_1,t_2}:=(\hat\Delta_{t_1}\hat\phi_{t_1,t_2})[V(G_i)]$ as the set of isomorphisms restricted to the root bags for all $t_1\in C_1,t_2\in C_2$.\\
Define $\rho_1:V(G_2)\to\{1,\ldots,|V(G_2)|\}$ as some arbitrary labeling.\\
For each $t\in C_1\cup C_2$, we choose a representative $t^*\in C_2$ in the isomorphism class, i.e.,
it holds that $\CX_{i,t}\cong\CX_{1,t^*}$ are isomorphic
and that for two isomorphic $\CX_{1,t_1}\cong\CX_{2,t_2}$ it holds that $\CX_{1,t_1^*}=\CX_{1,t_2^*}$.\\ 
Define $J_i:=\{\Lambda_t\mid t\in C_i\}$ where $\Lambda_t:=\Delta_t\phi_{t,t^*}\rho_1$ for both $i=1,2$.
\begin{cs}
\case{the adhesion sets in $H_i$ are all equal for both $i=1,2$}~\\
Assign each representative $t^*\in C_2$ a different natural number $k(t^*)\in\NN$
and assign each $t\in C_i$ the number $m(t):=|\{t'\in C_i\mid t'^*=t^*$ and $\Lambda_t=\Lambda_{t'}\}|$.\\
Define a function $\alpha_i:J_i\to\NN$ that assigns each $\Lambda_t\in J_i$ a pair $(k(t^*),m(t))\in \NN^2$ for both $i=1,2$.\\
\comment{The number $k(t^*)$ encodes the isomorphism type of the subgraph corresponding to $t\in C_i$ and the number $m(t)$ encodes the multiplicity of
$\Lambda_t$ in its isomorphism class}\\
\comment{We claim that $\iso(G_1,J_1,\alpha_1;G_2,J_2,\alpha_2)=\iso(\hat\CX_1;\hat\CX_2)[V(G_1)]$. Let $\phi\in\iso(G_1,J_1,\alpha_1;G_2,J_2,\alpha_2)$.
Therefore, for each pair in $\Lambda_{t_1}\in J_1$ there is a pair $\Lambda_{t_2}\in J_2$
such that $\Lambda_{t_1}^\phi=\Lambda_{t_2}$.
Since $\alpha_1^\phi=\alpha_2$, it follows that $\CX_{1,t_1}\cong\CX_{2,t_2}$ and thus
$t_1^*=t_2^*$.
Therefore, $\phi^{-1}\Delta_{t_1}\phi_{t_1,t_1^*}\rho_1=\Lambda_{t_1}^\phi=\Lambda_{t_2}=\Delta_{t_2}\phi_{t_2,t_1^*}\rho_1$.
Equivalently, for all $\Lambda_{t_1}\in J_1$ there is a $\Lambda_{t_2}\in J_2$ such that $\phi\in\Lambda_{t_1}\Lambda_{t_2}^{-1}=\phi_{t_1,t_1^*}\Delta_{t_1^*}\phi_{t_2,t_1^*}^{-1}=\iso(\hat\CX_{1,t_1};\hat\CX_{2,t_2})[V(G_1)]$.
In other words, there is a function $\psi:J_1\to J_2$ such that for all $\Lambda_{t_1}\in J_1$
it holds that $\phi\in\Lambda_{t_1}\psi(\Lambda_{t_1})^{-1}$.
Moreover, $\psi:J_1\to J_2$ is injective (and thus bijective), which can be seen as follows.
Assume that $\psi(\Lambda_{t_1})=\psi(\Lambda_{t_1'})=\Lambda_{t_2}$. Then, $\phi\in\iso(\hat\CX_{1,t_1};\hat\CX_{2,t_2})[V(G_1)]\cap\iso(\hat\CX_{1,t_1'};\hat\CX_{2,t_2})[V(G_1)]$.
Therefore, $\iso(\hat\CX_{1,t_1};\hat\CX_{1,t_1'})[V(G_1)]$ contains the identity,
which means that $\Lambda_t,\Lambda_{t'}$ intersect non-trivially.
Since $\CX_{1,t},\CX_{1,t'}$ are isomorphic, $\Lambda_t,\Lambda_{t'}$ are cosets of the same group and thus $\Lambda_t=\Lambda_{t'}$,
which shows that $\psi$ is bijective.
Since $\psi(\Lambda_{t_1})=\Lambda_{t_2}$ implies $m(t_1)=m(t_2)$,
there is a bijective function $\tilde\psi:C_1\to C_2$ with $\Lambda_{\tilde\psi(t_1)}=\psi(\Lambda_{t_1})$.
Therefore, $\phi\in\iso(\hat\CX_1;\hat\CX_2)[V(G_1)]$}\\
Compute $\Delta\phi:=\iso(G_1,J_1,\alpha_1;G_2,J_2,\alpha_2)$ using \cref{cor:canObj}.\\
\comment{By the properties of the decomposition from \cref{theo:cliqueDecomposition}, it holds that $|V(G_1)|=|\beta_1(r_1)|\leq (\tw G_1^k)+1=(\tw G_1)+1\leq k+1$.
Moreover, $|J_1|\leq|V(T_1)|\in|V(\hat G_1)|^{\CO(1)}$.
For this reason, the algorithm from \cref{cor:canObj} runs in time $|V(\hat G_1)|^{\polylog(k)}$}

\case{the adhesion sets in $H_i$ are pairwise distinct for both $i=1,2$}~\\
Define a function $\alpha_i:H_i\to J_i$ that assigns each adhesion set $\beta_i(r_i)\cap\beta_i(t_i)\in H_i$
the coset $\Lambda_{t_i}\in J_i$ for both $i=1,2$.\\
\comment{Again, it holds that $\iso(G_1,J_1,\alpha_1;G_2,J_2,\alpha_2)=\iso(\hat\CX_1;\hat\CX_2)[{V(G_1)}]$}\\
Compute $\Delta\phi:=\iso(Y_1,H_1;Y_2,H_2)$ using \cref{lem:isoBasicClique}.\\
\comment{The lemma can be applied since \cref{theo:cliqueDecomposition} ensures
that the adhesion sets in $H_1,H_2$ are cliques in $Y_1,Y_2$, respectively.
By \cref{lem:improve}, it holds that $\sep Y_1,\tw Y_1\leq k$. For this reason, the algorithm from \cref{lem:isoBasicClique} runs in time $n^{\polylog(k)}$}\\
Compute $\Delta\phi:=\iso(V(G_1),H_1,J_1,\alpha_1;V(G_2),H_2,J_2,\alpha_2)\cap\Delta\phi$ using \cref{lem:cosetHyper}.\\
\comment{We have $|H_1|\leq|V(T_1)|\in|V(\hat G_1)|^{\CO(1)}$.
Furthermore, there is a point $v_1\in V(G_1)$ such that $\cw\Delta_{(v_1)}\leq\max(\sep Y_1,\tw Y_1\leq k)\leq k$.
For this reason, the algorithm runs in time $|V(\hat G_1)|^{\polylog(k)}$}\\
Compute $\Delta\phi:=\iso(G_1;G_2)\cap\Delta\phi$ using \cref{theo:boundedDegree}.\\
\comment{The algorithm from \cref{theo:boundedDegree} runs in time $|V(\hat G_1)|^{\polylog(k)}$}
\end{cs}
\comment{In both cases, we found the isomorphisms restricted to
the root bag, i.e., $\Delta\phi=\iso(\hat\CX_1;\hat\CX_2)[V(G_1)]$}\\
We define $\hat\Delta\hat\phi:=\iso(\hat\CX_1;\hat\CX_2)$.\\
\comment{This can be computed as follows.
We consider the homomorphism $h:\hat\Delta\to\Delta$ that maps $\hat\delta\in\hat\Delta$ to $\hat\delta[V(G_1)]\in\Delta$.
First, we explain how to compute the kernel $\hat K:=\ker(h)\leq\hat\Delta$.
The pointwise stabilizers $\hat\Theta_{t_1}:=(\hat{\Delta}_{t_1})_{(\beta_1(r)\cap\beta_1(t_1))}$ for $t_1\in C_1$ are polynomial-time computable.
Let $\hat\Theta_{t_1}'\leq\sym(V(\hat G_1))$ be the group that acts like $\hat\Theta_{t_1}$ on $\hat V(G_{i,t_1})$
and fixes all points in $V(\hat G_1)\setminus \hat V(G_{i,t_1})$.
Then, the kernel $\hat K\leq\sym(V(\hat G_1))$ is the group generated by all groups $\hat\Theta_{t_1}'\leq\sym(V(\hat G_1))$
for $t_1\in C_1$.
Next, we can compute a subgroup $\hat\Theta\leq\hat\Delta$ with $\hat\Theta[V(G_1)]=\Delta$
by extending each generator of $\Delta$ in an arbitrary way.
In the same way, we can compute an isomorphism $\hat\phi\in\hat\Delta\hat\phi$
with $\hat\phi[V(G_1)]=\phi$.
Finally, we define $\hat\Delta$ as the group generated by the groups $\hat K,\hat\Theta\leq\hat\Delta$}
\end{proof}

\section{Outlook and Open Questions}

One could ask the question whether our isomorphism algorithm
for graphs can be improved to a FPT-algorithm that runs in time $2^{\polylog(k)}n^{\CO(1)}$
where $n$ is the number of vertices and $k$ is the maximum treewidth of
the given graphs.
There are various reasons why this might be difficult.
One reason is that our approach would require a FPT-algorithm
for the isomorphism problem of graphs of maximum degree $d$ that runs in time $2^{\polylog(d)}n^{\CO(1)}$.
However, it is an open question whether any FPT-algorithm for the graph isomorphism problem parameterized by maximum degree
exists.
Another reason is that an algorithm for graphs running in time $2^{\polylog(k)}n^{\CO(1)}$
would imply an isomorphism algorithm for hypergraphs $(V,H)$ running in
time $2^{\polylog|V|}|H|^{\CO(1)}$.
It is also an open question, whether such a hypergraph isomorphism algorithm exists \cite{Babai2018GroupsG}.
If this were indeed the case, one could hope for an improvement of our canonization algorithm for a set $J$ consisting of labeling cosets that runs in time 
$2^{\polylog|V|}|J|^{\CO(1)}$.

Recently, Babai extended his quasipolynomial-time algorithm to the canonization problem for graphs \cite{DBLP:conf/stoc/Babai19}.
With Babai's result, it is a natural question
whether the bounded-degree
isomorphism algorithm of \cite{DBLP:conf/focs/GroheNS18}
extends to canonization as well.
The present isomorphism
algorithm for graphs parameterized by treewidth
should then be amenable to
canonization as well.

Another question that arises is about permutation groups $G\leq\sym(V)$.
The canonical labeling problem for permutation groups is of great interest because it also solves the normalizer problem.
In our recent work, we gave a canonization algorithm for explicitly given
permutation groups running in time $2^{\CO(|V|)}|G|^{\CO(1)}$ \cite{DBLP:conf/stoc/SchweitzerW19}.
Recently, the framework was extended to permutation groups
that are implicitly given and the running time was improved
to $2^{\CO(|V|)}$ \cite{DBLP:conf/soda/Wiebking20}.
The present work implies a canonization algorithm running in time $(|V|+|G|)^{\polylog|V|}$.
An important question is whether the present techniques can be combined with
the canonization techniques for implicitly given permutation groups to obtain a canonization algorithm running in time
$2^{\polylog|V|}$.

Finally, we ask whether the isomorphism problem
can be solved in time $n^{\polylog(|V(H)|)}$
where $n$ is the number of vertices
and $H$ is an excluded topological subgraph $H$ of the given graphs.
Even for excluded minors $H$, we do not have such an algorithm.

\appendix

\section{Proof of Lemma 9} 
\label{app}

\begin{proof}[Proof of \cref{lem:canSet0}]
Let $U:=V\cupdot V_1\cupdot\ldots\cupdot V_t$
where $V_i:=V\times\{i\}$.
The sets $V_i$ can be seen as disjoint copies of $V$.
Let $\Delta_i^\ord:=\rho_i^{-1}\Delta_i\rho_i$ for all $\Delta_i\rho_i\in J$
(this is well defined and does not depend on the representative $\rho_i$ of $\Delta_i\rho_i$).
We define a labeling coset $\Delta_U\rho_U\leq\lab(U)$.
Informally, the labeling coset $\Delta_U\rho_U$ orders the set of components $\{V_1,\ldots,V_t\}$ according to the ordering ``$\prec$'' defined on $\Delta_i^\ord$ for $i\in[t]$,
and it orders the component $V_i$ according to a labeling in $\Delta_i\rho_i$.
More formally, we define
$\Delta_U\rho_U:=\{\lambda_U\in\lab(U)\mid
\lambda_U|_V\in\lab(V)\text{ and }
\exists\lambda_1\in\Delta_1\rho_1\ldots\exists\lambda_t\in\Delta_t\rho_t\exists k_1,\ldots,k_t\in\NN\forall i,j\in[t],v\in V:
\Delta_i^\ord\prec\Delta_j^\ord\implies k_i<k_j\text{ and }
\lambda_U(v,i)=\lambda_i(v)+k_i\cdot|V|\}$.
We define a graph $G=(U,E)$ which identifies the vertices in $V$ with their corresponding copies
in $V_i$ for all $i\in[t]$.
More formally, we define
$E=\{(v,(v,i))\mid v\in V, (v,i)\in V_i,i\in[t]\}
\subseteq U^2$.
We compute $\Theta_U\tau_U:=\canGraph(E,\Delta_U\rho_U)$ using \cref{theo:canGraph}.
We claim that the labeling coset $\Delta\rho:=(\Theta_U\tau_U)[V]$ induced on $V$ defines a canonical labeling for $J$.

\begin{cl1}
Assume we have $J^\phi$
instead of $J$ as an input.
We have to show that the algorithm outputs $\phi^{-1}\Delta\rho$ instead of $\Delta\rho$.
The group $\Delta_i^\ord$ does not depend on $\phi$
since $\rho_i^{-1}\Delta_i\rho_i=(\phi^{-1}\rho_i)^{-1}\phi^{-1}\Delta_i\rho_i$.
By construction, we obtain $\phi_U^{-1}\Delta_U\rho_U,E^{\phi_U}$
instead of $\Delta_U\rho_U,E$ where $\phi_U$ is a bijection with $\phi_U|_V=\phi$.
By (CL1) of $\canGraph$, we obtain $\phi^{-1}\Theta_U\tau_U$ instead of $\Theta_U\tau_U$.
Finally, we obtain $(\phi_U^{-1}\Theta_U\tau_U)[V^\phi]=\phi^{-1}\Delta\rho$
instead of $\Delta\rho$.
\end{cl1}

\begin{cl2}
We have to show that $\Delta=(\aut(E)\cap\Delta_U)|_V=\aut(J)$.
The inclusion $\aut(J)\subseteq\Delta$
follows from (CL1) of this reduction. We thus need to show the reversed inclusion
$(\aut(E)\cap\Delta_U)|_V\subseteq\aut(J)$.
So let $\sigma_U\in\aut(E)\cap\Delta_U$.
Since $\sigma_U\in\Delta_U$,
it follows that there are $\lambda_1,\ldots,\lambda_t,\lambda'_1,\ldots,\lambda'_t$
and $k_1,\ldots,k_t,k'_1,\ldots,k'_t\in\NN$
such that for all $i\in[t],(v,i)\in V_i$
it holds
that $\sigma_U(v,i)=({\lambda'_j}^{-1}(\lambda_i(v)+k_i\cdot |V|-k_j'\cdot |V|),j)$ for some $j\in[t]$.
It must hold that $k_i=k_j'$
and therefore $\sigma_U(v,i)=({\lambda'_j}^{-1}(\lambda_i(v)),j)$.
In particular, $j\in[t]$ only depends on the choice of $i\in[t]$ (and not on the choice of $v\in V_i$).
Therefore, there is a $\psi\in\sym(t)$ such that
for all $(v,i)\in U$ it holds that
$\sigma_U(v,i)=(w,\psi(i))$ for some $w\in V$.
Since $\sigma_U\in\aut(E)$, it follows that
for all $i\in[t]$ there are $\lambda_i\in\Delta_i\rho_i=\rho_i\Delta_i^\ord$
and $\lambda_{\psi(i)}\in\Delta_{\psi(i)}\rho_{\psi(i)}=\rho_{\psi(i)}\Delta_{\psi(i)}^\ord$
such that
$\sigma_U|_V=\lambda_i\lambda_{\psi(i)}^{-1}$.
Since $\Delta_i^\ord=\Delta_{\psi(i)}^\ord$ this is equivalent to
$(\sigma_U|_V)^{-1}\Delta_i\rho_i=\Delta_{\psi(i)}\rho_{\psi(i)}$.
This implies $\sigma_U|_V\in\aut(J)$.
\end{cl2}
\end{proof}

\bibliographystyle{alpha}
\bibliography{references}

\newcommand{\etalchar}[1]{$^{#1}$}
\begin{thebibliography}{AGvM{\etalchar{+}}18}

\bibitem[AGvM{\etalchar{+}}18]{DBLP:journals/siamcomp/AllenderGMMM18}
Eric Allender, Joshua~A. Grochow, Dieter van Melkebeek, Cristopher Moore, and
  Andrew Morgan.
\newblock Minimum circuit size, graph isomorphism, and related problems.
\newblock {\em {SIAM} J. Comput.}, 47(4):1339--1372, 2018.

\bibitem[Ami01]{DBLP:conf/uai/Amir01}
Eyal Amir.
\newblock Efficient approximation for triangulation of minimum treewidth.
\newblock In {\em {UAI} '01: Proceedings of the 17th Conference in Uncertainty
  in Artificial Intelligence, University of Washington, Seattle, Washington,
  USA, August 2-5, 2001}, pages 7--15, 2001.

\bibitem[Bab15]{DBLP:journals/corr/Babai15}
L{\'{a}}szl{\'{o}} Babai.
\newblock Graph isomorphism in quasipolynomial time.
\newblock {\em CoRR}, abs/1512.03547, 2015.

\bibitem[Bab16]{DBLP:conf/stoc/Babai16}
L{\'{a}}szl{\'{o}} Babai.
\newblock Graph isomorphism in quasipolynomial time [extended abstract].
\newblock In Daniel Wichs and Yishay Mansour, editors, {\em Proceedings of the
  48th Annual {ACM} {SIGACT} Symposium on Theory of Computing, {STOC} 2016,
  Cambridge, MA, USA, June 18-21, 2016}, pages 684--697. {ACM}, 2016.

\bibitem[Bab18]{Babai2018GroupsG}
L{\'a}szl{\'o} Babai.
\newblock Groups, graphs, algorithms: The graph isomorphism problem.
\newblock In {\em Proc. ICM}, pages 3303--3320, 2018.

\bibitem[Bab19]{DBLP:conf/stoc/Babai19}
L{\'{a}}szl{\'{o}} Babai.
\newblock Canonical form for graphs in quasipolynomial time: preliminary
  report.
\newblock In {\em Proceedings of the 51st Annual {ACM} {SIGACT} Symposium on
  Theory of Computing, {STOC} 2019, Phoenix, AZ, USA, June 23-26, 2019.}, pages
  1237--1246, 2019.

\bibitem[BCC{\etalchar{+}}06]{bodlaender2006open}
Hans~L. Bodlaender, Leizhen Cai, Jianer Chen, Michael~R. Fellows, Jan~Arne
  Telle, and D{\'a}niel Marx.
\newblock Open problems in parameterized and exact computation-iwpec 2006.
\newblock {\em UU-CS}, 2006, 2006.

\bibitem[BDK12]{DBLP:conf/iwpec/BoulandDK12}
Adam Bouland, Anuj Dawar, and Eryk Kopczynski.
\newblock On tractable parameterizations of graph isomorphism.
\newblock In {\em Parameterized and Exact Computation - 7th International
  Symposium, {IPEC} 2012, Ljubljana, Slovenia, September 12-14, 2012.
  Proceedings}, pages 218--230, 2012.

\bibitem[BL83]{DBLP:conf/stoc/BabaiL83}
L{\'{a}}szl{\'{o}} Babai and Eugene~M. Luks.
\newblock Canonical labeling of graphs.
\newblock In {\em Proceedings of the 15th Annual {ACM} Symposium on Theory of
  Computing, 25-27 April, 1983, Boston, Massachusetts, {USA}}, pages 171--183,
  1983.

\bibitem[Bod90]{DBLP:journals/jal/Bodlaender90}
Hans~L. Bodlaender.
\newblock Polynomial algorithms for graph isomorphism and chromatic index on
  partial k-trees.
\newblock {\em J. Algorithms}, 11(4):631--643, 1990.

\bibitem[DF13]{DBLP:series/txcs/DowneyF13}
Rodney~G. Downey and Michael~R. Fellows.
\newblock {\em Fundamentals of Parameterized Complexity}.
\newblock Texts in Computer Science. Springer, 2013.

\bibitem[Die12]{DBLP:books/daglib/0030488}
Reinhard Diestel.
\newblock {\em Graph Theory, 4th Edition}, volume 173 of {\em Graduate texts in
  mathematics}.
\newblock Springer, 2012.

\bibitem[ES16]{DBLP:conf/stacs/ElberfeldS16}
Michael Elberfeld and Pascal Schweitzer.
\newblock Canonizing graphs of bounded tree width in logspace.
\newblock In {\em {STACS}}, volume~47 of {\em LIPIcs}, pages 32:1--32:14.
  Schloss Dagstuhl - Leibniz-Zentrum fuer Informatik, 2016.

\bibitem[FM80]{DBLP:conf/stoc/FilottiM80}
I.~S. Filotti and Jack~N. Mayer.
\newblock A polynomial-time algorithm for determining the isomorphism of graphs
  of fixed genus (working paper).
\newblock In {\em Proceedings of the 12th Annual {ACM} Symposium on Theory of
  Computing, April 28-30, 1980, Los Angeles, California, {USA}}, pages
  236--243, 1980.

\bibitem[GM15]{DBLP:journals/siamcomp/GroheM15}
Martin Grohe and D{\'{a}}niel Marx.
\newblock Structure theorem and isomorphism test for graphs with excluded
  topological subgraphs.
\newblock {\em {SIAM} J. Comput.}, 44(1):114--159, 2015.

\bibitem[GNS18]{DBLP:conf/focs/GroheNS18}
Martin Grohe, Daniel Neuen, and Pascal Schweitzer.
\newblock A faster isomorphism test for graphs of small degree.
\newblock In {\em 59th {IEEE} Annual Symposium on Foundations of Computer
  Science, {FOCS} 2018, Paris, France, October 7-9, 2018}, pages 89--100, 2018.

\bibitem[GNSW18]{DBLP:conf/icalp/GroheNSW18}
Martin Grohe, Daniel Neuen, Pascal Schweitzer, and Daniel Wiebking.
\newblock An improved isomorphism test for bounded-tree-width graphs.
\newblock In {\em 45th International Colloquium on Automata, Languages, and
  Programming, {ICALP} 2018, July 9-13, 2018, Prague, Czech Republic}, pages
  67:1--67:14, 2018.

\bibitem[Gro00]{DBLP:conf/stoc/Grohe00}
Martin Grohe.
\newblock Isomorphism testing for embeddable graphs through definability.
\newblock In {\em Proceedings of the Thirty-Second Annual {ACM} Symposium on
  Theory of Computing, May 21-23, 2000, Portland, OR, {USA}}, pages 63--72,
  2000.

\bibitem[Hel17]{helfgott2017isomorphismes}
Harald~Andrés Helfgott.
\newblock Isomorphismes de graphes en temps quasi-polynomial (d'après babai et
  luks, weisfeiler-leman...), 2017.

\bibitem[HT71]{DBLP:journals/ipl/HopcroftT71}
John~E. Hopcroft and Robert~Endre Tarjan.
\newblock A v{\({^2}\)} algorithm for determining isomorphism of planar graphs.
\newblock {\em Inf. Process. Lett.}, 1(1):32--34, 1971.

\bibitem[Kaw15]{DBLP:journals/corr/Kawarabayashi15a}
Ken{-}ichi Kawarabayashi.
\newblock Graph isomorphism for bounded genus graphs in linear time.
\newblock {\em CoRR}, abs/1511.02460, 2015.

\bibitem[KM08]{DBLP:conf/stoc/KawarabayashiM08}
Ken{-}ichi Kawarabayashi and Bojan Mohar.
\newblock Graph and map isomorphism and all polyhedral embeddings in linear
  time.
\newblock In {\em Proceedings of the 40th Annual {ACM} Symposium on Theory of
  Computing, Victoria, British Columbia, Canada, May 17-20, 2008}, pages
  471--480, 2008.

\bibitem[KS10]{DBLP:conf/swat/KratschS10}
Stefan Kratsch and Pascal Schweitzer.
\newblock Isomorphism for graphs of bounded feedback vertex set number.
\newblock In {\em Algorithm Theory - {SWAT} 2010, 12th Scandinavian Symposium
  and Workshops on Algorithm Theory, Bergen, Norway, June 21-23, 2010.
  Proceedings}, pages 81--92, 2010.

\bibitem[Lei93]{DBLP:journals/dm/Leimer93}
Hanns{-}Georg Leimer.
\newblock Optimal decomposition by clique separators.
\newblock {\em Discrete Mathematics}, 113(1-3):99--123, 1993.

\bibitem[LPPS17]{DBLP:journals/siamcomp/LokshtanovPPS17}
Daniel Lokshtanov, Marcin Pilipczuk, Michal Pilipczuk, and Saket Saurabh.
\newblock Fixed-parameter tractable canonization and isomorphism test for
  graphs of bounded treewidth.
\newblock {\em {SIAM} J. Comput.}, 46(1):161--189, 2017.

\bibitem[Luk82]{DBLP:journals/jcss/Luks82}
Eugene~M. Luks.
\newblock Isomorphism of graphs of bounded valence can be tested in polynomial
  time.
\newblock {\em J. Comput. Syst. Sci.}, 25(1):42--65, 1982.

\bibitem[Mil79]{DBLP:journals/jcss/Miller79}
Gary~L. Miller.
\newblock Graph isomorphism, general remarks.
\newblock {\em J. Comput. Syst. Sci.}, 18(2):128--142, 1979.

\bibitem[Mil80]{DBLP:conf/stoc/Miller80}
Gary~L. Miller.
\newblock Isomorphism testing for graphs of bounded genus.
\newblock In {\em Proceedings of the 12th Annual {ACM} Symposium on Theory of
  Computing, April 28-30, 1980, Los Angeles, California, {USA}}, pages
  225--235, 1980.

\bibitem[Mil83]{DBLP:journals/iandc/Miller83a}
Gary~L. Miller.
\newblock Isomorphism of k-contractible graphs. {A} generalization of bounded
  valence and bounded genus.
\newblock {\em Information and Control}, 56(1/2):1--20, 1983.

\bibitem[MK11]{DBLP:journals/jcss/MyrvoldK11}
Wendy Myrvold and William~L. Kocay.
\newblock Errors in graph embedding algorithms.
\newblock {\em J. Comput. Syst. Sci.}, 77(2):430--438, 2011.

\bibitem[Neu19]{neuenDiss}
Daniel Neuen.
\newblock {\em The Power of Algorithmic Approaches to the Graph Isomorphism
  Problem}.
\newblock PhD thesis, {RWTH} Aachen University, Aachen, Germany, 2019.

\bibitem[Neu20]{NeuenHyper}
Daniel Neuen.
\newblock Hypergraph isomorphism for groups with restricted composition
  factors.
\newblock In {\em 47th International Colloquium on Automata, Languages, and
  Programming, {ICALP} 2020, July 8-11, 2020, Saarbrücken, Germany (virtual
  conference)}, 2020.
\newblock To appear.

\bibitem[Ota12]{DBLP:conf/isaac/Otachi12}
Yota Otachi.
\newblock Isomorphism for graphs of bounded connected-path-distance-width.
\newblock In {\em Algorithms and Computation - 23rd International Symposium,
  {ISAAC} 2012, Taipei, Taiwan, December 19-21, 2012. Proceedings}, pages
  455--464, 2012.

\bibitem[Pon91]{Ponomarenko1991}
I.~N. Ponomarenko.
\newblock The isomorphism problem for classes of graphs closed under
  contraction.
\newblock {\em Journal of Soviet Mathematics}, 55(2):1621--1643, Jun 1991.

\bibitem[RS83]{DBLP:journals/jct/RobertsonS83}
Neil Robertson and Paul~D. Seymour.
\newblock Graph minors. i. excluding a forest.
\newblock {\em J. Comb. Theory, Ser. {B}}, 35(1):39--61, 1983.

\bibitem[Sch88]{DBLP:journals/jcss/Schoning88}
Uwe Schöning.
\newblock Graph isomorphism is in the low hierarchy.
\newblock {\em J. Comput. Syst. Sci.}, 37(3):312--323, 1988.

\bibitem[Ser03]{seress}
{\'A}kos Seress.
\newblock {\em Permutation group algorithms}, volume 152 of {\em Cambridge
  Tracts in Mathematics}.
\newblock Cambridge University Press, Cambridge, 2003.

\bibitem[SW19]{DBLP:conf/stoc/SchweitzerW19}
Pascal Schweitzer and Daniel Wiebking.
\newblock A unifying method for the design of algorithms canonizing
  combinatorial objects.
\newblock In {\em Proceedings of the 51st Annual {ACM} {SIGACT} Symposium on
  Theory of Computing, {STOC} 2019, Phoenix, AZ, USA, June 23-26, 2019.}, pages
  1247--1258, 2019.

\bibitem[Wei66]{weinberg1966simple}
Louis Weinberg.
\newblock A simple and efficient algorithm for determining isomorphism of
  planar triply connected graphs.
\newblock {\em IEEE Transactions on Circuit Theory}, 13(2):142--148, 1966.

\bibitem[Wie20]{DBLP:conf/soda/Wiebking20}
Daniel Wiebking.
\newblock Normalizes and permutational isomorphisms in simply-exponential time.
\newblock In {\em Proceedings of the 2020 {ACM-SIAM} Symposium on Discrete
  Algorithms, {SODA} 2020, Salt Lake City, UT, USA, January 5-8, 2020}, pages
  230--238, 2020.

\bibitem[YBdFT99]{DBLP:journals/algorithmica/YamazakiBFT99}
Koichi Yamazaki, Hans~L. Bodlaender, Babette de~Fluiter, and Dimitrios~M.
  Thilikos.
\newblock Isomorphism for graphs of bounded distance width.
\newblock {\em Algorithmica}, 24(2):105--127, 1999.

\bibitem[Zem70]{zemlyachenko1970canonical}
Viktor~N. Zemlyachenko.
\newblock Canonical numbering of trees.
\newblock In {\em Proc. Seminar on Comb. Anal. at Moscow State University},
  page~55, 1970.

\bibitem[ZKT85]{zemlyachenko1985graph}
Viktor~N. Zemlyachenko, Nickolay~M. Korneenko, and Regina~I. Tyshkevich.
\newblock Graph isomorphism problem.
\newblock {\em Journal of Soviet Mathematics}, 29(4):1426--1481, 1985.

\end{thebibliography}

\addcontentsline{toc}{section}{Bibliography}

\end{document}